\documentclass[12pt,a4paper]{article}

\usepackage[hidelinks]{hyperref}

\usepackage{afterpage, amsfonts, amsmath, amssymb, amsthm, bm, bbm, caption, color, enumitem, float, graphicx, epstopdf, lipsum, longtable, lscape, mathrsfs, multirow, nccmath, nonfloat, pdflscape, rotating, tocloft}

\usepackage[margin=0.65in]{geometry}

\usepackage[round]{natbib}
\usepackage[doublespacing]{setspace}
\usepackage[table]{xcolor}

\def\be{\begin{equation}}
\def\ee{\end{equation}}
\def\bea{\begin{eqnarray}}
\def\eea{\end{eqnarray}}

\author{}
\title{}

\DeclareMathOperator*{\argmin}{\arg\!\min}
\DeclareMathOperator*{\plim}{p\!\lim}

\begin{document}
\newcommand\blfootnote[1]{
\begingroup
\renewcommand\thefootnote{}\footnote{#1}
\addtocounter{footnote}{-1}
\endgroup
}

\newtheorem{corollary}{Corollary}
\newtheorem{definition}{Definition}
\newtheorem{lemma}{Lemma}
\newtheorem{proposition}{Proposition}
\newtheorem{remark}{Remark}
\newtheorem{theorem}{Theorem}
\newtheorem{assumption}{Assumption}
\newtheorem{example}{Example}

\numberwithin{corollary}{section}
\numberwithin{definition}{section}
\numberwithin{lemma}{section}
\numberwithin{proposition}{section}
\numberwithin{remark}{section}
\numberwithin{theorem}{section}

\allowdisplaybreaks[4]

\begin{titlepage}

\begin{center}
{\large \textbf{A Class of Time--Varying Vector Moving Average ($\infty$) Models:\\ Nonparametric Kernel Estimation and Application}}\blfootnote{{\it Corresponding author}: Jiti Gao, Department of Econometrics and Business Statistics, Monash University, Caulfield East, Victoria 3145, Australia. Email: \url{Jiti.Gao@monash.edu}.

The authors of this paper would like to thank George Athanasopoulos, Rainer Dahlhaus, David Frazier, Oliver Linton, Gael Martin, Peter CB Phillips and Wei Biao Wu for their constructive comments on earlier versions of this paper. The second author would also like to acknowledge financial support from the Australian Research Council Discovery Grants Program under Grant Numbers: DP170104421 and DP200102769.}

\bigskip

{\sc Yayi Yan, Jiti Gao and Bin Peng
\smallskip}

Monash University

\bigskip

\today

\end{center}

\begin{abstract}

Multivariate dynamic time series models are widely encountered in practical studies, e.g., modelling policy transmission mechanism and measuring connectedness between economic agents. To better capture the dynamics, this paper proposes a wide class of multivariate dynamic models with time--varying coefficients, which have a general time--varying vector moving average (VMA) representation, and nest, for instance, time--varying vector autoregression (VAR), time--varying  vector autoregression moving--average (VARMA), and so forth as special cases. The paper then develops a unified estimation method for the unknown quantities before an asymptotic theory for the proposed estimators is established. In the empirical study, we investigate the transmission mechanism of monetary policy using U.S. data, and uncover a fall in the volatilities of exogenous shocks. In addition, we find that (i) monetary policy shocks have less influence on inflation before and during the so--called Great Moderation, (ii) inflation is more anchored recently, and (iii) the long--run level of inflation is below, but quite close to the Federal Reserve's target of two percent after the beginning of the Great Moderation period.

\bigskip


\end{abstract}
\end{titlepage}

\section{Introduction}\label{Section1}
\renewcommand{\theequation}{1.\arabic{equation}}

\setcounter{equation}{0}

Vector autoregressions (VARs), as well as their extensions like vector autoregressive moving average (VARMA)
models and VARs with exogenous variables (VARX), are among some of the most popular frameworks for modelling dynamic interactions among multiple variables. These models arise mainly as a response to the ``incredible'' identification conditions embedded in the large--scale macroeconomic models \cite[]{sims1980}. VAR modelling begins with minimal restrictions on the multivariate dynamic models. Gradually armed with identification information, VARs plus their statistical tool--kits like impulse response functions, are powerful approaches for conducting policy analysis. Also, VARs can be applied to other important tasks including data description and forecasting, see \cite{stock2001vector} for a detailed review. Despite the popularity, linear VAR models can always be rejected by data in empirical studies \cite[cf.,][]{tsay1998testing}. For example, \cite{stock2016dynamic} point out, ``changes associated with the Great Moderation go beyond reduction in variances to include changes in dynamics and reduction in predictability.''

To go beyond linear VAR models, various parametric time--varying VAR models have been proposed (e.g., \citealp{tsay1998testing}, \citealp{sims2006}, and references therein) in order to allow for abrupt structural breaks in economic relationships and obtain efficient estimation. However, model misspecification and parameter instability may undermine the performance of parametric time--varying VAR models. Usually, researchers do not know the true functional forms of the time--varying parameters, so the choice on the functional forms might be somewhat arbitrary. In addition, as pointed out by \cite{hansen2001new}, ``it may seem unlikely that a structural break could be immediate and might seem more reasonable to allow a structural change to take a period of time to take effect''. Therefore, it is reasonable to allow smooth structural changes over a period of time rather than abrupt structural changes. Another strand of the VAR literature assumes that structural coefficients evolve in a random way, such as \cite{primiceri2005time} and \cite{giannone2019priors}. Recently, \cite{giraitis2014inference} point out that ``the theoretical asymptotic properties of estimating such processes via the Kalman, or related filters are unclear''. Along this line, \cite{giraitis2014inference} and \cite{giraitis2018inference} have achieved some useful asymptotic results. However, estimation theory for time--varying impulse response functions, which are of interest in interpreting multivariate dynamic models, are not yet established in the random walk case.

It is worth pointing out that although nonparametric estimation for deterministic time--varying models has received much attention initially on time series regression models (\citealp{robinson1989nonparametric, cai07, lcg11, cgl12, zhang2012inference}) over the past three decades and then on univariate autoregressive models (\citealp{dahlhaus2006statistical, richter2019cross}) in recent years, little study about multivariate autoregressive models with deterministic time--varying coefficients has been done. One exception is \cite{dahlhaus1999nonlinear}, who use a wavelet method to transform the time--varying VAR model to a linear approximation with an orthonormal wavelet basis, and then show that the wavelet estimator attains the usual near--optimal minimax rate of $L_2$ convergence. 
Up to this point, it is worth bringing up the terminology ``local stationarity", which dates back to the seminal work by \cite{dahlhaus1996kullback} at least. While several studies have been conducted along this line (\citealp{dahlhaus1996kullback} and \citealp{zhang2012inference} on time--varying AR;  \citealp{dahlhaus2009empirical} on time--varying ARMA; \citealp{rohan2013nonparametric} and \citealp{truquet2017parameter} on time--varying ARCH/GARCH), the literature has not ventured much outside the univariate setting. A commonly used method is to approximate a locally stationary process by a stationary approximation on each of the segments \cite[]{dahlhaus2019towards}. However, it remains unclear to us how to extend this approximation method for the univariate setting to the multivariate case where the segments on which stationarity approximations for each univariate time series may be quite different.

This paper is to show the versatility of an alternative approach that is especially designed for a wide class of time--varying VMA$(\infty)$ processes. Our approach relies on an explicit decomposition of time--varying VMA$(\infty)$ processes into long--run and transitory elements, which is known as the Beveridge--Nelson (BN) decomposition (\citealp{beveridge1981new,phillips1992asymptotics}). The long--run component in the decomposition yields a martingale approximation to the sum of time--varying VMA$(\infty)$ processes. We are then able to deal with a wide class of multivariate dynamic models with smooth time--varying coefficients, which have a general time--varying VMA$(\infty)$ representation, nesting VAR, VARMA, VARX and so forth as special cases. Specifically, the structural coefficients are unknown functions of the re--scaled time, so that the proposed models can better capture the simultaneous relations among variables of interest over time. Such a modelling strategy is especially useful for analysing time series over a long horizon, since it offers a comprehensive treatment on tracking interests which are affected by frequently updated policies, environment, system, etc. In an economy system consisting of inflation, unemployment and interest rates, one priority in Section \ref{Section4} is inferring time--varying impacts of the interest rate change, which helps stabilize fluctuations in inflation and unemployment in long--run. Under the proposed framework, it is achieved by investigating the corresponding time--varying impulse response functions.

In summary, our contributions are in threefold. First, we propose a wide class of time--varying VMA$(\infty)$ models, which covers several classes of multivariate dynamic models. Second, we develop a time--varying counterpart of the conventional BN decomposition and propose a unified estimation method for a class of unknown time--varying functions. We then establish the corresponding asymptotic theory for the proposed models and estimators. Third, in the empirical study of Section \ref{Section4}, we study the changing dynamics of three key U.S. macroeconomic variables (i.e., inflation, unemployment, and interest rate), and uncover a fall in the volatilities of exogenous shocks. In addition, we find that (i) monetary policy shocks have less influence on inflation before and during the so--called Great Moderation; (ii) inflation is more anchored recently; and (iii) the long--run level of inflation is below, but quite close to the Federal Reserve's target of two percent after the beginning of the Great Moderation period.

The organization of this paper is as follows. Section \ref{Section2} proposes a class of time--varying VMA($\infty$) models. In Section \ref{Section3}, we discuss a class of time--varying VAR models and then establish an estimation theory for the unknown quantities. Section \ref{Section4} presents an empirical study to show the practical relevance and applicability of the proposed models and estimation theory. Section \ref{Section5} gives a short conclusion. The main proofs of the theorems are given in Appendix A. In the online supplementary material, simulation results are given in Appendix B.1 and some technical lemmas and proofs are given in the rest of Appendix B.

Before proceeding further, it is convenient to introduce some notation: $\| \cdot \|$ denotes the Euclidean norm of a vector or the Frobenius norm of a matrix; $\otimes$ denotes the Kronecker product; $\bm{I}_a$ and $\bm{0}_a$ are $a\times a$ identity and null matrices respectively, and $\bm{0}_{a\times b}$ stands for a $a\times b$ matrix of zeros; for a function $g(w)$, let $g^{(j)}(w)$ be the $j^{th}$ derivative of $g(w)$, where $j\ge 0$ and $g^{(0)}(w) \equiv g(w)$; $K_h(\cdot) =K(\cdot/h)/h$, where $K(\cdot)$ and $h$ stand for a nonparametric kernel function and a bandwidth respectively; let $\tilde{c}_k =\int_{-1}^{1} u^k K(u) du$ and $\tilde{v}_k= \int_{-1}^{1} u^k K^2(u) du$ for integer $k\ge 0$; $\mathrm{vec}(\cdot)$ stacks the elements of an $m\times n$ matrix as an $mn \times 1$ vector; for any $a\times a$ square matrix $\bm{A}$, $\mathrm{vech}\left(\bm{A}\right)$ denotes the $\frac{1}{2}a(a+1)\times 1$ vector obtained from $\mathrm{vec}\left(\bm{A}\right)$ by eliminating all supra--diagonal elements of $\bm{A}$; $\mathrm{tr}\left(\bm{A}\right)$ denotes the trace of $\bm{A}$; finally, let $\to_P$ and $\to_D$ denote convergence in probability and convergence in distribution, respectively.

\section{Structure of Time--Varying VMA\texorpdfstring{$(\infty)$}{}}\label{Section2}
\renewcommand{\theequation}{2.\arabic{equation}}

\setcounter{equation}{0}

We start our investigation by considering a class of time--varying VMA$(\infty)$ model:

\begin{equation}\label{eq1}
\bm{x}_t=\bm{\mu}_t+\sum_{j=0}^{\infty}\bm{B}_{j,t}\bm{\epsilon}_{t-j} := \bm{\mu}_t+ \bm{\mathbb{B}}_t(L)\bm{\epsilon}_t
\end{equation}
for $t=1,\ldots,T$, where $\bm{x}_t$ is a $d$--dimensional vector of observable variables, $\bm{\mu}_t$ is a $d$--dimensional unknown trending function, $\bm{\epsilon}_t$ is a vector of $d$--dimensional random innovations, and $d$ is fixed throughout this paper. Moreover, $\bm{\mathbb{B}}_t(L)=\sum_{j=0}^{\infty}\bm{B}_{j,t}L^j$, where $L$ is the lag operator, and $\bm{B}_{j,t}$ is a matrix of $d \times d$ unknown deterministic coefficients.

We first comment on the usefulness of the structure in \eqref{eq1}, and the corresponding BN decomposition. An application of the BN decomposition gives
\begin{eqnarray}
\bm{x}_t=\bm{\mu}_t+\bm{\mathbb{B}}_t(1)\bm{\epsilon}_t+\widetilde{\mathbb{B}}_t(L)\bm{\epsilon}_{t-1}-\widetilde{\mathbb{B}}_t(L)\bm{\epsilon}_t, \label{eq2}
\end{eqnarray}
where we have used the decomposition of $\bm{\mathbb{B}}_t(L)$ as $\bm{\mathbb{B}}_t(L)= \bm{\mathbb{B}}_t(1) -(1-L)\widetilde{\mathbb{B}}_t(L)$, in which $\widetilde{\mathbb{B}}_t(L) = \sum_{j=0}^{\infty}\bm{\widetilde{B}}_{j,t}L^j$ and $\bm{\widetilde{B}}_{j,t}=\sum_{k=j+1}^{\infty}\bm{B}_{k,t}$. Equation (\ref{eq2}) indicates that one may establish some general asymptotic properties for partial sums and quadratic forms of $\bm{x}_t$ with minor restrictions on $\{\bm{\epsilon}_t\}$. For example, one can show that the simple average of $\bm{x}_t$ becomes
\begin{eqnarray}
\frac{1}{T} \sum_{t=1}^T \bm{x}_t = \frac{1}{T} \sum_{t=1}^T \bm{\mu}_t + \frac{1}{T} \sum_{t=1}^T \bm{\mathbb{B}}_t(1)\bm{\epsilon}_t + O_P(T^{-1}).
\label{eq3}
\end{eqnarray}
Similar to \eqref{eq3}, asymptotic properties for partial sums and quadratic forms mainly depend on the probabilistic structure of $\{\bm{\epsilon}_t\}$ and regularity conditions on $\{\bm{B}_{j,t}\}$. In other words, there is no need to impose any further structure on $\bm{x}_t$, such as requiring $\bm{x}_t$ to be locally stationary time series in a similar way to what has been done in the relevant literature for the univariate case. As a consequence, it facilitates to develop general theory for the multivariate case. Moreover, it should be added that our settings in (\ref{eq2}) and (\ref{eq3}) considerably extend similar treatments by \cite{phillips1992asymptotics} for the univariate linear process case where both $\bm{\mu}_t$ and $\bm{B}_{j,t}$ reduce to constant scalars: $\bm{\mu}_t = \mu$ and $\bm{B}_{j,t} = B_j$.

\subsection{Examples and Useful Properties}\label{Section2.1}

Let us now stress that \eqref{eq1} covers a wide range of models, which are of general interest in both theory and practice. Below, we list a few examples, of which the parametric counterparts can be seen in \cite{lutkepohl2005new}.

\begin{example} \label{Example1}
\normalfont

Suppose that $\bm{x}_t$ is a $d$--dimensional time--varying VAR$(p)$ process:
\begin{equation}\label{eq4}
\bm{x}_t=\bm{A}_{1,t}\bm{x}_{t-1}+\cdots +\bm{A}_{p,t}\bm{x}_{t-p}+\bm{\epsilon}_t,
\end{equation}
which has been widely studied in the literature with Bayesian framework being the dominant approach (e.g., \citealp{benati2009var, paul2019time}). Similar to \citet[p. 260]{hamilton1994time}, \eqref{eq4} can be expressed as a time--varying MA($\infty$) process $\bm{x}_t=\sum_{j=0}^{\infty}\bm{B}_{j,t}\bm{\epsilon}_{t-j}$, where $\bm{B}_{0,t}=\bm{I}_d$ and $\bm{B}_{j,t}=\bm{J}\prod_{i=0}^{j-1}\bm{\Phi}_{t-i}\bm{J}^\top$ for $j\geq 1$ with $\bm{J}=[\bm{I}_d,\bm{0}_{d\times d(p-1)}]$ and
\begin{eqnarray*}
\bm{\Phi}_t=\left(\begin{matrix}
\bm{A}_{1,t} & \cdots & \bm{A}_{p-1,t} & \bm{A}_{p,t}  \\
\bm{I}_d & \cdots& \bm{0}_{d} & \bm{0}_{d}\\
\vdots & \ddots&\vdots & \vdots\\
\bm{0}_{d} &\cdots &\bm{I}_d &\bm{0}_{d}\\
\end{matrix} \right).
\end{eqnarray*}
\end{example}

\begin{example} \label{Example2}
\normalfont

Suppose that $\bm{x}_t$ is a $d$--dimensional time--varying VARMA$(p,q)$ process as follows:
\begin{equation}\label{eq5}
\bm{x}_t=\bm{A}_{1,t}\bm{x}_{t-1}+...+\bm{A}_{p,t}\bm{x}_{t-p}+\bm{\epsilon}_t+\bm{\Theta}_{1,t}\bm{\epsilon}_{t-1}+...+\bm{\Theta}_{q,t}\bm{\epsilon}_{t-q},
\end{equation}
which then can be expressed as $\bm{x}_t =\sum_{b=0}^{\infty}\bm{D}_{b,t}\bm{\epsilon}_{t-b}$ with $\bm{D}_{b,t}=\sum_{j=\mathrm{max}(0,b-q)}^{b} \bm{B}_{j,t}\bm{\Theta}_{b-j,t-j}$, $\bm{B}_{j,t}$ defined similarly as in Example \ref{Example1}, and $\bm{\Theta}_{0,t}\equiv \bm{I}_d$ independent of $t$.
\end{example}

\begin{example} \label{Example3}
\normalfont

Let $\bm{x}_t$ be a $d$--dimensional time--varying double MA$(\infty)$ process:
\begin{equation}\label{eq6}
\bm{x}_t=\sum_{j=0}^{\infty}\bm{\Psi}_{j,t}\bm{v}_{t-j}\quad \text{and}\quad \bm{v}_t=\sum_{l=0}^{\infty}\bm{\Theta}_{l,t}\bm{\epsilon}_{t-l},
\end{equation}
in which the innovations $\bm{v}_{t}$'s also follow a time--varying MA$(\infty)$ process. Simple algebra shows that $\bm{x}_t =\sum_{j=0}^{\infty} \bm{B}_{j,t} \bm{\epsilon}_{t-j}$, where $\bm{B}_{j,t}=\sum_{l=0}^{j}\bm{\Psi}_{l,t}\bm{\Theta}_{j-l,t-l}$.
\end{example}

\medskip

To facilitate the development of our general theory, we introduce the following assumptions.

\begin{assumption}\label{Ass1}
$\max_{t\ge 1} \sum_{j=1}^{\infty} j  \|\bm{B}_{j,t}\|  < \infty$, $\limsup_{T\to \infty}\sum_{t=1}^{T-1}\sum_{j=1}^{\infty}j \|\bm{B}_{j,t+1}-\bm{B}_{j,t}\|<\infty$, and $\limsup_{T\to \infty}\sum_{t=1}^{T-1}\|\bm{\mu}_{t+1}-\bm{\mu}_{t}\|<\infty$.
\end{assumption}

\begin{assumption}\label{Ass2}
  $\{\bm{\epsilon}_t\}_{t=-\infty}^{\infty}$ is a martingale difference sequences (m.d.s.) adapted to the filtration $\left\{\mathcal{F}_t\right\}$, where $\mathcal{F}_t=\sigma\left(\bm{\epsilon}_t,\bm{\epsilon}_{t-1},\ldots\right)$ is the $\sigma$--field generated by $\left(\bm{\epsilon}_t,\bm{\epsilon}_{t-1},\ldots\right)$, $E\left(\bm{\epsilon}_t \bm{\epsilon}_t^\top | \mathcal{F}_{t-1}\right)=\bm{I}_d$ almost surely (a.s.), and $ \max_{t\ge 1}E\left[\left\|\bm{\epsilon}_t\right\|^\delta\right] < \infty$ for some $\delta > 2$.
\end{assumption}

Assumption \ref{Ass1} regulates the matrices $\bm{B}_{j,t}$'s, and ensures the validity of the BN decomposition under the time--varying framework. It covers cases such as (i) the parametric setting of \cite{phillips1992asymptotics}, and (ii) $\bm{B}_{j,t}:= \bm{B}_j(\tau_t)$, where $\bm{B}_j(\cdot)$ satisfies Lipschitz continuity on $[0,1]$ for all $j$. Assumption \ref{Ass2} imposes conditions on the innovation error terms by replacing the commonly used independent and identically distributed (i.i.d.) innovations (e.g., \citealp{dahlhaus2009empirical}) with a martingale difference structure.

\medskip

We are now ready to present a summary of useful results for Examples 1--3, which explains why model \eqref{eq1} serves as a foundation of the examples given above.

\begin{proposition}\label{Proposition2.1}
\item

\begin{enumerate}
\item Consider Examples \ref{Example1} and \ref{Example2}. Suppose that the roots of $\bm{I}_d-\bm{A}_{1,t}-\cdots -\bm{A}_{p,t}=\bm{0}_d$ all lie outside the unit circle uniformly over $t$, $\limsup_{T\to \infty}\sum_{t=1}^{T-1}\left\|\bm{A}_{m,t+1}-\bm{A}_{m,t}\right\| < \infty$ for $m=1,\ldots,p$ and $\bm{A}_{m,t}=\bm{A}_{m,1}$ for $t \leq 0$ and $m=1,\ldots,p$. In addition, suppose that in Example 2, $\limsup_{T\to \infty}\sum_{t=1}^{T-1}\left\|\bm{\Theta}_{m,t+1}-\bm{\Theta}_{m,t}\right\| < \infty$ for $m=1,\ldots,q$. Then both \eqref{eq5} and \eqref{eq6} are time--varying MA$(\infty)$ processes, in which the MA coefficients satisfy Assumption \ref{Ass1}.

\item For Example \ref{Example3}, let $\limsup_{T\to \infty}\sum_{t=1}^{T-1}\sum_{j=1}^{\infty}j \|\bm{\Psi}_{j,t+1}-\bm{\Psi}_{j,t}\|<\infty$ and $\max_{t\ge 1} \sum_{j=1}^{\infty} j  \|\bm{\Psi}_{j,t}\|  < \infty$. Moreover, let $\{\bm{\Theta}_{j,t}\}$ satisfy the same conditions as those for $\{\bm{\Psi}_{j,t}\}$. Then \eqref{eq6} is a time--varying MA$(\infty)$, in which the MA coefficients satisfy Assumption \ref{Ass1}.
\end{enumerate}
\end{proposition}

\medskip

We now move on to investigate asymptotic properties for model \eqref{eq1}. We first propose some estimates for several population moments of $\bm{x}_t$ in \eqref{eq1}, which help derive the asymptotic theory throughout this paper. To conserve space, we present the rates of the uniform convergence below, while extra results on point--wise convergence are given in the supplementary Appendix B.

\begin{theorem}\label{Theorem2.1}
Let Assumptions \ref{Ass1} and \ref{Ass2} hold. In addition, let $\left\{\bm{W}_{T,t}(\cdot)\right\}_{t=1}^T$ be a sequence of $m \times d$ matrices of deterministic functions, in which $m$ is fixed, each functional component is Lipschitz continuous and defined on a compact set $[a,b]$. Moreover, suppose that

\begin{enumerate}
\item $\sup_{\tau\in[a,b]}\sum_{t=1}^{T} \|\bm{W}_{T,t}(\tau) \|=O(1)$;
\item $\sup_{\tau\in[a,b]}\sum_{t=1}^{T-1} \|\bm{W}_{T,t+1}(\tau)-\bm{W}_{T,t}(\tau) \| = O(d_T)$, where $d_T=\sup_{\tau\in[a,b],t\ge 1} \|\bm{W}_{T,t}(\tau) \|$.
\end{enumerate}
Then as $T\to \infty$,
\begin{enumerate}
\item $ \sup_{\tau\in[a,b]} \|\sum_{t=1}^{T}\bm{W}_{T,t}(\tau)\left(\bm{x}_t-E(\bm{x}_t)\right) \|=O_P(\sqrt{d_T\log T})$ provided $T^{\frac{2}{\delta}} d_T \log T \rightarrow0$;

\item $\sup_{\tau\in[a,b]} \|\sum_{t=1}^{T}\bm{W}_{T,t}(\tau)\left(\bm{x}_t\bm{x}_{t+p}^\top-E(\bm{x}_t\bm{x}_{t+p}^\top)\right)\|=O_P (\sqrt{d_T\log T} )$ for any fixed integer $p\geq0$ provided $T^{\frac{4}{\delta}} d_T \log T \rightarrow0$ and $\max_{t\ge 1} E(\|\bm{\epsilon}_t  \|^4 |\mathcal{F}_{t-1} ) < \infty$ a.s., where $\delta>2$ is the same as in Assumption 2.
\end{enumerate}
\end{theorem}

Theorem \ref{Theorem2.1} is readily used for studying many useful cases, including weighted kernel estimators (see Lemma \ref{LemmaB.7} of Appendix B for example), and will be repeatedly used in many of the theoretical derivations of this paper. Theorem \ref{Theorem2.1} is also helpful to a broad range of studies, such as those mentioned in \cite{HLG20}, \cite{FanYao03}, \cite{Gao2007}, \cite{LR2007}, \cite{hansen2008uniform}, \cite{wang2014uniform}, and \cite{li2016uniform}.

\subsection{On the Trending Term --- \texorpdfstring{$\bm{\mu}_t$}{}}\label{Section2.2}

As modelling time--varying means is always an important task in time series analysis (e.g., \citealp{wu2007inference,friedrich2020autoregressive}), we infer $\bm{\mu}_t$ of the model  \eqref{eq1} below. Up to this point, we have not imposed any specific form on the components $\bm{\mu}_t $ and $\bm{B}_{j,t}$ of $\bm{x}_t$. To carry on with our investigation,  we suppose further that $\bm{\mu}_t = \bm{\mu}(\tau_t)$ and $\bm{B}_{j,t}=\bm{B}_{j}(\tau_t)$ with $\tau_t=t/T$, so \eqref{eq1} can be expressed by

\begin{equation}\label{eq7}
\bm{x}_t=\bm{\mu}(\tau_t) +\sum_{j=0}^{\infty}\bm{B}_{j}(\tau_t)\bm{\epsilon}_{t-j}.
\end{equation}
The challenge then lies in the fact that ``residuals'' are time--varying linear processes. Some detailed explanations can be found in \cite{dahlhaus2012locally}.

The following assumptions are necessary for the development of our trend estimation.

\begin{assumption}\label{Ass3}
$\bm{\mu}(\cdot)$ and $\bm{B}_{j}(\cdot)$ are $d \times 1$ vector and $d \times d$ matrix respectively. Each functional component of $\bm{\mu}(\cdot)$ and $\bm{B}_{j}(\cdot)$ is second order continuously differentiable on $[0, 1]$. Moreover, $\sup_{\tau\in [0.1]} \sum_{j=1}^{\infty} j \|\bm{B}_j^{(\ell)}(\tau) \|  < \infty$ for $\ell =0,1$.
\end{assumption}

\begin{assumption}\label{Ass4}
Let $K(\cdot)$ be a symmetric probability kernel function and Lipschitz continuous on $[-1,1]$. Also,  let $h \to 0$ and $Th\to \infty$ as $T\rightarrow \infty$.
\end{assumption}

Assumption \ref{Ass3} can be considered as a stronger version than Assumption \ref{Ass1}. Assumption \ref{Ass4} is standard in the literature of kernel regression (\citealp{LR2007}).

\medskip

For $\forall\tau\in (0,1)$, we recover $\bm{\mu}(\tau)$ by the next estimator.
\begin{equation}\label{eq8}
 \widehat{\bm{\mu}}(\tau)=\left(\sum_{t=1}^{T}K_h(\tau_t-\tau)\right)^{-1}\sum_{t=1}^{T}\bm{x}_t K_h (\tau_t-\tau).
\end{equation}

We are now  ready to establish an important and useful theorem.

\begin{theorem}\label{Theorem2.2}
Let Assumptions \ref{Ass2}--\ref{Ass4} hold. For $\forall\tau \in (0,1)$, as $T\to \infty$,
\begin{equation*}
\sqrt{Th}\left(\widehat{\bm{\mu}}(\tau)-\bm{\mu}(\tau)-\frac{1}{2}h^2\tilde{c}_2{\bm{\mu}}^{(2)}(\tau)\right) \to_D N\left(\bm{0}_{d\times 1},\bm{\Sigma}_{\bm{\mu}}(\tau)\right),
\end{equation*}
where $\bm{\Sigma}_{\bm{\mu}}(\tau) =\tilde{v}_0 \left(\sum_{j=0}^{\infty}\bm{B}_j(\tau) \, \sum_{j=0}^{\infty}\bm{B}_j^\top(\tau)\right)$, and $\tilde{c}_2$ and $\tilde{v}_0$ are defined in the end of Section \ref{Section1}.
\end{theorem}

Note that $\bm{\Sigma}_{\bm{\mu}}(\cdot)$ is the long--run covariance matrix, which in general cannot be estimated directly. To construct confidence intervals practically, we use a dependent wild bootstrap (DWB) method which is initially proposed by \cite{shao2010dependent} for stationary time series. For the sake of space, the detailed procedure with the associated asymptotic properties is presented in Appendix \ref{AppendixA.2}.

\section{Time--Varying VAR}\label{Section3}
\renewcommand{\theequation}{3.\arabic{equation}}

\setcounter{equation}{0}

In this section, we pay particular attention to one of the most popular models of the VMA$(\infty)$ family --- VAR. Many multivariate time series exhibit time--varying simultaneous interrelationships and changes in unconditional volatility \cite[e.g.,][]{justiniano2008time,coibion2011monetary}. Along this line, time--varying VAR models have proven to be especially useful for describing the dynamics of the multivariate time series. Majority time--varying VAR models are investigated under the Bayesian framework, while little has been done using a frequentists' approach. Building on Section \ref{Section2}, we consider a time--varying VAR model under the nonparametric framework, and establish the corresponding estimation theory.

Suppose that we observe $\left(\bm{x}_{-p+1},\ldots,\bm{x}_0,\bm{x}_1,\ldots,\bm{x}_T\right)$ from the following data generating process. Accounting for heteroscedasticity, we consider the next model.
\begin{equation}\label{eq9}
\bm{x}_t=\bm{a}(\tau_t)+\sum_{j=1}^{p}\bm{A}_{j}(\tau_t) \bm{x}_{t-j}+\bm{\eta}_t, \ \ \mbox{with} \ \ \bm{\eta}_t=\bm{\omega}(\tau_t)\bm{\epsilon}_{t},
\end{equation}
where $\bm{\omega}(\tau)$ is a matrix of unknown functions of $\tau$. The model \eqref{eq9} allows dynamic variations for both the coefficients and the covariance matrix. We infer $\bm{a}(\cdot)$ and $ \bm{A}_{j}(\cdot)$'s below, which are respectively $d\times 1$ vector and $d\times d$ matrices of unknown smooth functions. In addition, we are interested in the $d\times d$ dimension $\bm{\omega}(\cdot)$ which governs the dynamics of the covariance matrix of $\{\bm{\eta}_t\}_{t=1}^T$. As mentioned in \cite{primiceri2005time}, allowing $\bm{\omega}(\cdot)$ to vary over time is important theoretically and practically, because a constant covariance matrix implies that the shock to the $i$th variable of $\bm{x}_{t}$ has a time--invariant effect on the $j$th variable of $\bm{x}_{t}$, restricting simultaneous interactions among multiple variables to be time--invariant.

\subsection{Estimation Method and Asymptotic Theory}

To facilitate the development, we first impose the following conditions.

\begin{assumption}\label{Ass5}

\item

\begin{enumerate}
\item The roots of $\bm{I}_d-\bm{A}_{1}(\tau)L-\cdots -\bm{A}_{p}(\tau)L^p=\bm{0}_d$ all lie outside the unit circle uniformly in $\tau \in [0,1]$.

\item Each element of $\bm{A}(\tau)=\left(\bm{a}(\tau),\bm{A}_1(\tau),\ldots,\bm{A}_{p}(\tau)\right)$ is second order continuously differentiable on $[0,1]$ and $\bm{A}(\tau)=\bm{A}(0)$ for $\tau<0$.

\item Each element of $\bm{\omega}(\tau)$ is second order continuously differentiable on $[0,1]$. Moreover, $\bm{\Omega}(\tau)=\bm{\omega}(\tau)\bm{\omega}(\tau)^\top$ is positive definite uniformly in $\tau \in [0,1]$ and $\bm{\omega}(\tau)=\bm{\omega}(0)$ for $\tau<0$.
\end{enumerate}
\end{assumption}

Assumption \ref{Ass5}.1 ensures that model \eqref{eq9} is neither unit--root nor explosive, while Assumption \ref{Ass5}.2 allows the underlying data generate process to evolve over time in a smooth manner. In addition, the conditions $\bm{A}(\tau)=\bm{A}(0)$ and $\bm{\omega}(\tau)=\bm{\omega}(0)$ for $\tau<0$ yield
\be
\bm{x}_t=\bm{a}(0)+\sum_{j=1}^{p}\bm{A}_{j}(0) \bm{x}_{t-j}+\bm{\omega}(0)\bm{\epsilon}_{t}
\label{eqjiti1}
\ee
for all $t\leq0$, which ensures \eqref{eqjiti1} behaves like a stationary parametric VAR$(p)$ model for $t\le 0$. Similar treatments can be found in \cite{vogt2012nonparametric} for the univariate setting. With the above assumptions in hand, the following proposition shows that model \eqref{eq9} can be approximated by a time--varying VMA$(\infty)$ process satisfying Assumption \ref{Ass3}.

\begin{proposition}\label{Proposition3.1}
Under Assumption \ref{Ass5}, there exists a VMA$(\infty)$ process
\begin{eqnarray}\label{eq10}
\widetilde{\bm{x}}_t= \bm{\mu}(\tau_t)+\bm{B}_0(\tau_t)\bm{\epsilon}_t+\bm{B}_1(\tau_t)\bm{\epsilon}_{t-1}+\bm{B}_2(\tau_t)\bm{\epsilon}_{t-2}+\cdots
\end{eqnarray}
such that $\max_{t\geq 1} E\left\|\bm{x}_t-\widetilde{\bm{x}}_t\right\|=O(T^{-1})$, where $\bm{\mu}(\tau)=\bm{a}(\tau)+\sum_ {j=1}^{\infty}\bm{\Psi}_j(\tau)\bm{a}(\tau)$, $\bm{B}_0(\tau)=\bm{\omega}(\tau)$, $\bm{B}_j(\tau)=\bm{\Psi}_j(\tau)\bm{\omega}(\tau)$, $\bm{\Psi}_j(\tau)=\bm{J}\bm{\Phi}^j(\tau) \bm{J}^\top$ for $j\geq 1$, $\bm{\Phi}(\cdot)$ is defined as follows:
\begin{eqnarray}\label{eq10a}
\bm{\Phi}(\tau)=\left(\begin{matrix}
       \bm{A}_{1}(\tau) & \cdots & \bm{A}_{p-1}(\tau) & \bm{A}_{p}(\tau)  \\
       \bm{I}_d & \cdots& \bm{0}_d & \bm{0}_d\\
       \vdots & \ddots&\vdots & \vdots\\
       \bm{0}_d &\cdots &\bm{I}_d & \bm{0}_d\\
    \end{matrix} \right),
\end{eqnarray}
and $\bm{J}=\left[\bm{I}_d,\bm{0}_{d\times d(p-1)}\right]$. Moreover, $\bm{\mu}(\cdot)$ and $\bm{B}_j(\cdot)$ fulfil Assumption \ref{Ass3}.
\end{proposition}


Under Assumption \ref{Ass5}, when $\tau_t$ is in a small neighbourhood of $\tau$, we have

\begin{eqnarray}\label{eq11}
\bm{x}_t \approx \bm{Z}_{t-1}^\top\mathrm{vec} [\bm{A}(\tau)]+\bm{\eta}_t,
\end{eqnarray}
where $\bm{Z}_{t-1}= \bm{z}_{t-1}\otimes \bm{I}_d$ and $\bm{z}_{t-1}=(1,\bm{x}_{t-1}^\top,\ldots,\bm{x}_{t-p}^\top)^\top$. The estimators of $\bm{A}(\tau)$ and $\bm{\Omega}(\tau)$ are then sequentially given by

\begin{eqnarray}\label{eq12}
\mathrm{vec} [\bm{\widehat{A}}(\tau)]&=&\left( \sum_{t=1}^{T} \bm{Z}_{t-1}\bm{Z}_{t-1}^\top K_h (\tau_t-\tau) \right)^{-1} \sum_{t=1}^{T}\bm{Z}_{t-1}\bm{x}_t K_h (\tau_t-\tau) , \nonumber\\
\bm{\widehat{\Omega}}(\tau)&=&\left(\sum_{t=1}^{T}K_h (\tau_t-\tau )\right)^{-1}\sum_{t=1}^{T}\bm{\widehat{\eta}}_t\bm{\widehat{\eta}}_t^\top K_h (\tau_t-\tau ),
\end{eqnarray}
where $\bm{\widehat{\eta}}_t=\bm{x}_t-\bm{\widehat{A}}(\tau_t)\bm{z}_{t-1}$. The asymptotic properties associated with \eqref{eq12} are summarized in the next theorem.

\begin{theorem}\label{Theorem3.1}
Let Assumptions \ref{Ass2}, \ref{Ass4} and \ref{Ass5} hold. Suppose further that $\frac{T^{1-\frac{4}{\delta}}h}{\log T} \to \infty$ as $T\rightarrow \infty$ and $\max_{t\geq1} E (\|\bm{\epsilon}_t \|^4 |\mathcal{F}_{t-1} ) < \infty $ a.s. Then the following results hold.

\begin{enumerate}
\item $\sup_{\tau \in [h,1-h]} \| \bm{\widehat{A}}(\tau)-\bm{A}(\tau) \|=O_P \left(h^2+ (\frac{\log T}{Th} )^{1/2} \right)$;

\item In addition, conditional on $\mathcal{F}_{t-1}$, the third and fourth moments of $\bm{\epsilon}_t$ are identical to the corresponding unconditional moments a.s.. For $\forall\tau \in (0,1)$,
\begin{equation*}
\sqrt{Th}\left(\begin{matrix}
\mathrm{vec}\left(\bm{\widehat{A}}(\tau)-\bm{A}(\tau)-\frac{1}{2}h^2\tilde{c}_2\bm{A}^{(2)}(\tau)\right) \\
\mathrm{vech}\left(\bm{\widehat{\Omega}}(\tau)-\bm{\Omega}(\tau)-\frac{1}{2}h^2\tilde{c}_2\bm{\Omega}^{(2)}(\tau)\right)  \end{matrix}
\right)\to_D N\left(\bm{0},\bm{V}(\tau)\right),
\end{equation*}
where $\bm{V}(\tau) $ is defined in \eqref{EqA.3} for the sake of presentation.

\item $ \widehat{\bm{V}}(\tau)\to_P \bm{V}(\tau)$, where $\widehat{\bm{V}}(\tau)$ is defined in \eqref{EqA.5} for the sake of presentation.
\end{enumerate}
\end{theorem}

The first result of Theorem \ref{Theorem3.1} provides the uniform convergence rate for $\bm{\widehat{A}}(\tau)$. As a consequence, it allows us to establish a joint asymptotic distribution for the estimates of the coefficients and innovation covariance in the second result. For $\delta>5$, the usual optimal bandwidth $h_{opt}=O\left(T^{-1/5}\right)$ satisfies the condition $\frac{T^{1-\frac{4}{\delta}}h}{\log T} \to \infty$. The third result ensures the confidence interval can be constructed practically.

\medskip

Before moving on to impulse responses, we consider a practical issue --- the choice of the lag $p$, which is usually unknown in practice and needs to be decided by the data. We select the number of lags by minimizing the next information criterion:
\begin{eqnarray}\label{eq13}
\widehat{\mathsf{p}} = \argmin_{1\le \mathsf{p}\le\mathsf{P} }\text{IC}(\mathsf{p}),
\end{eqnarray}
where $\text{IC}(\mathsf{p})=\log \left\{\text{RSS}(\mathsf{p})\right\}+\mathsf{p}\cdot\chi_T$, $\text{RSS}(\mathsf{p})=\frac{1}{T}\sum_{t=1}^{T}\widehat{\eta}_{\mathsf{p},t}^\top \widehat{\eta}_{\mathsf{p},t}$, $\chi_T$ is the penalty term, and $\mathsf{P}$ is a sufficiently large fixed positive integer. The next theorem shows the validity of \eqref{eq13}.

\begin{theorem}\label{Theorem3.2}
Let Assumptions \ref{Ass2}, \ref{Ass4} and \ref{Ass5} hold. Suppose $\frac{T^{1-\frac{4}{\delta}}h}{\log T} \to \infty$, $\max_{t\geq1} E (\|\bm{\epsilon}_t \|^4 |\mathcal{F}_{t-1} ) < \infty $ a.s., $\chi_T\to 0$, and $(c_T\phi_T)^{-1}\chi_T\to \infty$, where $c_T=h^2+\left(\frac{\log T}{Th}\right)^{1/2}$ and $\phi_T=h+\left(\frac{\log T}{Th}\right)^{1/2}$. Then $\Pr\left(\widehat{\mathsf{p}}=p\right)\to 1$ as $T\to \infty$.
\end{theorem}

In view of the conditions on $\chi_T$, a natural choice is
\begin{eqnarray*}
\chi_T = \max\left\{h^3,h\left(\frac{\log T}{Th}\right)^{1/2},\frac{\log T}{Th}\right\}\cdot \log(1/h).
\end{eqnarray*}
In the supplementary Appendix \ref{AppendixB.1}, we conduct intensive simulations to examine the finite sample performance of the above information criterion.

\subsection{Impulse Response Functions}

We now focus on the impulse responses below, which capture the dynamic interactions among the variables of interest in a wide range of practical cases.

By Proposition \ref{Proposition3.1}, the impulse response functions of $\bm{x}_t$ is asymptotically equivalent to those of $\widetilde{x}_t$. Hence, recovering the impulse response functions requires estimating $\bm{\Psi}_j(\cdot)$'s and $\bm{\omega}(\cdot)$, which is then down to the estimation of $\bm{\Phi}(\cdot)$ and $\bm{\omega}(\cdot)$ by construction. Note further that $\bm{\Phi}(\cdot)$ is a matrix consisting of the coefficients of \eqref{eq9}. The estimator of $\bm{\Phi}(\cdot)$ is intuitively defined as $\widehat{\bm{\Phi}}(\cdot)$, in which we replace $\bm{A}_j(\cdot)$ of \eqref{EqA.6} with the corresponding estimator obtained from \eqref{eq12}. Furthermore, we require $\bm{\omega}(\tau)$ to be a lower--triangular matrix in order to fulfil the identification restriction. Thus, $\widehat{\bm{\omega}}(\tau)$ is chosen as the lower triangular matrix from the Cholesky decomposition of $\bm{\widehat{\Omega}}(\tau)$ such that $\bm{\widehat{\Omega}}(\tau)=\bm{\widehat{\omega}}(\tau)\bm{\widehat{\omega}}^\top(\tau)$.

With the above notation in hand, we are now ready to present the estimator of $\bm{B}_j(\tau) $ by

\begin{eqnarray*}
\bm{\widehat{B}}_j(\tau)=\bm{\widehat{\Psi}}_j(\tau)\bm{\widehat{\omega}}(\tau),
\end{eqnarray*}
where $\bm{\widehat{\Psi}}_j(\tau)=\bm{J} \bm{\widehat{\Phi}}^j(\tau) \bm{J}^\top$. The corresponding asymptotic results are summarized in the following theorem.

\begin{theorem}\label{Theorem3.3}
Let Assumptions \ref{Ass2}, \ref{Ass4} and \ref{Ass5} hold, and let $T\to \infty$. Suppose further that $\frac{T^{1-\frac{4}{\delta}}h}{\log T} \to \infty$, $\max_{t\geq1} E\left(\left\|\bm{\epsilon}_t \right\|^4 |\mathcal{F}_{t-1}\right) < \infty $ a.s. and conditional on $\mathcal{F}_{t-1}$, the third and fourth moments of $\bm{\epsilon}_t$ are identical to the corresponding unconditional moments a.s.. Then for any fixed integer $j\ge 0$

\begin{eqnarray*}
 \sqrt{Th}\left(\mathrm{vec}\left(\bm{\widehat{B}}_j(\tau)-\bm{B}_j(\tau)\right)-\frac{1}{2}h^2\tilde{c}_2\bm{B}_j^{(2)}(\tau)\right)\to_D N\left(0,\bm{\Sigma}_{\bm{B}_j}(\tau)\right),
\end{eqnarray*}
where
\begin{eqnarray*}
\bm{\Sigma}_{\bm{B}_j}(\tau)&=&[\bm{C}_{j,1}(\tau),\bm{C}_{j,2}(\tau)]\bm{V}(\tau)[\bm{C}_{j,1}(\tau),\bm{C}_{j,2}(\tau)]^\top,\\
\bm{B}_j^{(2)}(\tau)&=&\bm{C}_{j,1}(\tau)\mathrm{vec}\left(\bm{A}^{(2)}(\tau)\right) +\bm{C}_{j,2}(\tau)\mathrm{vech}\left(\bm{\Omega}^{(2)}(\tau)\right),\\
\bm{C}_{0,1}(\tau)&=&0,\quad \bm{C}_{0,2}(\tau)= \bm{L}_d^\top\left(\bm{L}_d(\bm{I}_{d^2}+\bm{K}_{dd})(\bm{\omega}(\tau)\otimes \bm{I}_d)\bm{L}_d^\top \right)^{-1},\\
\bm{C}_{j,1}(\tau)&=&(\bm{\omega}^\top(\tau)\otimes \bm{I}_d) \left(\sum_{m=0}^{j-1} \bm{J}(\bm{\Phi}^\top(\tau))^{j-1-m}\otimes (\bm{J} \bm{\Phi}^m(\tau)\bm{J}^\top)\right)\cdot [\bm{0}_{d^2p\times d},\bm{I}_{d^2p}],\ j \geq 1,\\
\bm{C}_{j,2}(\tau)&=&(\bm{I}_d\otimes(\bm{J}\bm{\Phi}^j(\tau)\bm{J}^\top)) \bm{L}_d^\top\left(\bm{L}_d(\bm{I}_{d^2}+\bm{K}_{dd})(\bm{\omega}(\tau)\otimes \bm{I}_d)\bm{L}_d^\top \right)^{-1},\ j \geq 1,
\end{eqnarray*}
in which the elimination matrix $\bm{L}_d$ satisfies that $\mathrm{vech}(\bm{F})=\bm{L}_d\mathrm{vec}(\bm{F})$ for any $d\times d$ matrix $\bm{F}$, and the commutation matrix $\bm{K}_{mn}$ satisfies that  $\bm{K}_{mn}\mathrm{vec}(\bm{G})=\mathrm{vec}(\bm{G}^\top)$ for any $m\times n$ matrix $\bm{G}$.
\end{theorem}

To close this section, we comment on how to construct the confidence interval. Since $\widehat{\bm{\Phi}}(\tau)\to_P \bm{\Phi}(\tau)$, $\widehat{\bm{\omega}}(\tau)\to_P\bm{\omega}(\tau)$ and $\widehat{\bm{V}}(\tau)\to_P \bm{V}(\tau)$ by Theorem \ref{Theorem3.1}, it is straightforward to have $\widehat{\bm{\Sigma}}_{\bm{B}_j}(\tau)\to_P\bm{\Sigma}_{\bm{B}_j}(\tau)$, where $\widehat{\bm{\Sigma}}_{\bm{B}_j}(\tau)$ has a form identical to $\bm{\Sigma}_{\bm{B}_j}(\tau)$ but replacing $\bm{\Phi}(\tau)$, $\bm{\omega}(\tau)$ and $ \bm{V}(\tau)$ with their estimators, respectively.

\medskip

We next show in Section \ref{Section4} about how to apply the proposed model and estimation method to an empirical data. Our findings show that the estimated coefficient matrices and impulse response functions capture various time--varying features.

\section{Empirical Study}\label{Section4}

\renewcommand{\theequation}{4.\arabic{equation}}

\setcounter{equation}{0}

In this section, we study the transmission mechanism of the monetary policy, and infer the long--run level of inflation (i.e., trend inflation) and the natural rate of unemployment (NAIRU). The trend inflation and NAIRU are of central position in setting monetary policy since the Federal Reserve Bank aims to mitigate deviations of inflation and unemployment from their long--run targets.  See \cite{stock2016core} for more relevant discussions.

As well documented, inflation is higher and more volatile during 1970--1980, but substantially decreases in the subsequent period, which is often referred to as the Great Moderation (\citealp{primiceri2005time}). The literature has considered two main classes of explanations: bad policy or bad luck. The first type of explanations focuses on the changes in the transmission mechanism \cite[e.g.,][]{cogley2005drifts}, while the second regards it as a consequence of changes in the size of exogenous shocks \cite[e.g.,][]{sims2006}. In what follows, we revisit the arguments associated with the Great Moderation using our approach. Also, we use the VMA$(\infty)$ representation of the VAR$(p)$ model to infer the path of trend inflation and NAIRU over time.

\subsection{In--Sample Study}

First, we estimate the time--varying VAR$(p)$ model using three commonly adopted macroeconomic variables of the literature (\citealp{primiceri2005time,cogley2010inflation}), which are the inflation rate (measured by the 100 times the year--over--year log change in the GDP deflator), the unemployment rate, representing the non--policy block, and the interest rate (measured by the average value for the Federal funds rates over the quarter), representing the monetary policy block. To isolate the monetary policy shocks, the interest rate is ordered last in the VAR model, and is treated as the monetary policy instrument. The identification requirement is that the monetary policy actions affect the inflation and the unemployment with at least one period of lag \cite[]{primiceri2005time}. The data are quarterly observations measured at an annual rate from 1954:Q3 to 2020:Q1, which are taken from the Federal Reserve Bank of St. Louis economic database. Figure \ref{Figure1} plots the three macro variables.

The order of the VAR$(p)$ model and the optimal bandwidth are determined by the information criterion \eqref{eq13} and the cross validation criterion \eqref{EqB.1}, respectively. We obtain $\widehat{\mathsf{p}}=3$ and $\widehat{h}_{cv}=0.435$. In the literature, the lag length is often assumed to be known with the values varying from $2$ to $4$, while our data--driven method indicates that 3 is the optimal value.

We first consider measuring the changes in the size of exogenous shocks. Figure \ref{Figure2} plots the estimated volatilities of the innovations as well as the associated 95\% confidence intervals. The figure exhibits evidence for a general decline in unconditional volatilities. Our results thus support the ``bad luck'' explanations \cite[e.g.,][]{primiceri2005time,sims2006}.

We then consider the responses of the inflation to the monetary policy shocks. Figure \ref{Figure3} plots the time--varying impulse responses of the inflation to a structural monetary shock as well as the associated 95\% confidence intervals. It is clear that the confidence intervals are much wider at the beginning of the sample period implying higher uncertainty before 1970. On the other hand, the structural responses of the inflation seem to be statistically insignificant from 1970 to 2010. Thus, we conclude that the monetary shocks have less influence on the inflation before and during the period of the Great Moderation.

Finally, we investigate the trend inflation and the NAIRU. \cite{petrova2019quasi} considers a Bayesian time--varying VAR(2) model, and induces the long--run mean of $\bm{x}_t$ by $\bm{\mu}_t=\lim_{p\rightarrow\infty}E_t(\bm{x}_{t+p})=(\bm{I}-\bm{A}_t)^{-1}\bm{a}_t$, where $\bm{a}_t$ is the intercept term and $\bm{A}_t$ is the autoregressive coefficients. The main difference between our method and Petrova's method is that we invert time--varying VAR$(p)$ model to the time--varying MA$(\infty)$ model, and then explicitly estimate the underlying trends of the inflation and the NAIRU using model \eqref{eq7}.

Figure \ref{Figure4} plots the estimates of the trend inflation and the NAIRU. It is obvious that the underlying trend of the inflation is high in the 1970s, but decreases in the subsequent period. After the Great Moderation, the long--run level of the inflation is below, but quite close to the Federal Reserve's target of 2\%, which indicates that the inflation is more anchored now than in the 1970s. However, the NAIRU is less persistent and fluctuates over time.

\subsection{Out--of--Sample Forecasting}

In this subsection, we focus on the out--of--sample forecasting, and compare the forecasts of a Bayesian time--varying parameter VAR with stochastic volatility (TVP--SV) \cite[cf.,][]{primiceri2005time}, as well as a VAR model with constant parameters (CVAR).

Specifically, we consider the 1--8 quarters ahead forecasts. That is to forecast

\begin{eqnarray*}
\overline{\bm{x}}_{t+1|t+h}=h^{-1}\sum_{i=1}^{h}\bm{x}_{t+i}
\end{eqnarray*}
for $h=1$, $2$, $4$ and $8$, where $\bm{x}_t$ includes the values of the three aforementioned macro variables at the date $t$. The forecasts of $\overline{\bm{x}}_{t+1|t+h}$ are constructed by $\widehat{\overline{\bm{x}}}_{t+1|t+h}=\widehat{\bm{A}}_t\bm{z}_t$, in which $\widehat{\bm{A}}_t$ is estimated from CVAR, TVP--SV and the time--varying VAR model using the available data at time $t$ and $\bm{z}_t=[1,\bm{x}_{t}^\top,\bm{x}_{t-1}^\top,\bm{x}_{t-2}^\top]^\top$. The expanding window scheme is adopted. For comparison, we compute the root of mean square errors (RMSE) for CVAR as a benchmark, and the RMSE ratios for others. The out--of--sample forecast period covers 1985:Q2--end, about 35 years.

The forecasting results are presented in Table \ref{Table1}, in which the values represent the ratios of the RMSEs of the corresponding method over the RMSEs of the benchmark method (i.e., CVAR). The result shows that the time--varying VAR model and TVP--SV perform much better than CVAR, which implies the desirability of introducing variations in forecasting models. In addition, the time--varying VAR model has a better forecasting performance than TVP--SV with the increase of the forecast horizon.

{
\small
\begin{table}[H]
  \centering
  \caption{The 1--8 quarters ahead forecast. The values represent the ratios of the RMSEs of the corresponding method over the RMSEs of the benchmark method (i.e., CVAR). In each panel, the numbers in bold font of each column indicate the method that provides the best out--of--sample forecast for given $h$.}\label{Table1}
    \begin{tabular}{l rrrr}
    \hline\hline
          &$h=1Q$   &$h=2Q$    &$h=4Q$  &$h=8Q$  \\
    \hline
           & \multicolumn{4}{c}{Inflation rate, 1985:Q1--end} \\
    TVP--SV & \textbf{0.9771} &0.9841  &0.9850  &0.9880 \\
    time--varying VAR & 0.9986 & \textbf{0.9732}  &\textbf{0.9656} &\textbf{0.9855} \\
    \hline
          & \multicolumn{4}{c}{Unemployment rate, 1985:Q1--end} \\
    TVP--SV & \textbf{0.9119}  &0.9683 &1.0228  &1.0519 \\
    time--varying VAR & 0.9344 & \textbf{0.9671}  &\textbf{0.9928} &1.0076 \\
    \hline
          & \multicolumn{4}{c}{Interest rate, 1985:Q1--end} \\
    TVP--SV & \textbf{0.8388} &\textbf{0.8773} & \textbf{0.9435} & 1.0066 \\
    time--varying VAR & 0.9720 &0.9743 &0.9738 & \textbf{0.9892} \\
    \hline\hline
    \end{tabular}%
\end{table}
}

\begin{figure}[H]
\centering
{\includegraphics[width=16cm]{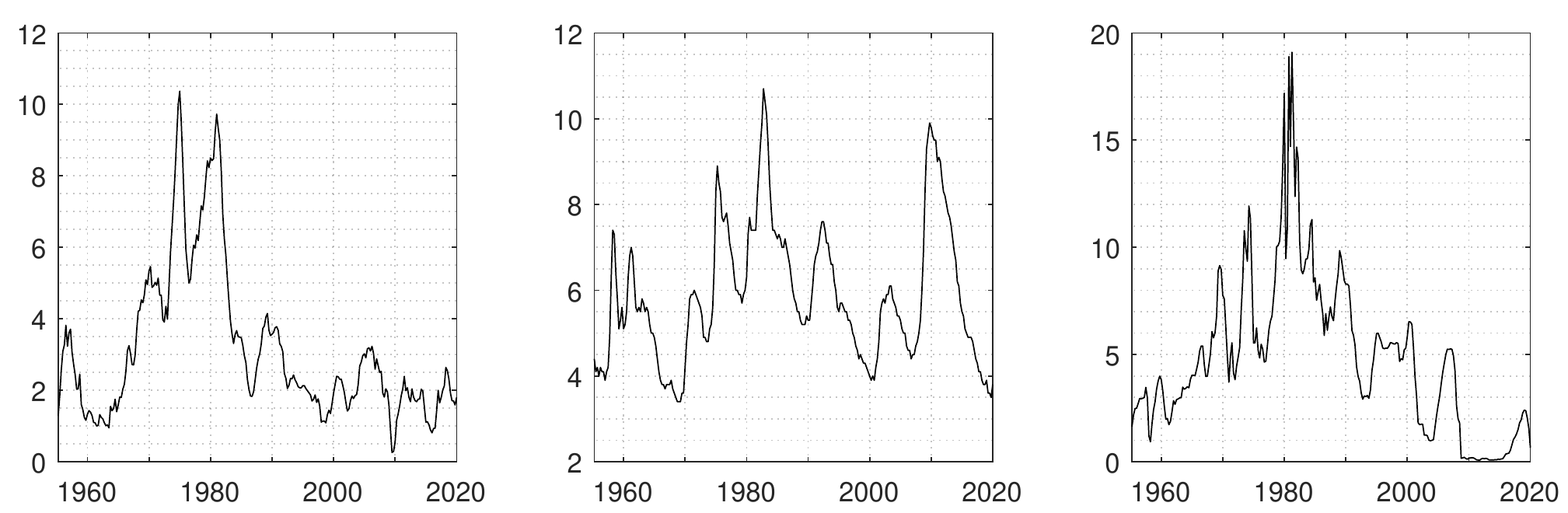}}
\caption{Plots of the inflation (left), the unemployment rate (middle) and the interest rate (right)}\label{Figure1}
\end{figure}

\begin{figure}[H]
\centering
{\includegraphics[width=16cm]{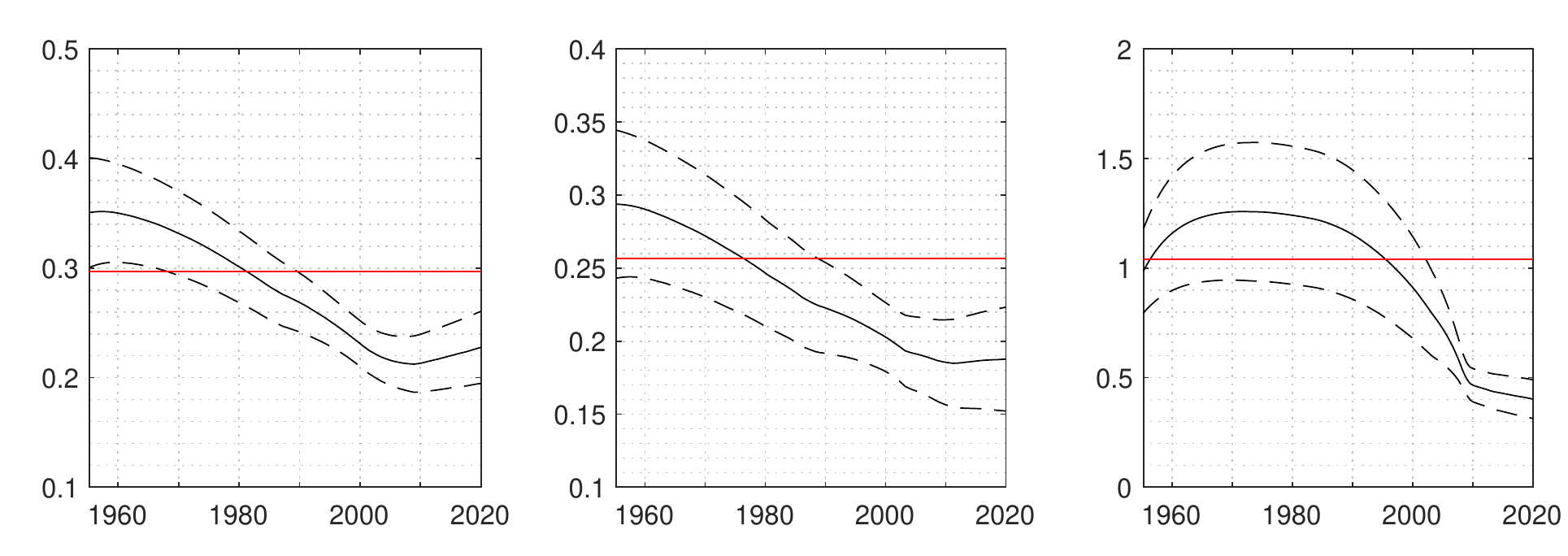}}
\caption{The estimated volatilities of the innovations in the inflation equation (left), the unemployment equation (middle) and the interest rate equation (right) as well as the associated 95\% confidence intervals. The red line denotes the estimated volatilities using the constant VAR model.}\label{Figure2}
\end{figure}

\begin{figure}[H]
\centering
{\includegraphics[width=10cm]{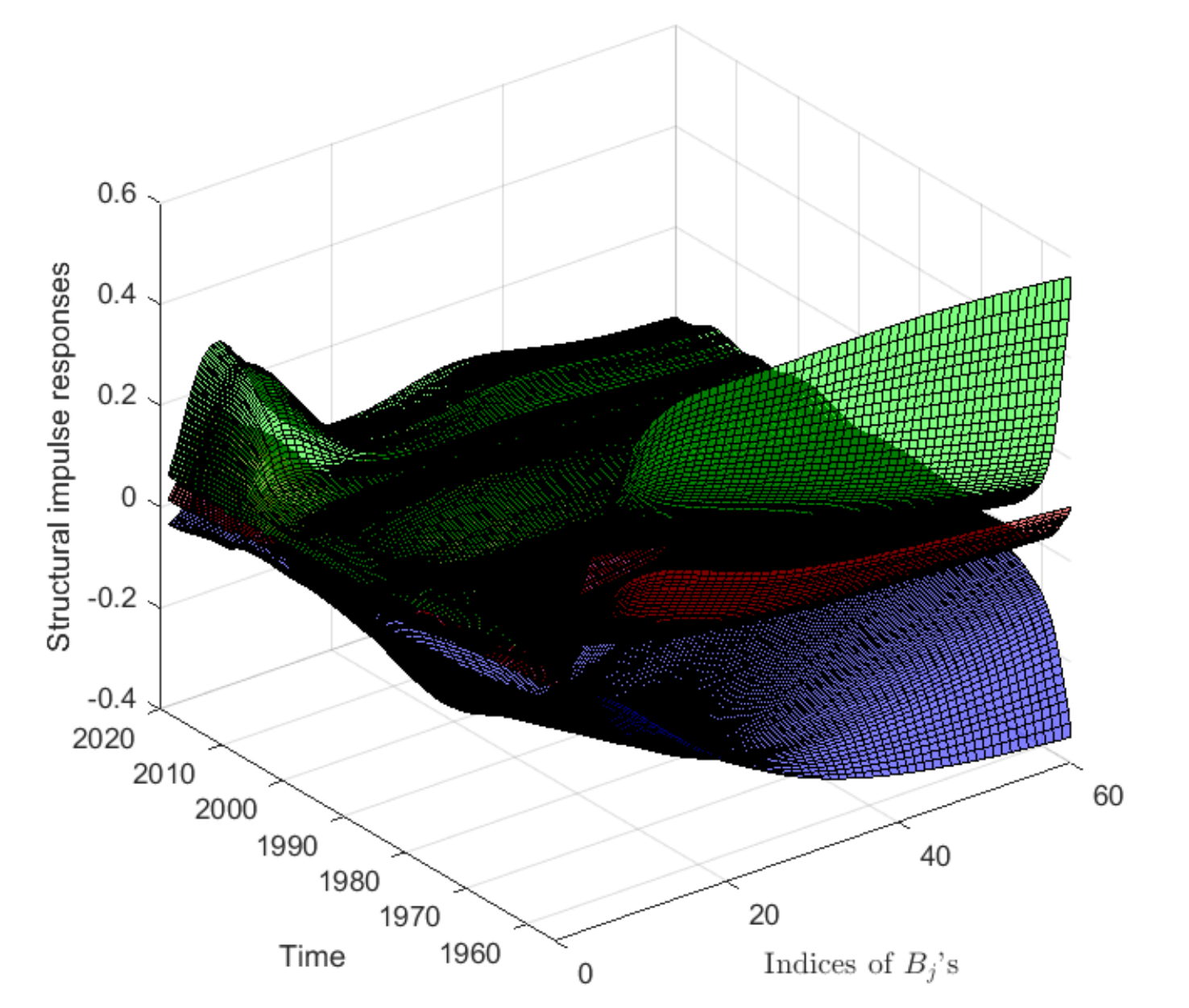}}
\caption{The time--varying structural impulse response of the inflation to monetary policy shocks as well as the associated 95\% confidence interval.}\label{Figure3}
\end{figure}

\begin{figure}[H]
\centering
{\includegraphics[width=16cm]{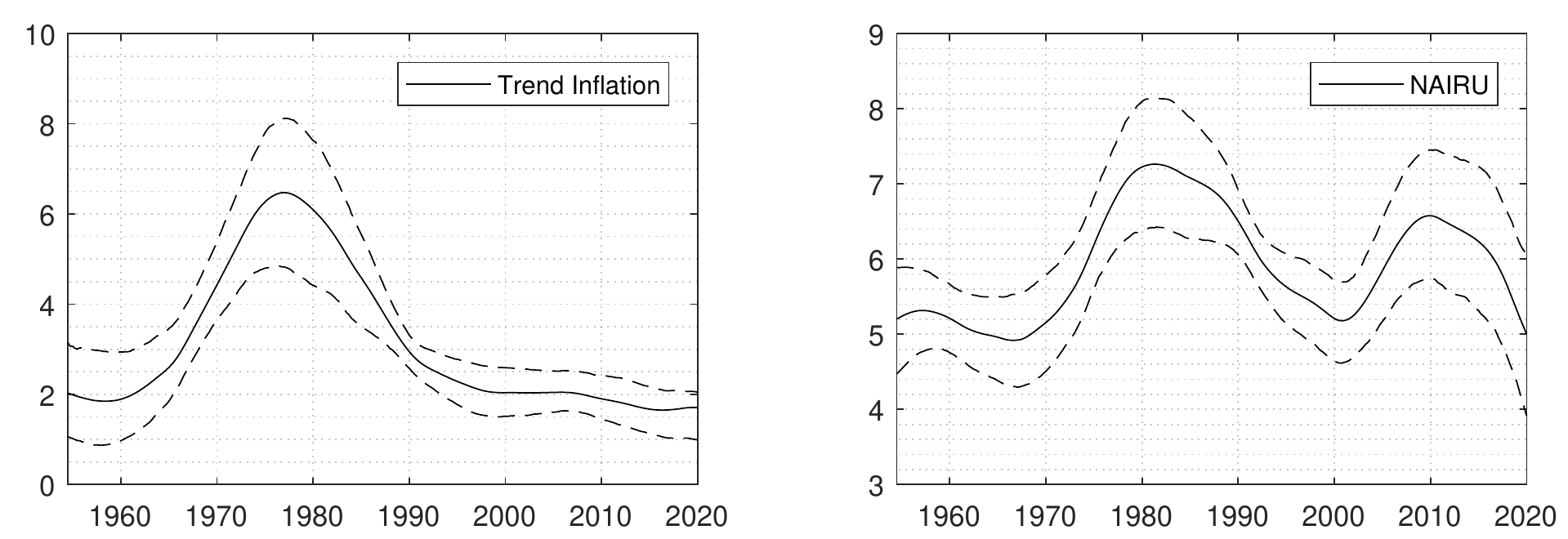}}
\caption{The estimated trends of the inflation and the NAIRU as well as the associated 95\% confidence intervals}\label{Figure4}
\end{figure}

\section{Conclusion}\label{Section5}

\renewcommand{\theequation}{5.\arabic{equation}}

\setcounter{equation}{0}

This paper has proposed a class of time--varying VMA($\infty$) models, which nest, for instance, time--varying VAR, time--varying VARMA, and so forth as special cases. Both the estimation methodology and asymptotic theory have been established accordingly. In the empirical study, we have investigated the transmission mechanism of monetary policy using U.S. data, and uncover a fall in the volatilities of exogenous shocks. Our findings include (i) monetary policy shocks have less influence on inflation before and during the so--called Great Moderation, (ii) inflation is more anchored recently, and (iii) the long--run level of inflation is below, but quite close to the Federal Reserve's target of two percent after the beginning of the Great Moderation period. In addition, in Appendix B of the online supplementary material, we have evaluated the finite--sample performance of the proposed model and estimation theory.

There are several directions for possible extensions. The first one is to test about whether the $d$--dimensional components of the {\rm VAR}$(p)$ process is cross--sectionally independent for the case where the dimensionality, $d$, and the number of lags, $p$, may diverge along with the sample size, $T$. The second one is about model specification testing to check whether some of the time--varying coefficient matrices, $\bm{A}_j(\tau)$, may just be constant matrices, $\bm{A}_{j0}$. Existing studies by \cite{Gao08}, \cite{pgy14}, and \cite{cw19} may be useful for both issues. The third one is to allow for some cointegrated structure in our settings. The recent work by \cite{zry19} provides us with a good reference. We wish to leave such issues for future study.

{\footnotesize
\bibliography{ma}
}

\clearpage


\section*{Appendix A}

\renewcommand{\theequation}{A.\arabic{equation}}
\renewcommand{\thesection}{A.\arabic{section}}
\renewcommand{\thefigure}{A.\arabic{figure}}
\renewcommand{\thetable}{A.\arabic{table}}
\renewcommand{\thelemma}{A.\arabic{lemma}}
\renewcommand{\theassumption}{A.\arabic{assumption}}
\renewcommand{\thetheorem}{A.\arabic{theorem}}

\setcounter{equation}{0}
\setcounter{lemma}{0}
\setcounter{section}{0}
\setcounter{table}{0}
\setcounter{figure}{0}
\setcounter{assumption}{0}

{\small

For the sake of presentation, we first provide some notation and mathematical sybmols in Appendix \ref{AppendixA.1}, and then present the dependent wild bootstrap (DWB) procedure with the associated asymptotic properties in Appendix \ref{AppendixA.2}. Some proofs of the main results are provided in Appendix \ref{AppA.3}. Simulations, some secondary results and omitted proofs are given in the online supplementary Appendix B. In what follows, $M$ and $O(1)$ always stand for constants, and may be different at each appearance.

\section{Notation and Mathematical Symbols}\label{AppendixA.1}

For ease of notation, we define three matrices $\bm{\Sigma}(\tau)$, $\bm{V}(\tau)$ and $\bm{\Phi}(\tau)$ with their estimators respectively. First, for $\forall \tau\in (0,1)$, let

\begin{eqnarray}\label{EqA.1}
 \bm{\Sigma}(\tau)&=&\left(\begin{matrix}
      1          & \bm{\mu}^\top(\tau)         & \cdots & \bm{\mu}^\top(\tau) \\
      \bm{\mu}(\tau)  & \bm{\Sigma}_{0}(\tau)     & \cdots & \bm{\Sigma}_{p-1}^\top(\tau) \\
      \vdots     & \vdots                 & \ddots & \vdots \\
      \bm{\mu}(\tau)  & \bm{\Sigma}_{p-1}(\tau)   & \cdots & \bm{\Sigma}_{0}(\tau)
    \end{matrix}\right),
\end{eqnarray}
in which $\bm{\mu}(\tau)$ and $\bm{B}_j(\tau)$ are defined in Proposition \ref{Proposition3.1} and $\bm{\Sigma}_m(\tau)=\bm{\mu}(\tau)\bm{\mu}(\tau)^\top+\sum_{j=0}^{\infty}\bm{B}_j(\tau)\bm{B}_{j+m}^\top(\tau)$ for $m=0,\ldots,p-1$. We define the estimator of $\bm{\Sigma}(\tau)$ as

\begin{eqnarray}\label{EqA.2}
\widehat{\bm{\Sigma}}(\tau)=\frac{1}{T}\sum_{t=1}^{T}\bm{z}_{t-1}\bm{z}_{t-1}^\top K_h(\tau_t-\tau),
\end{eqnarray}
where $\bm{z}_t$ is defined in \eqref{eq11}.

\medskip

Next, we let

\begin{eqnarray}\label{EqA.3}
\bm{V}(\tau)=\left(\begin{matrix}
                 \bm{V}_{1,1}(\tau) & \bm{V}_{2,1}^\top(\tau) \\
                 \bm{V}_{2,1}(\tau) & \bm{V}_{2,2}(\tau)
               \end{matrix} \right),
\end{eqnarray}
where

\begin{eqnarray}\label{EqA.4}
\bm{V}_{1,1}(\tau)&=&\tilde{v}_0\bm{\Sigma}^{-1}(\tau)\otimes \bm{\Omega}(\tau),\nonumber \\
\bm{V}_{2,1}(\tau)&=& \lim_{T\rightarrow \infty}\frac{h}{T}\sum_{t=1}^{T} E\left( \mathrm{vech} (\bm{\eta}_t\bm{\eta}_t^\top )\bm{\eta}_t^\top \bm{Z}_{t-1}^\top \right)K_h (\tau_t-\tau)^2\cdot  (\bm{\Sigma}^{-1}(\tau)\otimes \bm{I}_d ) , \nonumber \\
 \bm{V}_{2,2}(\tau)&=&\lim_{T\rightarrow\infty} \frac{h}{T}\sum_{t=1}^{T}E\left(\mathrm{vech}(\bm{\eta}_t\bm{\eta}_t^\top)\mathrm{vech}(\bm{\eta}_t\bm{\eta}_t^\top)^\top \right) K_h (\tau_t-\tau)^2 \nonumber \\
 &&-\tilde{v}_0\mathrm{vech}\left(\bm{\Omega}(\tau)\right) \mathrm{vech}\left(\bm{\Omega}(\tau)\right)^\top.
\end{eqnarray}
The estimator of $\bm{V}(\tau)$ is then defined as follows.

\begin{eqnarray}\label{EqA.5}
\widehat{\bm{V}}(\tau)&=&\left(\begin{matrix}
\widehat{\bm{V}}_{1,1}(\tau) & \widehat{\bm{V}}_{2,1}^\top(\tau) \\
\widehat{\bm{V}}_{2,1}(\tau) & \widehat{\bm{V}}_{2,2}(\tau)
  \end{matrix}\right),
\end{eqnarray}
where $\bm{\widehat{V}}_{1,1}(\tau)$, $\bm{\widehat{V}}_{2,1}(\tau)$ and $\bm{\widehat{V}}_{2,2}(\tau)$ have the forms identical to their counterparts of \eqref{EqA.4}, but we replace $ \bm{\Sigma}(\tau)$, $ \bm{\eta}_t$ and $\bm{\Omega}(\tau)$ with their estimators presented in \eqref{EqA.2} and \eqref{eq12}.

Finally, recall the following definition:
\begin{eqnarray}\label{EqA.6}
\bm{\Phi}(\tau)=\left(\begin{matrix}
       \bm{A}_{1}(\tau) & \cdots & \bm{A}_{p-1}(\tau) & \bm{A}_{p}(\tau)  \\
       \bm{I}_d & \cdots& \bm{0}_d & \bm{0}_d\\
       \vdots & \ddots&\vdots & \vdots\\
       \bm{0}_d &\cdots &\bm{I}_d & \bm{0}_d\\
    \end{matrix} \right).
\end{eqnarray}
Replacing $\bm{A}_j(\tau)$'s of \eqref{EqA.6} with their estimators obtained from \eqref{eq12} yields the estimator $\widehat{\bm{\Phi}}(\tau)$ straight away.

\section{Dependent Wild Bootstrap}\label{AppendixA.2}

We now present the detailed dependent wild bootstrap (DWB) procedure which is used to establish the confidence interval associated with Theorem \ref{Theorem2.2}.

\begin{enumerate}
\item For $\forall\tau\in(0,1)$, let $\widetilde{\bm{\mu}}(\tau)$ be the same as defined in \eqref{eq8} using an over--smoothing bandwidth $\widetilde{h}$. Obtain residuals $\bm{\widehat{e}}_t=\bm{x}_t-\widetilde{\bm{\mu}}(\tau_t)$ for $t\ge 1$.

\item Generate the bootstrap observations $\bm{x}_t^*=\widetilde{\bm{\mu}}(\tau_t)+\bm{e}_t^*$ for $t\ge 1$, where $\bm{e}_t^*=\xi_t^* \bm{\widehat{e}}_t$, $\{\xi_t^*\}$ is an $l$-dependent process satisfying $E(\xi_t^*)=0$, $E|\xi_t^*|^2=1$, $E|\xi_t^*|^\delta<\infty$ with $\delta>2$, and $E(\xi_t^*\xi_s^*)=a\left((t-s)/l\right)$ with some kernel function $a\left(\cdot\right)$ and tuning parameter $l$.

\item For $\forall\tau\in(0,1)$, let $\widehat{\bm{\mu}}^*(\tau)$ be the same as \eqref{eq8} but using $\{\bm{x}_t^*\}$.

\item Repeat Steps 2--3 $J$ times. Let $\bm{q}_{\alpha}(\tau)$ be the $\alpha$--quantile of $J$ statistics $\widehat{\bm{\mu}}^*(\tau)-\widetilde{\bm{\mu}}(\tau)$, and denote the $(1-\alpha)\cdot 100\%$ confidence interval of $\widehat{\bm{\mu}}(\tau)$ as $\left[\bm{\widehat{\mu}}(\tau)-\bm{q}_{1-\alpha/2}(\tau),\bm{\widehat{\mu}}(\tau)-\bm{q}_{\alpha/2}(\tau)\right]$.
\end{enumerate}

The above DWB procedure requires a tuning parameter $l$, which is the so--called ``block length'' (\citealp{shao2010dependent}). The following conditions are required to ensure the validity of the DWB procedure.

\begin{assumption}\label{AssA.1}
Suppose that $l\to \infty$, $\max \{\widetilde{h},h/\widetilde{h} \}\to 0$ and $l\cdot\max\{1/\sqrt{Th},\widetilde{h}^4,1/T\widetilde{h}\}\to 0$. Additionally, let $a(\cdot)$ be a symmetric kernel defined on $[-1,1]$ satisfying that $a(0)=1$ and $a(\cdot)$ is continuous at $0$ with $a^{(1)}(0)<\infty$.
\end{assumption}

We summarize the asymptotic properties of the DWB method by the next theorem.
\begin{theorem}\label{TheoremA.1}
Let Assumptions \ref{Ass2}--\ref{Ass4} and \ref{AssA.1} hold. For $\forall\tau \in (0,1)$, as $T\to \infty$,

\begin{enumerate}
\item $\sup_{\bm{w}\in \mathbb{R}^d}\left|\text{\normalfont Pr}^*\left[\sqrt{Th}\left(\bm{\widehat{\mu}}^*(\tau)-\bm{\widetilde{\mu}}(\tau)\right)\leq \bm{w}\right]-\Pr\left[\sqrt{Th}\left(\bm{\widehat{\mu}}(\tau)-\bm{\mu}(\tau)\right)\leq \bm{w}\right]\right|=o_P(1)$,

\item $ \liminf_{T\to \infty}\Pr\left(\bm{\mu}(\tau)\in \left[\bm{\widehat{\mu}}(\tau)-\bm{q}_{1-\alpha/2}(\tau),\bm{\widehat{\mu}}(\tau)-\bm{q}_{\alpha/2}(\tau)\right]\right)= 1-\alpha $,
\end{enumerate}
where $\text{\normalfont Pr}^*$ denotes the probability measure induced by the DWB procedure.
\end{theorem}

We comment on some practical issues.

\textbf{Bandwidth Selection:} Since the observations are dependent, we use the modified cross--validation criterion proposed by \cite{chu1991comparison}. Specifically, it is a ``leave--(2k+1)--out'' version of cross--validation and $h_{mcv}$ is selected to minimize the following objective function:
\begin{equation*}
  h_{mcv}=\min_h \sum_{t=1}^{T}\left(\bm{x}_t-\widehat{\bm{\mu}}_{k,h}(\tau_t)\right)^\top\left(\bm{x}_t-\widehat{\bm{\mu}}_{k,h}(\tau_t)\right),
\end{equation*}
where $\widehat{\bm{\mu}}_{k,h}(\tau)=\left(\sum_{t:|t-\tau T|>k}K\left(\frac{\tau_t-\tau}{h}\right)\right)^{-1}\sum_{t:|t-\tau T|>k}\bm{x}_t K\left(\frac{\tau_t-\tau}{h}\right)$. Moreover, in the first step of the bootstrap procedure, we follow the suggestion of
\cite{buhlmann1998sieve} to use $\widetilde{h}=c_0\cdot h_{mcv}^{5/9}$ with $c_0=2$.

\medskip

\textbf{Tuning parameter:} In the second step of the bootstrap procedure, we choose the kernel function $a(\cdot)$ and the bandwidth $l$ as in \cite{shao2010dependent}.

\section{Proofs of the Main Results}\label{AppA.3}

\begin{proof}[Proof of Theorem \ref{Theorem2.1}]

\item

\noindent (1). By Lemma \ref{LemmaB.3}, we have

\begin{equation*}
\bm{x}_t=\bm{\mu}_t+\bm{\mathbb{B}}_t(1)\bm{\epsilon}_t+\widetilde{\mathbb{B}}_t(L)\bm{\epsilon}_{t-1}-\widetilde{\mathbb{B}}_t(L)\bm{\epsilon}_t,
\end{equation*}
where $\bm{\mathbb{B}}_t(1)$ and $\widetilde{\mathbb{B}}_t(L)$ have been defined in equation (\ref{eq2}). We are then able to write

\begin{eqnarray*}
& &\sup_{\tau \in[a,b]} \left\|\sum_{t=1}^{T}\bm{W}_{T,t}(\tau)(\bm{x}_t-E(\bm{x}_t)) \right\|\\
&\leq &\sup_{\tau \in[a,b]} \left\|\sum_{t=1}^{T}\bm{W}_{T,t}(\tau)\bm{\mathbb{B}}_{t}(1)\bm{\epsilon}_t\right\|+\sup_{\tau \in[a,b]}  \left\|\bm{W}_{T,1}(\tau)\widetilde{\mathbb{B}}_1(L)\bm{\epsilon}_0\right\|+\sup_{\tau \in[a,b]}  \left\|\bm{W}_{T,T}(\tau)\widetilde{\mathbb{B}}_T(L)\bm{\epsilon}_T\right\|\\
& &+\sup_{\tau \in[a,b]} \left\|\sum_{t=1}^{T-1} \left(\bm{W}_{T,t+1}(\tau)\widetilde{\mathbb{B}}_{t+1}(L)-\bm{W}_{T,t}(\tau)\widetilde{\mathbb{B}}_t(L)\right) \bm{\epsilon}_t\right\|\\
&:=&I_{T,1}+I_{T,2}+I_{T,3}+I_{T,4},
\end{eqnarray*}
where the definitions of $I_{T,j}$ for $j=1,\ldots, 4$ should be obvious.

By Lemma \ref{LemmaB.5}, we have $I_{T,1}=O_{P}\left(\sqrt{d_T\log T}\right)$. Also, it's easy to see that $I_{T,2}=O_P(d_T)$ and $I_{T,3}=O_{P}(d_T)$, because $E \|\widetilde{\mathbb{B}}_1(L)\bm{\epsilon}_0 \|<\infty$ and $E\|\widetilde{\mathbb{B}}_T(L)\bm{\epsilon}_T\|<\infty$ in view of the fact that

\begin{eqnarray*}
\|\widetilde{\mathbb{B}}_1(1) \|\leq\sum_{j=0}^{\infty} \| \widetilde{\bm{B}}_{j,1} \|<\infty \quad\text{and}\quad \|\widetilde{\mathbb{B}}_T(1) \|\leq\sum_{j=0}^{\infty} \| \widetilde{\bm{B}}_{j,T} \|<\infty
\end{eqnarray*}
by Lemma \ref{LemmaB.3}. Thus, we just need to consider $I_{T,4}$ below.

Note that

\begin{enumerate}
\item[(1).] $  \sum_{t=1}^{T-1}\|\widetilde{\mathbb{B}}_{t+1}(1)-\widetilde{\mathbb{B}}_{t}(1) \|=O(1) $  by Lemma \ref{LemmaB.3};

\item[(2).]  $ T^{2/\delta}d_T\log T\to 0 $ and $ \sup_{\tau\in[a,b ]}\sum_{t=1}^{T-1}\left\|\bm{W}_{T,t+1}(\tau)-\bm{W}_{T,t}(\tau) \right\|=O\left(d_T\right)$ by the conditions in the body of this theorem;

\item[(3).]  $\max_{1\leq t \leq T-1}\|\widetilde{\mathbb{B}}_{t+1}(L)\bm{\epsilon}_t \|=O_P(T^{1/\delta})$ by $E \| \widetilde{\mathbb{B}}_{t+1}(L)\bm{\epsilon}_t \|^\delta < \infty$ and

\begin{eqnarray*}
\max_{1\leq t \leq T-1}\|\widetilde{\mathbb{B}}_{t+1}(L)\bm{\epsilon}_t \|\leq\left(\sum_{t=1}^{T-1} \|\widetilde{\mathbb{B}}_{t+1}(L)\bm{\epsilon}_t \|^\delta \right)^{1/\delta}=O_P(T^{1/\delta}).
\end{eqnarray*}
\end{enumerate}
Hence, write
\begin{equation*}
\begin{split}
     & \sup_{\tau\in[a,b ]} \left\|\sum_{t=1}^{T-1}\left(\bm{W}_{T,t+1}(\tau)\widetilde{\mathbb{B}}_{t+1}(L)-\bm{W}_{T,t}(\tau)\widetilde{\mathbb{B}}_t(L)\right)\bm{\epsilon}_t\right\| \\
      =\ &\sup_{\tau\in[a,b ]} \left\|\sum_{t=1}^{T-1}(\bm{W}_{T,t+1}(\tau)-\bm{W}_{T,t}(\tau))\widetilde{\mathbb{B}}_{t+1}(L)\bm{\epsilon}_t+\bm{W}_{T,t}(\tau)(\widetilde{\mathbb{B}}_{t+1}(L)-\widetilde{\mathbb{B}}_t(L))\bm{\epsilon}_t\right\| \\
      \leq\ &\max_{1\leq t \leq T-1}\|\widetilde{\mathbb{B}}_{t+1}(L)\bm{\epsilon}_t \| \cdot \sup_{\tau\in[a,b]}\sum_{t=1}^{T-1} \|\bm{W}_{T,t+1}(\tau)-\bm{W}_{T,t}(\tau) \|\\
     &+\sup_{\tau \in [a,b],1\leq t\leq T}\left\|\bm{W}_{T,t}(\tau)\right\| \cdot \sum_{t=1}^{T-1} \|(\widetilde{\mathbb{B}}_{t+1}(L)-\widetilde{\mathbb{B}}_{t}(L))\bm{\epsilon}_t \|\\
     =\ &O_P (T^{1/\delta}\cdot d_T) +   O_P(d_T )=o_P (\sqrt{d_T \log T} ).
\end{split}
\end{equation*}
The first result then follows.

\medskip

\noindent (2). Below, we consider $p=0$ only. The cases with fixed $p\ge 1$ can be verified in a similar manner, so omitted.

\begin{equation*}
  \begin{split}
  &\sup_{ \tau \in [a,b]} \left\| \sum_{t=1}^{T}\mathrm{vec}\left(\bm{W}_{T,t}(\tau)(\bm{x}_t\bm{x}_t^\top-E(\bm{x}_t\bm{x}_t^\top))\right) \right\| \\
 \leq\ & 2\sup_{ \tau \in [a,b]}  \left\|\sum_{t=1}^{T}(\bm{I}_d\otimes \bm{W}_{T,t}(\tau))\sum_{j=0}^{\infty}(\bm{B}_{j,t}\otimes \bm{\mu}_t)\bm{\epsilon}_{t-j} \right\|\\
       &+\sup_{ \tau \in [a,b]}  \left\|\sum_{t=1}^{T}(\bm{I}_d\otimes \bm{W}_{T,t}(\tau))\sum_{j=0}^{\infty}(\bm{B}_{j,t}\otimes \bm{B}_{j,t})\left(\mathrm{vec}\left(\bm{\epsilon}_{t-j}\bm{\epsilon}_{t-j}^\top\right)-\mathrm{vec}\left(\bm{I}_d\right)\right) \right\|\\
       &+2\sup_{ \tau \in [a,b]}  \left\|\sum_{t=1}^{T}(\bm{I}_d\otimes \bm{W}_{T,t}(\tau))\sum_{r=1}^{\infty}\sum_{j=0}^{\infty}(\bm{B}_{j+r,t}\otimes \bm{B}_{j,t})\mathrm{vec}\left(\bm{\epsilon}_{t-j}\epsilon_{t-j-r}\right)\right\|\\
 :=\ &2I_{T,1}+I_{T,2}+2I_{T,3},
  \end{split}
\end{equation*}
wherein $I_{T,1} =O_{P}(\sqrt{d_T \log T})$ by a proof similar to the first result of this theorem.

Consider $I_{T,2}$. Using Lemma \ref{LemmaB.3}, write

\begin{equation*}
  \begin{split}
     I_{T,2}\leq\ & \sup_{\tau\in[a,b]} \left\|\sum_{t=1}^{T}(\bm{I}_d\otimes \bm{W}_{T,t}(\tau))\bm{\mathbb{B}}_t^0(1)\left(\mathrm{vec}\left(\bm{\epsilon}_t\bm{\epsilon}_t^\top\right)-\mathrm{vec}\left(\bm{I}_d\right) \right)\right\|\\
       &+\sup_{\tau\in[a,b]}\left\|(\bm{I}_d\otimes \bm{W}_{T,1}(\tau))\widetilde{\mathbb{B}}_{1}^0(L) \mathrm{vec}\left(\bm{\epsilon}_0 \bm{\epsilon}_0^\top\right)\right\|+\sup_{ \tau \in[a,b]} \left\|(\bm{I}_d\otimes \bm{W}_{T,T}(\tau)) \widetilde{\mathbb{B}}_{T}^0(L) \mathrm{vec}\left(\bm{\epsilon}_T\bm{\epsilon}_T^\top\right)\right\|\\
       &+\sup_{\tau\in[a,b]}  \left\| \sum_{t=1}^{T-1} \left( (\bm{I}_d\otimes \bm{W}_{T,t+1}(\tau))\widetilde{\mathbb{B}}_{t+1}^0(L)-(\bm{I}_d\otimes \bm{W}_{T,t}(\tau))\widetilde{\mathbb{B}}_{t}^0(L)\right)\cdot\mathrm{vec}\left(\bm{\epsilon}_t\bm{\epsilon}_t^\top\right)\right\|\\
      :=\ &I_{T,21}+I_{T,22}+I_{T,23}+I_{T,24}.
  \end{split}
\end{equation*}
By Lemma \ref{LemmaB.6}, we have $I_{T,21}=O_P\left( \sqrt{d_T \log T}\right)$. Also, $I_{T,22}=O_P(d_T)$ and $I_{T,23}=O_P(d_T)$, because $\|\widetilde{\mathbb{B}}_{1}^0(1)\| < \infty$ and $\|\widetilde{\mathbb{B}}_{T}^0(1)\| < \infty$ by Lemma \ref{LemmaB.3}. Similar to the proof of the first result, for $I_{T,24}$, we write

\begin{equation*}
  \begin{split}
       & \sup_{\tau \in [a,b]}  \left\|\sum_{t=1}^{T-1}\left( (\bm{I}_d\otimes \bm{W}_{T,t+1}(\tau)) \widetilde{\mathbb{B}}_{t+1}^0(L)-(\bm{I}_d\otimes \bm{W}_{T,t}(\tau))\widetilde{\mathbb{B}}_{t}^0(L)\right)\mathrm{vec}\left(\bm{\epsilon}_t\bm{\epsilon}_t^\top\right)\right\| \\
     \leq\  &  \sqrt{d} \sup_{\tau \in[a,b], 1\leq t\leq T} \left\|\bm{W}_{T,t+1}(\tau)\right\| \cdot \sum_{t=1}^{T-1}\left\|\left( \widetilde{\mathbb{B}}_{t+1}^0(L)-\widetilde{\mathbb{B}}_{t}^0(L)\right)\mathrm{vec}\left(\bm{\epsilon}_t\bm{\epsilon}_t^\top\right)\right\|\\
       & +\sqrt{d} \max_t \left\|  \widetilde{\mathbb{B}}_{t}^0(L) \mathrm{vec}\left(\bm{\epsilon}_t\bm{\epsilon}_t^\top\right)\right\| \cdot\sup_{\tau \in [a,b]} \sum_{t=1}^{T-1} \left\|\bm{W}_{T,t+1}(\tau)-\bm{W}_{T,t}(\tau)\right\| = o_P\left(\sqrt{d_T \log T}\right),
  \end{split}
\end{equation*}
where we have used the facts that
\begin{eqnarray*}
(1). &&T^{4/\delta}d_T\log T\to 0;\\
(2). &&\max_{t\ge 1} \left\|\widetilde{\mathbb{B}}_{t}^0(L) \mathrm{vec}\left(\bm{\epsilon}_t\bm{\epsilon}_t^\top\right)\right\|=O_P(T^{2/\delta});\\
(3). &&\sup_{\tau\in[a,b ]}\sum_{t=1}^{T-1}\left\|\bm{W}_{T,t+1}(\tau)-\bm{W}_{T,t}(\tau) \right\|=O(d_T);\\
(4). &&\sum_{t=1}^{T-1}\left\|\widetilde{\mathbb{B}}_{t+1}^0(1)-\widetilde{\mathbb{B}}_{t}^0(1)\right\|=O(1).
\end{eqnarray*}
Then we can conclude that $I_{T,24}=O_{P}\left(\sqrt{d_T \log T}\right)$.

We now consider $I_{T,3}$. Using Lemma \ref{LemmaB.3}, we have

\begin{equation*}
  \begin{split}
       & \sup_{\tau \in [a,b]} \left\|\sum_{t=1}^{T}(\bm{I}_d\otimes \bm{W}_{T,t}(\tau))\sum_{r=1}^{\infty}\sum_{j=0}^{\infty}(\bm{B}_{j+r,t}\otimes \bm{B}_{j,t})\mathrm{vec}\left(\bm{\epsilon}_{t-j}\bm{\epsilon}_{t-j-r}\right)\right\|\\
      \leq\ & \sup_{\tau \in [a,b]}  \left\|\sum_{t=1}^{T}(\bm{I}_d\otimes \bm{W}_{T,t}(\tau))\bm{\zeta}_t\bm{\epsilon}_t \right\|+\sup_{\tau \in [a,b]}  \left\|(\bm{I}_d\otimes W_{T,1}(\tau))\sum_{r=1}^{\infty} \widetilde{\mathbb{B}}_{1}^r(L)\mathrm{vec}\left(\bm{\epsilon}_0\bm{\epsilon}_{-r}^\top\right)\right\|\\
       &+\sup_{\tau \in [a,b]}  \left\|(\bm{I}_d\otimes \bm{W}_{T,T}(\tau))\sum_{r=1}^{\infty} \widetilde{\mathbb{B}}_{T}^r(L)\mathrm{vec}\left(\bm{\epsilon}_T\bm{\epsilon}_{T-r}^\top\right)\right\|\\
       &+\sup_{\tau \in [a,b]}  \left\|\sum_{t=1}^{T-1}\sum_{r=1}^{\infty}\left((\bm{I}_d\otimes \bm{W}_{T,t+1}(\tau))\widetilde{\mathbb{B}}_{t+1}^r(L)-(\bm{I}_d\otimes \bm{W}_{T,t}(\tau))\widetilde{\mathbb{B}}_{t}^r(L)\right)\mathrm{vec}\left(\bm{\epsilon}_t\bm{\epsilon}_{t-r}^\top\right)\right\|  \\
     :=\ &I_{T,31}+I_{T,32}+I_{T,33}+I_{T,34},
  \end{split}
\end{equation*}
where $\bm{\zeta}_t$ is defined in Lemma \ref{LemmaB.6}.

By Lemma \ref{LemmaB.6}, $I_{T,31}=O_P\left(\sqrt{d_T \log T}\right)$. Moreover, $I_{T,32}=O_P(d_T)$ and $I_{T,33}=O_P(d_T)$, because $\sum_{r=1}^{\infty}\|\widetilde{\mathbb{B}}_{1}^r(1)\|<\infty$ and $\sum_{r=1}^{\infty}\|\widetilde{\mathbb{B}}_{T}^r(1)\|<\infty$ by Lemma \ref{LemmaB.3}. For $I_{T,34}$, we write

\begin{eqnarray*}
 && \sup_{\tau \in [a,b]}  \left\|\sum_{t=1}^{T-1}\sum_{r=1}^{\infty}\left((\bm{I}_d\otimes \bm{W}_{T,t+1}(\tau))\widetilde{\mathbb{B}}_{t+1}^r(L)-(\bm{I}_d\otimes \bm{W}_{T,t}(\tau))\widetilde{\mathbb{B}}_{t}^r(L)\right)\mathrm{vec}\left(\bm{\epsilon}_t\bm{\epsilon}_{t-r}^\top\right)\right\| \\
    &  \leq  &\sqrt{d}\sup_{\tau \in[a,b],1\leq t\leq T}  \left\|\bm{W}_{T,t}(\tau) \right\| \cdot \sum_{t=1}^{T-1}\left\|\sum_{r=1}^{\infty}\left(\widetilde{\mathbb{B}}_{t+1}^r(L)-\widetilde{\mathbb{B}}_{t}^r(L)\right)\mathrm{vec}\left(\bm{\epsilon}_t\bm{\epsilon}_{t-r}^\top\right)\right\| \\
       && + \sqrt{d} \max_t \left\|\sum_{r=1}^{\infty} \widetilde{\mathbb{B}}_{t}^r(L) \mathrm{vec}\left(\bm{\epsilon}_t\bm{\epsilon}_{t-r}^\top\right)\right\|\cdot\sup_{\tau \in[a,b]}\sum_{t=1}^{T-1} \left\|\bm{W}_{T,t+1}(\tau)-\bm{W}_{T,t}(\tau)\right\| = o_P(\sqrt{d_T \log T}),
\end{eqnarray*}
where we have used the fact that

\begin{eqnarray*}
(1). && T^{4/\delta}d_T\log T\to 0 ;\\
(2). && \max_t \left\|\sum_{r=1}^{\infty} \widetilde{\mathbb{B}}_{t}^r(L) \mathrm{vec}\left(\bm{\epsilon}_t\bm{\epsilon}_{t-r}^\top\right)\right\|=O_P(T^{2/\delta});\\
(3). && \sum_{t=1}^{T-1}\sum_{r=1}^{\infty} \|\widetilde{\mathbb{B}}_{t+1}^r(1)-\widetilde{\mathbb{B}}_{t}^r(1) \|=O(1);\\
(4). && \sup_{\tau\in[a,b ]}\sum_{t=1}^{T-1}\left\|\bm{W}_{T,t+1}(\tau)-\bm{W}_{T,t}(\tau) \right\|=O({d_T}).
\end{eqnarray*}

Based on the above development, the proof of the case with $p=0$ is complete. The proof is now completed.
\end{proof}
\medskip

\begin{proof}[Proof of Theorem \ref{Theorem2.2}]

\item

Write

\begin{equation*}
\begin{split}
&\frac{1}{\sqrt{Th}}\sum_{t=1}^{T}(\bm{x}_t-\bm{\mu}(\tau_t))K\left(\frac{\tau_t-\tau}{h}\right)\\
=\ &\frac{1}{\sqrt{Th}}\sum_{t=1}^{T}\bm{\mathbb{B}}_t(1)\bm{\epsilon}_t K\left(\frac{\tau_t-\tau}{h}\right)+\frac{1}{\sqrt{Th}}\widetilde{\mathbb{B}}_1(L)\bm{\epsilon}_0 K\left(\frac{\tau_1-\tau}{h}\right)-\frac{1}{\sqrt{Th}}\widetilde{\mathbb{B}}_T(L)\bm{\epsilon}_T K\left(\frac{\tau_T-\tau}{h}\right)
 \\
     &+\frac{1}{\sqrt{Th}}\sum_{t=1}^{T-1}\left(\widetilde{\mathbb{B}}_{t+1}(L)K\left(\frac{\tau_{t+1}-\tau}{h}\right)-\widetilde{\mathbb{B}}_{t}(L)K\left(\frac{\tau_t-\tau}{h}\right)\right)\bm{\epsilon}_t\\
    =\ & \frac{1}{\sqrt{Th}}\sum_{t=1}^{T}\bm{\mathbb{B}}_t(1)\bm{\epsilon}_t K\left(\frac{\tau_t-\tau}{h}\right)+o_P(1),
\end{split}
\end{equation*}
where the second equality follows from similar arguments to the proof of Theorem \ref{Theorem2.1}.

For the bias term, we have for any $\tau \in(0,1)$
\begin{eqnarray*}
\frac{1}{Th}\sum_{t=1}^{T}\bm{\mu}(\tau_t)K\left(\frac{\tau_t-\tau}{h}\right)=\bm{\mu}(\tau)+\frac{1}{2}h^2\widetilde{c}_2\bm{\mu}^{(2)}(\tau)+o(h^2)+O\left(\frac{1}{Th}\right).
\end{eqnarray*}

Since
\begin{eqnarray*}
\mathrm{Var}\left( \frac{1}{\sqrt{Th}}\sum_{t=1}^{T}\bm{\mathbb{B}}_t(1)\bm{\epsilon}_t K\left(\frac{\tau_t-\tau}{h}\right) \right)&=& \frac{1}{Th}\sum_{t=1}^{T}\bm{\mathbb{B}}_t(1)\bm{\mathbb{B}}_t^\top(1) K\left(\frac{\tau_t-\tau}{h}\right)^2\\
&\to&\widetilde{v}_0 \left\{ \sum_{j=0}^{\infty}\bm{B}_j(\tau)\right\}\left\{\sum_{j=0}^{\infty}\bm{B}_j^\top(\tau)\right\},
\end{eqnarray*}
we then use the Cram\'er-Wold device to prove its asymptotic normality. That is to show that for any conformable vector $\bm{l}$
\begin{eqnarray*}
\frac{1}{\sqrt{Th}}\sum_{t=1}^{T}\bm{l}^\top\bm{\mathbb{B}}_t(1)\bm{\epsilon}_t K\left(\frac{\tau_t-\tau}{h}\right)\to_D N\left(\bm{0},\widetilde{v}_0 \bm{l}^\top\left\{ \sum_{j=0}^{\infty}\bm{B}_j(\tau)\right\}\left\{\sum_{j=0}^{\infty}\bm{B}_j^\top(\tau)\right\}\bm{l}\right).
\end{eqnarray*}
Let $\bm{Z}_{t}(\tau)=\frac{1}{\sqrt{Th}}\bm{l}^\top \bm{\mathbb{B}}_t(1)\bm{\epsilon}_t K\left(\frac{\tau_t-\tau}{h}\right)$. By the law of large numbers for martingale differences and the assumption $E\left(\bm{\epsilon}_t \bm{\epsilon}_t^\top|\mathcal{F}_{t-1}\right)=\bm{I}_d$ a.s., we have for any $\tau\in(0,1)$
\begin{equation*}
  \sum_{t=1}^{T}\bm{Z}_{t}^2(\tau)\to_P\widetilde{v}_0 \bm{l}^\top\left\{ \sum_{j=0}^{\infty}\bm{B}_j(\tau)\right\}\left\{\sum_{j=0}^{\infty}\bm{B}_j^\top(\tau)\right\}\bm{l}.
\end{equation*}

Furthermore, for any $\nu > 0$ and $\tau\in(0,1)$, by H\"older's inequality and Markov's inequality,
\begin{eqnarray*}
  &&\sum_{t=1}^{T}E\left(\bm{Z}_t^2(\tau) I\left(|\bm{Z}_t(\tau)| > \nu\right) \right)\\
     &\leq & \frac{1}{Th}\sum_{t=1}^{T}K^2\left(\frac{\tau_t-\tau}{h}\right)\left(E|\bm{l}^\top\bm{\mathbb{B}}_t(1)\bm{\epsilon}_t|^\delta \right)^{2/\delta}\left(\frac{E|\bm{l}^\top\bm{\mathbb{B}}_t(1)\bm{\epsilon}_t|^\delta}{(Th)^{\delta/2}\nu^\delta}\right)^{(\delta-2)/\delta} \\
     & = &O\left( \frac{1}{(Th)^{{\delta-2}/2}}\right)=o(1).
\end{eqnarray*}
By Lemma \ref{LemmaB.1}, the proof is now completed.
\end{proof}

\bigskip

\begin{proof}[Proof of Theorem \ref{Theorem3.1}]

\item

\noindent (1). Similar to the proof of Proposition \ref{Proposition3.1}, one can show that $\sum_{t=1}^{T}\left\|\bm{x}_t\bm{x}_t^\top-\widetilde{\bm{x}}_t\widetilde{\bm{x}}_t^\top\right\|=O_P(1)$. Therefore, Lemma \ref{LemmaB.7} are still valid for the time--varying VAR process \eqref{eq9}. For example, consider the uniform convergence results, by Lemma \ref{LemmaB.7},
\begin{eqnarray*}
  &&\sup_{\tau\in[0,1]}\left\|\frac{1}{T}\sum_{t=1}^{T}\left(\bm{x}_t\bm{x}_t-E(\widetilde{\bm{x}}_t\widetilde{\bm{x}}_t^\top)\right)K_h\left( \tau_t-\tau\right)\right\|\\
 &\leq& \sup_{\tau\in[0,1]}\left\|\frac{1}{T}\sum_{t=1}^{T}\left(\widetilde{\bm{x}}_t\widetilde{\bm{x}}_t^\top-E(\widetilde{\bm{x}}_t\widetilde{\bm{x}}_t^\top)\right)K_h\left( \tau_t-\tau\right)\right\|+\left(\frac{1}{T}\sup_{\tau,\tau_t}\left| K_h(\tau_t-\tau)\right| \right) \cdot \sum_{t=1}^{T}\left\|\bm{x}_t\bm{x}_t^\top- \widetilde{\bm{x}}_t\widetilde{\bm{x}}_t^\top \right\|    \\
   &=& O_P\left(\sqrt{\frac{\log T}{Th}} \right)+O_P\left(\frac{1}{Th} \right)= O_P\left(\sqrt{\frac{\log T}{Th}} \right).
\end{eqnarray*}
Hence, in the following we will directly apply Lemma \ref{LemmaB.7} to the time-varying VAR process.

For notational simplicity, let
\begin{eqnarray*}
\bm{S}_{T,k}(\tau) &=&\frac{1}{T}\sum_{t=1}^{T}\bm{z}_{t-1} \bm{z}_{t-1}^{\top}\left(\frac{\tau_t-\tau}{h}\right)^k K_h(\tau_t-\tau) \text{ for } 0\leq k\leq2,\\
\bm{M}(\tau_t) &=& \bm{A}(\tau_t)-\bm{A}(\tau)-\bm{A}^{(1)}(\tau)(\tau_t-\tau)-\frac{1}{2}\bm{A}^{(2)}(\tau)(\tau_t-\tau)^2.
\end{eqnarray*}

We now begin our investigation. Since

\begin{equation*}
\bm{x}_t =\bm{Z}_{t-1}^\top \mathrm{vec}\left(\bm{A}(\tau)+\bm{A}^{(1)}(\tau)(\tau_t-\tau)+\frac{1}{2}\bm{A}^{(2)}(\tau)(\tau_t-\tau)^2+\bm{M}(\tau_t)\right)+\bm{\eta}_t,
\end{equation*}
we write
\begin{equation*}
\begin{split}
    &\mathrm{vec} (\bm{\widehat{A}}(\tau)-\bm{A}(\tau))\\
    =\ &\left(\frac{1}{T}\sum_{t=1}^{T} \bm{Z}_{t-1}\bm{Z}_{t-1}^\top K_h(\tau_t-\tau)\right)^{-1}\left(\frac{1}{T}\sum_{t=1}^{T}\bm{Z}_{t-1}\bm{x}_t  K_h(\tau_t-\tau)\right)-\mathrm{vec}\left(\bm{A}(\tau)\right)\\
    =\ &\left(\bm{S}_{T,0}^{-1}(\tau)\otimes \bm{I}_d\right)\left\{\left(\bm{S}_{T,1}(\tau)\otimes \bm{I}_d\right)h\mathrm{vec}\left(\bm{A}^{(1)}(\tau)\right)+\left( \bm{S}_{T,2}(\tau)\otimes \bm{I}_d\right)\frac{1}{2}h^2\mathrm{vec}\left(\bm{A}^{(2)}(\tau)\right)\right\}\\
     & +\left(\bm{S}_{T,0}^{-1}(\tau)\otimes \bm{I}_d\right)\left(\frac{1}{T}\sum_{t=1}^{T}(\bm{z}_{t-1}\bm{z}_{t-1}^{\top} \otimes \bm{I}_d )\mathrm{vec}\left(\bm{M}(\tau_t)\right)K_h(\tau_t-\tau)\right)\\
     & +\left( \bm{S}_{T,0}^{-1}(\tau)\otimes \bm{I}_d\right)\left(\frac{1}{T}\sum_{t=1}^{T}(\bm{z}_{t-1}\otimes \bm{I}_d)\bm{\eta}_t K_h(\tau_t-\tau)\right) := \, I_{T,1}+I_{T,2}+I_{T,3}.\\
\end{split}
\end{equation*}

By standard arguments for the local constant kernel estimator and the uniform convergence results in Lemma \ref{LemmaB.7}, we have
\begin{eqnarray*}
\|I_{T,1}+I_{T,2}\|=O(h^2)+O_P(h\sqrt{\log T/(Th)})
\end{eqnarray*}
uniformly over $\tau \in [h,1-h]$. By Lemma \ref{LemmaB.8}, we have $I_{T,3}=O_P ( (\frac{\log T}{Th} )^{\frac{1}{2}} )$ uniformly over $\tau \in [0,1]$. Thus, the first result follows.

\medskip

\noindent (2). We begin our investigation on the asymptotic normality by writing

\begin{eqnarray*}
 &&\frac{1}{Th}\sum_{t=1}^{T}\bm{\widehat{\eta}}_t\bm{\widehat{\eta}}_t^\top K\left(\frac{\tau_t-\tau}{h}\right)\\
&=&\frac{1}{Th}\sum_{t=1}^{T}\left(\bm{\eta}_t+\bm{\widehat{\eta}}_t-\bm{\eta}_t\right)\left(\bm{\eta}_t+\bm{\widehat{\eta}}_t-\bm{\eta}_t\right)^\top K\left(\frac{\tau_t-\tau}{h}\right)\\
&=& \frac{1}{Th}\sum_{t=1}^{T}\bm{\eta}_t\bm{\eta}_t^\top K\left(\frac{\tau_t-\tau}{h}\right)+\frac{1}{Th}\sum_{t=1}^{T}(\bm{\widehat{\eta}}_t-\bm{\eta}_t)(\bm{\widehat{\eta}}_t-\bm{\eta}_t)^\top K\left(\frac{\tau_t-\tau}{h}\right) \\
& &+\frac{1}{Th}\sum_{t=1}^{T}\bm{\eta}_t(\bm{\widehat{\eta}}_t-\bm{\eta}_t)^\top K\left(\frac{\tau_t-\tau}{h}\right)+\frac{1}{Th}\sum_{t=1}^{T}(\bm{\widehat{\eta}}_t-\bm{\eta}_t)\bm{\eta}_t^\top K\left(\frac{\tau_t-\tau}{h}\right)\\
&:=& \bm{I}_{T,4}+\bm{I}_{T,5}+\bm{I}_{T,6}+\bm{I}_{T,7}.
\end{eqnarray*}
Let $c_T=h^2+\sqrt{\frac{\log T}{Th}}$. By the first result, for $\forall \tau\in(0,1)$ we have
\begin{eqnarray*}
&&\left\|\frac{1}{Th}\sum_{t=1}^{T}(\bm{\widehat{\eta}}_t-\bm{\eta}_t)(\bm{\widehat{\eta}}_t-\bm{\eta}_t)^\top  K\left(\frac{\tau_t-\tau}{h}\right)\right\|\\
&\leq& \sup_{ \tau_t \in [h,1-h]}\|\bm{\widehat{A}}(\tau_t)-\bm{A}(\tau_t) \|^2 \cdot \frac{1}{Th}\sum_{t=1}^{T}\left\|\bm{z}_{t-1}\right\|^2 K\left(\frac{\tau_t-\tau}{h}\right)=O_P(c_T^2).
\end{eqnarray*}
By Lemma \ref{LemmaB.8}, $\bm{I}_{T,6}$ and $\bm{I}_{T,7}$ are both $o_P((Th)^{-1/2})$. Hence,

\begin{eqnarray*}
\sqrt{Th}\left(\frac{1}{Th}\sum_{t=1}^{T}\bm{\widehat{\eta}}_t\bm{\widehat{\eta}}_t^\top K\left(\frac{\tau_t-\tau}{h}\right)-\frac{1}{Th}\sum_{t=1}^{T}\bm{\eta}_t\bm{\eta}_t^\top K\left(\frac{\tau_t-\tau}{h}\right)-o_P(h^4)\right)=o_P(1).
\end{eqnarray*}

The above development yields that
\begin{eqnarray*}
&&\sqrt{Th}\left(\begin{matrix}
\mathrm{vec}\left(\bm{\widehat{A}}(\tau)-\bm{A}(\tau)-\frac{1}{2}h^2\tilde{c}_2\bm{A}^{(2)}(\tau)\right)-o_P(h^2) \\
\mathrm{vech}\left(\bm{\widehat{\Omega}}(\tau)-\bm{\Omega}(\tau)-\frac{1}{2}h^2\tilde{c}_2\bm{\Omega}^{(2)}(\tau)\right)-o_P(h^2)
                 \end{matrix} \right)\\
&=&\left(\begin{matrix}
           \left( \bm{\Sigma}^{-1}(\tau)\otimes \bm{I}_d\right)\left(\frac{1}{\sqrt{Th}}\sum_{t=1}^{T}\bm{Z}_{t-1} \bm{\eta}_t
           K\left(\frac{\tau_t-\tau}{h}\right)\right) \\
           \frac{1}{\sqrt{Th}}\sum_{t=1}^{T}\mathrm{vech}\left(\bm{\eta}_t\bm{\eta}_t^\top-\bm{\Omega}(\tau_t) \right)K\left(\frac{\tau_t-\tau}{h}\right)
         \end{matrix} \right)+o_P(1)\\
&:=&\bm{I}_{T,8}+o_P(1).
\end{eqnarray*}
Below, we focus on $\bm{I}_{T,8}$. First, show $\text{Var}(\bm{I}_{T,8})\to \bm{V}(\tau)$. Let

\begin{eqnarray*}
\text{Var}(\bm{I}_{T,8})=\left(\begin{matrix}
                                 \widetilde{\bm{V}}_{1,1}(\tau) & \widetilde{\bm{V}}_{2,1}^\top(\tau) \\
                                 \widetilde{\bm{V}}_{2,1}(\tau) & \widetilde{\bm{V}}_{2,2}(\tau)
                               \end{matrix}\right),
\end{eqnarray*}
where the definition of each block should be obvious. Moreover, simple algebra shows that $\widetilde{\bm{V}}_{i,j}(\tau) \to  \bm{V}_{i,j}(\tau)$ for $i,j\in\{1,2\}$.

By construction, $\bm{I}_{T,8}$ is a summation of m.d.s., we thus use Lemma \ref{LemmaB.1} and Cram\'er-Wold device to prove its asymptotic normality. It suffices to show that $\bm{d}^\top \bm{I}_{T,8}\to_D N\left(\bm{0},\bm{d}^\top\bm{V}(\tau)\bm{d}\right)$ for any conformable unit vector $\bm{d}$. Let
\begin{equation*}
  \bm{Z}_{T,t}(\tau)=\frac{1}{\sqrt{Th}}\bm{d}^\top\left(\begin{matrix}
           \left( \bm{\Sigma}^{-1}(\tau)\otimes \bm{I}_d\right)\left(\bm{Z}_{t-1}\bm{\eta}_t
           K\left(\frac{\tau_t-\tau}{h}\right)\right) \\
           \mathrm{vech}\left(\bm{\eta}_t\bm{\eta}_t^\top-\bm{\Omega}(\tau_t) \right)K\left(\frac{\tau_t-\tau}{h}\right)
         \end{matrix} \right).
\end{equation*}
By the law of large numbers for martingale differences, we have $\sum_{t=1}^T \bm{Z}_{T,t}^2(\tau) - \sum_{t=1}^T E(\bm{Z}_{T,t}^2(\tau)|\mathcal{F}_{t-1})\to_P 0$. Since conditional on $\mathcal{F}_{t-1}$ the third and fourth moments of $\bm{\epsilon}_t$ equal to the corresponding unconditional moments a.s., we can prove that $\sum_{t=1}^T E (\bm{Z}_{T,t}^2(\tau)|\mathcal{F}_{t-1} )\to_P \bm{d}^\top \bm{V}(\tau) \bm{d}$.

Furthermore, for any $\nu > 0$ and $\tau\in(0,1)$, similar to the proof of Theorem \ref{Theorem2.2}, we have
  \begin{equation*}
   \sum_{t=1}^{T}E\left(\left(\bm{d}^\top \bm{Z}_{T,t}(\tau)\right)^2I\left(|\bm{d}^\top \bm{Z}_{T,t}(\tau)|>\nu\right)  \right)\to 0.
  \end{equation*}
The result follows by Lemma \ref{LemmaB.1}.

\medskip

(3). By Lemma \ref{LemmaB.7} and the second result of Theorem \ref{Theorem3.1}, we have $\widehat{\bm{V}}_{1,1}(\tau)\to_P \bm{V}_{1,1}(\tau)$. Similar to the proof of the second result of this theorem, by the uniform convergence results of $\widehat{\bm{A}}(\tau)$, we can replace $\widehat{\bm{\eta}}_t$ with $\bm{\eta}_t$ in the following calculations. Therefore,
  \begin{eqnarray*}
     \widehat{\bm{V}}_{2,1}(\tau) &=& \frac{1}{Th}\sum_{t=1}^{T}\mathrm{vech}\left(\bm{\eta}_t\bm{\eta}_t^\top\right)\bm{\eta}_t^\top
    \bm{Z}_{t-1}^\top K^2\left(\frac{\tau_t-\tau}{h}\right)\left( \bm{\Sigma}^{-1}(\tau)\otimes \bm{I}_d\right)+o_P(1)\to_P\bm{V}_{2,1}(\tau),
  \end{eqnarray*}
  and
  \begin{eqnarray*}
    \widehat{\bm{V}}_{2,2}(\tau) &=& \frac{1}{Th}\sum_{t=1}^{T}\mathrm{vech}(\bm{\eta}_t\bm{\eta}_t^\top)\mathrm{vech}(\bm{\eta}_t\bm{\eta}_t^\top)^\top K^2\left(\frac{\tau_t-\tau}{h}\right)\\
    &&-\widetilde{v}_0\mathrm{vech}\left(\bm{\Omega}(\tau)\right)\mathrm{vech}\left(\bm{\Omega}(\tau)\right)^\top+o_P(1)\to_P \bm{V}_{2,2}(\tau).
  \end{eqnarray*}
The proof is now completed.
\end{proof}

\begin{proof}[Proof of Theorem \ref{Theorem3.2}]
\item

We prove that $\lim_{T\to \infty}\Pr\left(\text{IC}(\mathsf{p}) < \text{IC}(p)\right)=0$ for all $\mathsf{p}\neq p$ and $\mathsf{p}\leq \mathsf{P}$.

Note that
$$
\text{IC}(\mathsf{p})-\text{IC}(p)=\log[\text{RSS}(\mathsf{p})/\text{RSS}(p)]+(\mathsf{p}-p)\chi_T.
$$
For $\mathsf{p} < p$, Lemma \ref{LemmaB.9} implies that $\text{RSS}(\mathsf{p})/\text{RSS}(p) > 1 + \nu$ for some $\nu > 0 $ with large probability for all large $T$. Thus, $\log[\text{RSS}(\mathsf{p})/\text{RSS}(p)] \geq \nu/2$ for large $T$. Because $\chi_T\to 0$, we have $\text{IC}(\mathsf{p})-\text{IC}(p)\geq \nu/2-(p-\mathsf{p})\chi_T \geq \nu/3$ for large $T$ with large probability. Thus $\Pr\left(\text{IC}(\mathsf{p}) < \text{IC}(p)\right) \to 0$ for $\mathsf{p} < p$.

Next, consider $\mathsf{p} > p$. Lemma \ref{LemmaB.9} implies that $\log[\text{RSS}(\mathsf{p})/\text{RSS}(p)]=1+O_P(c_T \phi_T)$. Hence, $\log[\text{RSS}(\mathsf{p})/\text{RSS}(p)]=O_P(c_T \phi_T)$. Because $(\mathsf{p}-p) \chi_T \geq \chi_T $, which converges to zero at a slower rate than $c_T \phi_T$, it follows that
$$
\Pr\left(\text{IC}(\mathsf{p}) < \text{IC}(p)\right)\leq \Pr\left(O_P(c_T \phi_T)+ \chi_T < 0\right) \to 0.
$$
The proof is now completed.
\end{proof}

\begin{proof}[Proof of Theorem \ref{Theorem3.3}]

\item

Given the joint distribution of $\mathrm{vec} (\bm{\widehat{A}}(\tau) )$ and $\mathrm{vech} (\bm{\widehat\Omega}(\tau) )$ in Theorem \ref{Theorem3.1}, Theorem \ref{Theorem3.3} can by easily obtained by the Delta method since
$$
\sqrt{Th}\mathrm{vec}\left(\bm{\widehat{B}}_j(\tau)-\bm{B}_j(\tau)\right)=\sqrt{Th}\mathrm{vec}\left(\bm{J} \widehat{\bm{\Phi}}^j(\tau) \bm{J}^\top\widehat{\bm{\omega}}(\tau)-\bm{J} \bm{\Phi}^j(\tau) \bm{J}^\top\bm{\omega}(\tau)\right).
$$
Then, by standard arguments of the Delta method (see \citealp{lutkepohl2005new}, p.111), we can show that
\begin{eqnarray*}
 \sqrt{Th}\left(\mathrm{vec}\left(\bm{\widehat{B}}_j(\tau)-\bm{B}_j(\tau)\right)-\frac{1}{2}h^2\tilde{c}_2\bm{B}_j^{(2)}(\tau)
   \right)\to_D N\left(0,\bm{\Sigma}_{\bm{B}_j}(\tau)\right),
\end{eqnarray*}
where $\bm{\Sigma}_{\bm{B}_j}(\tau)$ has been defined in the body of the theorem.
\end{proof}

\clearpage

\setcounter{page}{1}

\begin{center}
{\large \textbf{Online Supplementary Appendix B to ``A Class of Time--Varying VMA($\infty$) Models: Nonparametric Kernel Estimation and Application"}}
\end{center}

%

This file includes the simulations, preliminary lemmas and proofs which are omitted in the main text. Specifically, the simulations are summarized in Appendix \ref{AppendixB.1}; Appendix \ref{AppendixB.2} presents the preliminary lemmas which are helpful to derive the main results of the paper; Appendix \ref{AppB.3} includes the omitted proofs of the main results;  the proofs of the secondary lemmas are presented in Appendix \ref{AppB.2}.


\section*{Appendix B}

\renewcommand{\theequation}{B.\arabic{equation}}
\renewcommand{\thesection}{B.\arabic{section}}
\renewcommand{\thefigure}{B.\arabic{figure}}
\renewcommand{\thetable}{B.\arabic{table}}
\renewcommand{\thelemma}{B.\arabic{lemma}}

\setcounter{equation}{0}
\setcounter{lemma}{0}
\setcounter{section}{0}
\setcounter{table}{0}
\setcounter{figure}{0}

\section{Simulation}\label{AppendixB.1}

In this section, we exam the above theoretical findings using intensive simulation. The Epanechnikov kernel $K(u)=0.75(1-u^2)I(|u|\leq 1) $ is adopted throughout the numerical studies of this paper for simplicity. For each estimation conducted below, we always select the number of lag by \eqref{eq13} by searching the estimate of $p$ over a sufficiently large range, say $\{1,\ldots, \lfloor \sqrt{Th}\rfloor\}$. Moreover, for each given $\mathsf{p}$, the bandwidth is selected by minimizing the following cross--validation criterion function

\begin{equation}\label{EqB.1}
\text{CV}(h)=\sum_{t=1}^{T}\left(\bm{x}_t-\bm{a}_{-t}(\tau_t)-\sum_{j=1}^{\mathsf{p}}\bm{A}_{j,-t}(\tau_t) \bm{x}_{t-j}\right)^\top\left(\bm{x}_t-\bm{a}_{-t}(\tau_t)-\sum_{j=1}^{\mathsf{p}}\bm{A}_{j,-t}(\tau_t) \bm{x}_{t-j}\right),
\end{equation}
where $\bm{a}_{-t}(\cdot)$, and $\bm{A}_{j,-t}(\cdot)$ are obtained using \eqref{eq10} but leaving the $t$th observation out. Once $\widehat{\mathsf{p}}$ is obtained, the rest calculation is relatively straightforward.

\medskip

We now start describing the data generating process. Let $\bm{\epsilon}_t$ be i.i.d. draws from $N(\bm{0}_{2\times 1}, \bm{I}_2)$ with $E\left(\bm{\epsilon}_t\bm{\epsilon}_t^\top\right)=\bm{I}_2$. Consider
\begin{equation*}
\bm{x}_t=\bm{a}(\tau_t)+\bm{A}_1(\tau_t)\bm{x}_{t-1}+\bm{A}_2(\tau_t)\bm{x}_{t-2}+\bm{\eta}_t,\quad \bm{\eta}_t=\bm{\omega}(\tau_t)\bm{\epsilon}_t\quad\text{with}\quad t=1,\ldots ,T,
\end{equation*}
where
\begin{eqnarray*}
\bm{a}(\tau)&=&[0.5\sin(2\pi \tau),0.5\cos(2\pi \tau)]^\top ,\nonumber \\
\bm{A}_1(\tau)&=&\left(\begin{matrix}
                       0.8\exp{(-0.5+\tau)} & 0.8(\tau-0.5)^3 \\
                       0.8(\tau-0.5)^3 & 0.8+0.3\sin(\pi \tau)
                       \end{matrix} \right), \nonumber \\
\bm{A}_2(\tau)&=&\left(\begin{matrix}
                       -0.2\exp{(-0.5+\tau)} & 0.8(\tau-0.5)^2 \\
                       0.8(\tau-0.5)^2 & -0.4+0.3\cos(\pi \tau)
                       \end{matrix} \right),  \nonumber \\
\bm{\omega}(\tau)&=&\left(\begin{matrix}
                    1.5+0.2\exp{(0.5-\tau)}& 0 \\
                       0.2(1.5+0.5(\tau-0.5)^2)(1.5+0.2\exp{(0.5-\tau)})   & 1.5+0.5(\tau-0.5)^2
                       \end{matrix} \right).
\end{eqnarray*}
We consider the sample size $T\in \{200, 400, 800\}$, and conduct 1000 replications for each choice of $T$.

Based on 1000 replications, we first report the percentages of $\widehat{\mathsf{p}} < 2$, $\widehat{\mathsf{p}} = 2$, and $\widehat{\mathsf{p}} > 2$ respectively in Table \ref{TableB.1} . It can be seen that the information criterion \eqref{eq13} performs reasonably well. The percentage associated to $\widehat{\mathsf{p}}=2$ increases as the sample size goes up.

{
\small
\begin{table}[htbp]
  \centering
  \caption{The percentages of $\widehat{\mathsf{p}} < 2$, $\widehat{\mathsf{p}} = 2$, and $\widehat{\mathsf{p}} > 2$}\label{TableB.1}
    \begin{tabular}{l ccc}
    \hline
    $T$ & $\widehat{\mathsf{p}} < 2$ & $\widehat{\mathsf{p}} = 2$ & $\widehat{\mathsf{p}} > 2$\\
    \hline
   200 &    0.061 & 0.895 & 0.044 \\
   400 &    0.011 & 0.964 & 0.025 \\
   800 &    0.001 & 0.983 & 0.016 \\
   \hline
    \end{tabular}
\end{table}
}

Next, we evaluate the estimates of $\bm{A}(\tau)$ and $\bm{\Omega}(\tau)$, and calculate the root mean square error (RMSE) as follows.
\begin{equation*}
\left\{\frac{1}{nT} \sum_{n=1}^{1000} \sum_{t=1}^{T}\|\widehat{\bm{\theta}}^{(n)}(\tau_t)-\bm{\theta}(\tau_t) \|^2\right\}^{1/2}
\end{equation*}
for $\bm{\theta}(\cdot)\in\left\{\bm{A}(\cdot),\bm{\Omega}(\cdot) \right\}$, where $\widehat{\bm{\theta}}^{(n)}(\tau)$ is the estimate of $\bm{\theta}(\tau)$ for the $n$-th replication. Of interest, we also examine the finite sample coverage probabilities of the confidence intervals based on our asymptotic theories. In the following, we compute the average of coverage probabilities for grid points in $\{\tau_t,t=1,\ldots,T\}$. The RMSEs and empirical coverage probabilities are reported in Table \ref{TableB.2}. As shown in Table \ref{TableB.2}, the RMSE decreases as the sample size goes up. The finite sample coverage probabilities are smaller than their nominal level (95\%) for small $T$, but are fairly close to 95\% as $T$ increases.

{
\small
\begin{table}[htbp]
  \centering
  \caption{The RMSEs and the empirical coverage probabilities with 95\% nominal level.}\label{TableB.2}
    \begin{tabular}{lccccc}
    \hline
        & \multicolumn{2}{c}{$\bm{A}(\tau)$}& & \multicolumn{2}{c}{$\bm{\Omega}(\tau)$}\\
        \cline{2-3} \cline{5-6}
   $ T$     & \text{RMSE} & \text{Coverage rate}& & \text{RMSE} & \text{Coverage rate} \\
    \hline
    200   & 0.4937 & 0.9265& & 0.8228 & 0.8687 \\
    400   & 0.3703 & 0.9291& & 0.7143 & 0.9059 \\
    800   & 0.2785 & 0.9319& & 0.6205 & 0.9226 \\
    \hline
    \end{tabular}
\end{table}
}

\section{Preliminary Lemmas}\label{AppendixB.2}

We present the preliminary lemmas below, which help facilitate the development of the main results.

\begin{lemma}\label{LemmaB.1}
Suppose $\{Z_t,\mathcal{F}_{t}\}$ is a martingale difference sequence, $S_T=\sum_{t=1}^{T}Z_t$, $U_T=\sum_{t=1}^{T}Z_t^2$ and $s_T^2=E(U_T^2)=E(S_T^2)$. If $s_T^{-2}U_T^2 \to_P 1$ and $\sum_{t=1}^{T}E[Z_{T,t}^2I(|Z_{T,t}|>\nu)] \to 0$ for any $\nu > 0$ with $Z_{T,t}=s_T^{-1}Z_t$, then as $T \to \infty$, $ s_T^{-1}S_T \to_D N(0,1).$
\end{lemma}

Lemma \ref{LemmaB.1} can be found in \cite{hall2014martingale}.

\medskip

\begin{lemma}\label{LemmaB.2}
Let $\{Z_t,\mathcal{F}_{t}\}$ be a martingale difference sequence. Suppose that $\left| Z_t\right|\leq M$ for a constant $M$, $t=1,\ldots,T$. Let $V_T=\sum_{t=1}^{T}\text{\normalfont Var}\left(Z_t|\mathcal{F}_{t-1}\right)\leq V$ for some $V>0$. Then for any given $\nu> 0$,
\begin{equation*}
  \Pr\left(\left|\sum_{t=1}^{T}Z_t\right| > \nu\right)\leq \exp\left\{-\frac{\nu^2}{2(V+M\nu)}\right\}.
\end{equation*}
\end{lemma}

Lemma \ref{LemmaB.2} is the Proposition 2.1 of \cite{freedman1975tail}.

\medskip

\begin{lemma}\label{LemmaB.3}
The following algebraic decompositions hold true.

\begin{enumerate}
\item $\bm{\mathbb{B}}_t(L)=\sum_{j=0}^{\infty} \bm{B}_{j,t}L^j$ can be decomposed as $\bm{\mathbb{B}}_t(L)= \bm{\mathbb{B}}_t(1) -(1-L)\widetilde{\mathbb{B}}_t(L)$, where $\widetilde{\mathbb{B}}_t(L) = \sum_{j=0}^{\infty}\bm{\widetilde{B}}_{j,t}L^j$ and $\bm{\widetilde{B}}_{j,t}=\sum_{k=j+1}^{\infty}\bm{B}_{k,t}$.

\item $\bm{\mathbb{B}}_{t}^r(L)=\sum_{j=0}^{\infty}(\bm{B}_{j+r,t}\otimes \bm{B}_{j,t}) L^j$ can be decomposed as $\bm{\mathbb{B}}_{t}^r(L) = \bm{\mathbb{B}}_{t}^r(1)-(1-L)\mathbb{\widetilde{B}}_{t}^r(L)$,
where $\mathbb{\widetilde{B}}_{t}^r(L)=\sum_{j=0}^{\infty}\bm{\widetilde{B}}_{j,t}^rL^j$ and $\bm{\widetilde{B}}_{j,t}^r=\sum_{k=j+1}^{\infty}(\bm{B}_{k+r,t}\otimes \bm{B}_{k,t})$.
\end{enumerate}

In addition, let Assumption \ref{Ass1} hold, then

\begin{enumerate}[resume]
\item $\max_{t\ge 1}\sum_{j=0}^{\infty} \|\bm{\widetilde{B}}_{j,t}\| < \infty$;

\item $\limsup_{T\to\infty}\sum_{t=1}^{T-1} \|\widetilde{\mathbb{B}}_{t+1}(1)-\widetilde{\mathbb{B}}_{t}(1) \| < \infty$;

\item $\max_{t\ge 1} \sum_{j=0}^{\infty} \|\bm{\widetilde{B}}_{j,t}^r\|<\infty$;

\item $\max_{t\ge 1} \sum_{r=1}^{\infty} \|\mathbb{\widetilde{B}}_{t}^r(1) \| <\infty$;

\item $\limsup_{T\to\infty}\sum_{t=1}^{T-1} \sum_{r=0}^{\infty} \|\mathbb{\widetilde{B}}_{t+1}^r(1)-\mathbb{\widetilde{B}}_{t}^r(1) \| < \infty$.
\end{enumerate}
\end{lemma}

\medskip

\begin{lemma}\label{LemmaB.4}
Let Assumptions \ref{Ass1} and \ref{Ass2} hold. Suppose $\left\{\bm{W}_{T,t}\right\}_{t=1}^T$ is  a sequence $m \times d$ deterministic matrices satisfying (1). $\sum_{t=1}^{T} \|\bm{W}_{T,t} \|=O(1)$, (2). $\max_{t\ge 1}\|\bm{W}_{T,t}\|=o(1)$, and (3). $\sum_{t=1}^{T-1}\left\|\bm{W}_{T,t+1}-\bm{W}_{T,t}\right\|=o(1)$. As $T\to \infty$,
\begin{eqnarray*}
\sum_{t=1}^{T}\bm{W}_{T,t}\left(\bm{x}_t-E(\bm{x}_t)\right) \to_P \bm{0} \quad \text{and} \quad \sum_{t=1}^{T}\bm{W}_{T,t}\left(\bm{x}_t\bm{x}_{t+p}^\top-E\left(\bm{x}_t\bm{x}_{t+p}^\top\right)\right)\to_P \bm{0},
\end{eqnarray*}
where $m\ge 1$ is fixed, and $p$ is a fixed non--negative integer.
\end{lemma}

\medskip

\begin{lemma}\label{LemmaB.5}
Let Assumptions \ref{Ass1} and \ref{Ass2} hold, and let $\left\{\bm{W}_{T,t}(\cdot)\right\}_{t=1}^T$ be a sequence of $m \times d$ matrices of functions, where $m\ge 1$ is fixed, and each functional component is Lipschitz continuous and defined on a compact set $[a,b]$. Moreover, suppose that (1).
$\sup_{\tau\in[a,b]}\sum_{t=1}^{T} \|\bm{W}_{T,t}(\tau) \|=O(1)$, and (2). $T^{\frac{2}{\delta}}d_T \log T\to 0$, where $d_T=\sup_{\tau\in[a,b],t\ge 1} \|\bm{W}_{T,t}(\tau) \|$. As $T\to \infty$, $\sup_{\tau\in[a,b]} \left\|\sum_{t=1}^{T}\bm{W}_{T,t}(\tau)\bm{\mathbb{B}}_{t}(1)\bm{\epsilon}_t \right\|=O_{P}\left(\sqrt{d_T\log T} \right)$.
\end{lemma}

\medskip

 \begin{lemma}\label{LemmaB.6}
Let the conditions of Lemma \ref{LemmaB.5} hold. Suppose $T^{\frac{4}{\delta}}d_T\log T \rightarrow 0$, $\max_{t\ge 1} E[\|\bm{\epsilon}_t\|^4|\mathcal{F}_{t-1}]<\infty$ a.s. and
$\sup_{\tau\in[a,b]}\sum_{t=1}^{T-1}\left\|\bm{W}_{T,t+1}(\tau)-\bm{W}_{T,t}(\tau)\right\|=o(1)$.
As $T\to \infty$
\begin{enumerate}
\item $\sup_{\tau\in [a,b]} \left\|\sum_{t=1}^{T}\left(\bm{I}_d\otimes \bm{W}_{T,t}(\tau)\right)\bm{\mathbb{B}}_{t}^0(1)\left(\mathrm{vec}[\bm{\epsilon}_t \bm{\epsilon}_t^\top]- \mathrm{vec}[\bm{I}_d] \right)  \right\|=O_{P}\left(\sqrt{d_T\log T} \right)$;

\item $\sup_{\tau\in [a,b]}\left\|\sum_{t=1}^{T}\left(\bm{I}_d\otimes \bm{W}_{T,t}(\tau)\right)\bm{\zeta}_{t} \bm{\epsilon}_t \right\|=O_{P}\left(\sqrt{d_T\log T}\right)$;
\end{enumerate}
where $\bm{\zeta}_{t} = \sum_{r=1}^{\infty} \sum_{s=0}^{\infty}\{\bm{B}_{s+r,t}\bm{\epsilon}_{t-r}\}\otimes \bm{B}_{s,t}$.
\end{lemma}

\medskip

\begin{lemma}\label{LemmaB.7}
Let Assumptions \ref{Ass2}, \ref{Ass3}, and \ref{Ass4} hold. As $T\to \infty$,
\begin{enumerate}
    \item for $\tau\in(0,1)$,
    \begin{eqnarray*}
    &&  \frac{1}{T}\sum_{t=1}^{T}\bm{x}_t\left(\frac{\tau_t-\tau}{h}\right)^kK_h(\tau_t-\tau)-\widetilde{c}_k\bm{\mu}(\tau)\to_P \bm{0},\\
    &&    \frac{1}{T}\sum_{t=1}^{T}\bm{x}_t\bm{x}_{t+p}^\top\left(\frac{\tau_t-\tau}{h}\right)^k K_h(\tau_t-\tau)-\widetilde{c}_k\bm{\Sigma}_{p}(\tau)\to_P \bm{0},
    \end{eqnarray*}
        where $\bm{\Sigma}_{p}(\tau)=\bm{\mu}(\tau)\bm{\mu}^\top(\tau)+\sum_{j=0}^{\infty}\bm{B}_j(\tau) \bm{B}_{j+p}^\top(\tau)$ for fixed integer $k,p\geq0$;
    \item given $\frac{T^{1-\frac{2}{\delta}}h}{\log T} \rightarrow \infty$,
\begin{equation*}
 \sup_{\tau\in[h, 1-h]} \left\|\frac{1}{T}\sum_{t=1}^{T}\bm{x}_t \left(\frac{\tau_t-\tau}{h}\right)^k K_h(\tau_t-\tau)-\widetilde{c}_k\bm{\mu}(\tau)\right\|=O_{P}\left(h^2+\left(\frac{\log T}{Th} \right)^\frac{1}{2}\right);
\end{equation*}

\item given $\frac{T^{1-\frac{4}{\delta}}h}{\log T} \rightarrow \infty$ and $\max_{t\ge 1} E[\|\bm{\epsilon}_t\|^4 |\mathcal{F}_{t-1} ] < \infty $ a.s.,

\begin{equation*}
 \sup_{\tau\in [h,1-h]} \left\|\frac{1}{T}\sum_{t=1}^{T}\bm{x}_t \bm{x}_{t+p}^\top \left(\frac{\tau_t-\tau}{h}\right)^k K_h(\tau_t-\tau)-\widetilde{c}_k\bm{\Sigma}_{p}(\tau) \right\|=O_{P}\left(h^2+\left(\frac{\log T}{Th}\right)^\frac{1}{2}\right).
\end{equation*}
\end{enumerate}
\end{lemma}

\medskip

\begin{lemma}\label{LemmaB.8}
Let Assumptions \ref{Ass2}, \ref{Ass4} and \ref{Ass5} hold. Suppose $\frac{T^{1-\frac{4}{\delta}}h}{\log T} \rightarrow \infty$ and $\max_{t\ge 1} E[\left\|\bm{\epsilon}_t \right\|^4 |\mathcal{F}_{t-1}]< \infty $ a.s. As $T\to \infty$,
  \begin{enumerate}
    \item $\sup_{\tau \in [0,1]}\left\|\frac{1}{Th}\sum_{t=1}^{T} \bm{Z}_{t-1}\bm{\eta}_t K\left(\frac{\tau_t-\tau}{h} \right)\right\|=O_P\left( \left(\frac{\log T}{Th} \right)^{\frac{1}{2}} \right)$;
    \item $\frac{1}{\sqrt{Th}}\sum_{t=1}^{T}\bm{\eta}_t \left(\bm{\eta}_t-\bm{\widehat{\eta}}_t\right)^\top K\left(\frac{\tau_t-\tau}{h} \right)=o_P(1)$ for $\forall \tau \in (0,1)$.
  \end{enumerate}
\end{lemma}

\medskip

\begin{lemma}\label{LemmaB.9}
Let Assumptions \ref{Ass2}, \ref{Ass4} and \ref{Ass5} hold. Suppose $\frac{T^{1-\frac{4}{\delta}}h}{\log T} \rightarrow \infty$ and $\max_{t\ge 1} E[\left\|\bm{\epsilon}_t \right\|^4 |\mathcal{F}_{t-1}]< \infty $ a.s. As $T\to \infty$,
  \begin{enumerate}
    \item if $\mathsf{p}\geq p$, then $\text{RSS}(\mathsf{p})=\frac{1}{T}\sum_{t=1}^{T}E\left(\bm{\eta}_t^\top\bm{\eta}_t\right) + O_P\left(c_T \phi_T\right)$;
    \item if $\mathsf{p} < p$, then $\text{RSS}(\mathsf{p})=\frac{1}{T}\sum_{t=1}^{T}E\left(\bm{\eta}_t^\top\bm{\eta}_t\right) + c + o_P\left(1\right)$ with some constant $c >0$.
  \end{enumerate}
\end{lemma}

\section{Omitted Proofs of the Main Results}\label{AppB.3}

We present the omitted proofs of the main results in this section.

\begin{proof}[Proof of Proposition \ref{Proposition2.1}]
\item

\noindent (1). Start from Example \ref{Example1}. Let $\rho$ denote the largest eigenvalue of $\bm{\Phi}_t$ uniformly over $t$. Then, $\rho < 1$ by the condition in Proposition \ref{Proposition2.1}. Similar to the proof of Proposition 2.4 in \cite{dahlhaus2009empirical}, we have $\max_{t\ge 1} \|\prod_{i=0}^{j-1}\bm{\Phi}_{t-i}\|\leq M \rho^j$, which yields that

\begin{eqnarray*}
\max_{t\geq1} \sum_{j=1}^{\infty}j \left\|\bm{B}_{j,t}\right\|=\max_{t\geq1} \sum_{j=1}^{\infty}j\left\| \bm{J}\prod_{i=0}^{j-1}\bm{\Phi}_{t-i}\bm{J}^\top\right\|
\leq M \sum_{j=1}^{\infty}j \rho^j=O(1).
\end{eqnarray*}
In addition, for any conformable matrices $\{\bm{A}_i\}$ and $\{\bm{B}_i\}$, since

\begin{equation*}
\prod_{i=1}^{r}\bm{A}_i-\prod_{i=1}^{r}\bm{B}_i=\sum_{j=1}^{r}\left( \prod_{k=1}^{j-1}\bm{A}_k \right)\left(\bm{A}_j-\bm{B}_j\right)\left(\prod_{k=j+1}^{r}\bm{B}_k\right),
\end{equation*}
we then obtain that

\begin{eqnarray*}
&&\limsup_{T\to\infty} \sum_{t=1}^{T-1} \sum_{j=1}^{\infty}j \left\|\bm{J}\left(\prod_{i=0}^{j-1}\bm{\Phi}_{t+1-i}- \prod_{i=0}^{j-1}\bm{\Phi}_{t-i}\right)\bm{J}^\top\right\| \\
&=& \limsup_{T\to\infty} \sum_{t=1}^{T-1} \sum_{j=1}^{\infty}j \left\|\bm{J}\sum_{m=1}^{j}\left(\prod_{k=1}^{m-1}\bm{\Phi}_{t+2-k}\right)\left(\bm{\Phi}_{t+2-m}-\bm{\Phi}_{t+1-m} \right)\left(\prod_{k=m}^{j}\bm{\Phi}_{t+1-k}\right)\bm{J}^\top\right\| \\
&\leq&M \sum_{j=1}^{\infty}j^2 \rho^{j-1}\limsup_{T\to\infty} \sum_{t=1}^{T-1}\left\|\bm{\Phi}_{t+1}-\bm{\Phi}_{t}\right\|=O(1)
\end{eqnarray*}
given the condition in Proposition \ref{Proposition2.1}.

\medskip

Consider Example \ref{Example2}. Similar to Example \ref{Example1},

\begin{eqnarray*}
\max_{t\geq 1}\sum_{b=1}^{\infty}b\left\|\bm{D}_{b,t}\right\|\leq M\max_{t\geq 1}\sum_{b=1}^{\infty}b\sum_{j=b-q}^{b}\left\|\bm{B}_{j,t}\right\|\leq M\sum_{b=1}^{\infty}b \rho^b=O(1).
\end{eqnarray*}
In addition,

\begin{eqnarray*}
    && \limsup_{T\to \infty} \sum_{t=1}^{T-1}\sum_{b=1}^{\infty}b\left\|\bm{D}_{b,t+1}-\bm{D}_{b,t}\right\|\\
    &\leq&  \limsup_{T\to \infty}  \sum_{t=1}^{T-1}\sum_{b=1}^{\infty}b\sum_{j=\max(0,b-q)}^{b}\left\|\bm{B}_{j,t+1}-\bm{B}_{j,t}\right\|\left\| \bm{\Theta}_{b-j,t+1-j}\right\|\\
    &&+\limsup_{T\to \infty}  \sum_{t=1}^{T-1}\sum_{b=1}^{\infty}b\sum_{j=\max(0,b-q)}^{b}\left\|\bm{B}_{j,t}\right\|\left\| \bm{\Theta}_{b-j,t+1-j}-\bm{\Theta}_{b-j,t-j}\right\|\\
    &\leq& M \limsup_{T\to \infty}  \sum_{t=1}^{T-1}\sum_{b=1}^{\infty}b\left\|\bm{B}_{b,t+1}-\bm{B}_{b,t}\right\|\\
    &&+\sum_{b=1}^{\infty}b\sum_{j=\max(0,b-q)}^{b}\left\|\bm{B}_{j,t}\right\| \limsup_{T\to \infty}  \sum_{t=1}^{T-1}\left\| \bm{\Theta}_{b-j,t+1-j}-\bm{\Theta}_{b-j,t-j}\right\| =O(1).
\end{eqnarray*}

\medskip

\noindent (2). By the condition of Proposition \ref{Proposition2.1}.3,

\begin{eqnarray*}
  &&      \max_{t\geq 1}\sum_{j=1}^{\infty}j\left\|\bm{B}_{j,t}\right\|\leq \max_{t\geq 1}\sum_{j=1}^{\infty}j \sum_{l=0}^{j}\left\|\bm{\Psi}_{l,t}\right\|\left\|\bm{\Theta}_{j-l,t-l}\right\|\\
        &=&\max_{t\geq 1}\sum_{j=0}^{\infty}\sum_{l=j+1}^{\infty}l\left\|\bm{\Psi}_{j,t}\right\|\left\|\bm{\Theta}_{l-j,t-j}\right\|=\max_{t\geq 1}\sum_{j=0}^{\infty}\left\|\bm{\Psi}_{j,t}\right\|\sum_{l=1}^{\infty}(l+j)\left\|\bm{\Theta}_{l,t-j}\right\|\\
        &=&\max_{t\geq 1}\sum_{j=0}^{\infty}j\left\|\bm{\Psi}_{j,t}\right\|\sum_{l=1}^{\infty}\left\|\bm{\Theta}_{l,t-j}\right\| +\max_{t\geq 1}\sum_{j=0}^{\infty}\left\|\bm{\Psi}_{j,t}\right\|\sum_{l=1}^{\infty}l\left\|\bm{\Theta}_{l,t-j}\right\|=O(1).
     \end{eqnarray*}
In addition,

\begin{eqnarray*}
\limsup_{T\to \infty}\sum_{t=1}^{T-1}\sum_{j=1}^{\infty}j\left\|\bm{B}_{j,t+1}-\bm{B}_{j,t}\right\| &\leq& \limsup_{T\to \infty}\sum_{t=1}^{T-1}\sum_{j=1}^{\infty}j\sum_{l=0}^{j}\left\| \bm{\Psi}_{l,t+1} \right\|\left\|\bm{\Theta}_{j-l,t+1-l}-\bm{\Theta}_{j-l,t-l}\right\| \\
  &&+ \limsup_{T\to \infty}\sum_{t=1}^{T-1}\sum_{j=1}^{\infty}j\sum_{l=0}^{j}\left\| \bm{\Psi}_{l,t+1} -\bm{\Psi}_{l,t} \right\|\left\|\bm{\Theta}_{j-l,t-l}\right\| \\
  &:=&I_{T,1}+I_{T,2}.
\end{eqnarray*}
We show that $I_{T,1}$ is bounded below, and the proof of $I_{T,2}$ can be established similarly.

\begin{eqnarray*}
   I_{T,1} &\leq& \max_{t\geq 1}\sum_{j=0}^{\infty}j\left\|\bm{\Psi}_{j,t}\right\|\limsup_{T\to \infty}\sum_{t=1}^{T-1}\sum_{l=1}^{\infty}\left\|\bm{\Theta}_{l,t+1-j}-\bm{\Theta}_{l,t-j}\right\| \\
  &&+\max_{t\geq 1}\sum_{j=0}^{\infty}\left\|\bm{\Psi}_{j,t}\right\|\limsup_{T\to \infty}\sum_{t=1}^{T-1}\sum_{l=1}^{\infty}l\left\|\bm{\Theta}_{l,t+1-j}-\bm{\Theta}_{l,t-j}\right\| =O(1).
\end{eqnarray*}
The proof is now completed.
\end{proof}
\medskip

\begin{proof}[Proof of Proposition \ref{Proposition3.1}]

\item
Consider the VMA representation of $\bm{x}_t$:
\begin{eqnarray*}
\bm{x}_t= \bm{\mu}_{t} +\bm{B}_{0,t}\bm{\epsilon}_t+\bm{B}_{1,t}\bm{\epsilon}_{t-1}+\bm{B}_{2,t}\bm{\epsilon}_{t-2}+\cdots,
\end{eqnarray*}
where  $\bm{B}_{0,t}=\bm{\omega}(\tau_t)$, $\bm{B}_{j,t} =\bm{\Psi}_{j,t}\bm{\omega}(\tau_{t-j})$,
$\bm{\Psi}_{j,t}=\bm{J}\prod_{m=0}^{j-1}\bm{\Phi}(\tau_{t-m}) \bm{J}^\top$ for $j\geq 1$, $\bm{\mu}_t=\bm{a}(\tau_t)+\sum_{j=1}^{\infty}\bm{\Psi}_{j,t} \bm{a}(\tau_{t-j})$ and $\tau_{t-j}=\frac{t-j}{T}I(t\geq j)$.

\noindent First, we investigate the validity of VMA representation of $\bm{x}_t$ and $\widetilde{\bm{x}}_t$. Let $\rho_A$ denote the largest eigenvalue of $\bm{\Phi}(\tau)$ uniformly over $\tau\in [0,1]$. Then, $\rho_A<1$ by Assumption \ref{Ass5}.1.

\noindent Similar to the proof of Proposition 2.4 in \cite{dahlhaus2009empirical}, we have $\max_{t\geq 1} \left\|\prod_{m=0}^{j-1}\bm{\Phi}(\tau_{t-m}) \right\| \leq M \rho_A^j$. It follows that $\left\|{\rm E}(\bm{x_t})\right\| \leq \sum_{j=0}^{\infty} \left\|\bm{\Psi}_{j,t} \right\|\cdot \left\|\bm{a}(\tau_{t-j})\right\| \leq M \sum_{j=0}^{\infty} \rho_A^j < \infty$ and
\begin{eqnarray*}
  \left\|{\rm Var}(\bm{x_t})\right\| &=& \left\|\sum_{j=0}^{\infty}\bm{B}_{j,t}\bm{B}_{j,t}^\top\right\| \leq \sum_{j=0}^{\infty}\left\|\bm{B}_{j,t}\right\|^2 \leq M \sum_{j=0}^{\infty}\rho_A^{2j} < \infty.
\end{eqnarray*}
Similarly, we have $\left\|{\rm E}\left(\widetilde{\bm{x}}_t \right)\right\| < \infty$ and $\left\|{\rm Var}\left(\widetilde{\bm{x}}_t \right)\right\| < \infty$.

Then, we need to verify that $\max_{t\geq 1}E\left\|\bm{x}_t-\widetilde{\bm{x}}_t \right\|=O(T^{-1})$. For any conformable matrices $\{\bm{A}_i\}$ and $\{\bm{B}_i\}$, since
\begin{equation*}
  \prod_{i=1}^{r}\bm{A}_i-\prod_{i=1}^{r}\bm{B}_i=\sum_{j=1}^{r}\left( \prod_{k=1}^{j-1}\bm{A}_k \right)\left(\bm{A}_j-\bm{B}_j\right)\left(\prod_{k=j+1}^{r}\bm{B}_k\right)
\end{equation*}
we have
\begin{eqnarray*}
  &&\left\|\bm{B}_{j,t} -\bm{B}_j(\tau_t)\right\| = \left\|\bm{J}\prod_{m=0}^{j-1}\bm{\Phi}(\tau_{t-m})\bm{J}^\top\bm{\omega}(\tau_{t-j})-\bm{J}\bm{\Phi}^j(\tau_t)\bm{J}^\top\bm{\omega}(\tau_t)\right\|\\
   &=& \left\| \left(\bm{J}\prod_{m=0}^{j-1}\bm{\Phi}(\tau_{t-m})\bm{J}^\top-\bm{J}\bm{\Phi}^j(\tau_t)\bm{J}^\top\right)\bm{\omega}(\tau_t)+\bm{J}\prod_{m=0}^{j-1}\bm{\Phi}(\tau_{t-m})\bm{J}^\top\left( \bm{\omega}(\tau_{t-j})-\bm{\omega}(\tau_t)\right)\right\|\\
   &\leq&M\sum_{i=1}^{j-1}\left\|\bm{\Phi}^i(\tau_t)(\bm{\Phi}(\tau_{t-i})-\bm{\Phi}(\tau_t))\prod_{m=i+1}^{j-1}\bm{\Phi}(\tau_{t-m}) \right\|+M \rho_A^j \frac{j}{T}\\
   &\leq&M\sum_{i=1}^{j-1}\frac{i}{T}\rho_A^{j-1}+M \rho_A^j \frac{j}{T}=O(T^{-1}),
\end{eqnarray*}
which follows that
\begin{eqnarray*}
  E\left\|\bm{x}_t-\widetilde{\bm{x}}_t \right\| &\leq &\sum_{j=1}^{\infty}\left\|\bm{\Psi}_{j,t}\bm{a}(\tau_{t-j})-\bm{\Psi}_j(\tau_t)\bm{a}(\tau_t)\right\|+ \sum_{j=1}^{\infty}\left\|\bm{B}_{j,t} -\bm{B}_j(\tau_t)\right\| \cdot E\left\|\bm{\epsilon}_t\right\|\\
   &\leq& M\sum_{j=1}^{\infty} \left(\sum_{i=1}^{j-1}\frac{i}{T}\rho_A^{j-1}+ \rho_A^j \frac{j}{T}\right) =O\left(T^{-1}\right).
\end{eqnarray*}

Finally, we check whether the MA coefficients of $\widetilde{x}_t$ satisfy Assumption \ref{Ass3}. For $\bm{\mu}(\tau)$, the series $\sum_{j=0}^{\infty} \bm{\Psi}_j(\tau) \bm{a}(\tau)$ converges uniformly on $[0,1]\in R$ since for every $\nu > 0$ there exists an $N_\nu >0$ such that
\begin{equation*}
  \left\| \bm{\Psi}_{m+1}(\tau) \bm{a}(\tau)+\cdots+\bm{\Psi}_{n}(\tau) \bm{a}(\tau)\right\| \leq M \sum_{j=m+1}^{n} \rho_A^j < \nu
\end{equation*}
whenever $n>m>N_\nu$.

By the term-by-term differentiability theorem, we have $\bm{\mu}^{(1)}(\tau)=\sum_{j=0}^{\infty}\left(\bm{\Psi}_j^{(1)}(\tau) \bm{a}(\tau)+ \bm{\Psi}_j(\tau) \bm{a}^{(1)}(\tau)\right)$,  where $\bm{\Psi}_j^{(1)}(\tau)=\bm{J}\left( \sum_{i=0}^{j-1}\bm{\Phi}^i(\tau)\bm{\Phi}^{(1)}(\tau)\bm{\Phi}^{j-1-i}(\tau)\right)\bm{J}^\top$.

Therefore, we can conclude that $\bm{\mu}(\cdot)$ and $\bm{\Psi}_j(\cdot)$ is first-order continuously differentiable. Similarly, we can show the second-order continuously differentiability of $\bm{\mu}(\cdot)$ and $\bm{\Psi}_j(\cdot)$.

In addition, since $\sup_{\tau\in[0,1]} \left\|\bm{\Phi}^j(\tau) \right\| \leq M \rho_A^j$, we have
\begin{equation*}
  \sum_{j=0}^{\infty}j\left\|\bm{B}_j(\tau)\right\|=\sum_{j=0}^{\infty}j\left\|\bm{\Psi}_j(\tau) \bm{\omega}(\tau)\right\|\leq M \sum_{j=0}^{\infty}j \rho_A^j < \infty,
\end{equation*}
and
\begin{eqnarray*}
\sum_{j=0}^{\infty}j\left\|\bm{B}_j^{(1)}(\tau)\right\| &=& \sum_{j=0}^{\infty}j\left\|\bm{\Psi}_j^{(1)}(\tau)\bm{\omega}(\tau)+\bm{\Psi}_j(\tau)\bm{\omega}^{(1)}(\tau)\right\|
  \leq M\sum_{j=0}^{\infty}(j^2\rho_A^{j-1} + j\rho_A^{j})< \infty.
  \nonumber
\end{eqnarray*}

The proof is therefore completed.

\end{proof}
\medskip

\begin{proof}[Proof of Theorem \ref{TheoremA.1}]

\item

Note that $\bm{x}_t^*=\widetilde{\bm{\mu}}(\tau_t)+\bm{e}_t^*$, so we can write

\begin{equation*}
\begin{split}
\widehat{\bm{\mu}}^*(\tau)-\widetilde{\bm{\mu}}(\tau)= & \left(\sum_{t=1}^{T}{W}_{T,t}(\tau)\bm{\widetilde{\mu}}(\tau_t)-\bm{\widetilde{\mu}}(\tau)\right)+ \sum_{t=1}^{T}W_{T,t}(\tau)\bm{e}_t^* \\
:= & \bm{I}_{T,1}+ \bm{I}_{T,2},
\end{split}
\end{equation*}
where $W_{T,t}(\tau)=K(\frac{\tau_t-\tau}{h})/ \sum_{t=1}^{T}K(\frac{\tau_t-\tau}{h})$.

We start our investigation from $ \bm{I}_{T,1}$, and write

\begin{eqnarray*}
 \bm{I}_{T,1} &= & \left(\frac{1}{Th}\sum_{t=1}^{T}K\left(\frac{\tau_t-\tau}{h}\right) \frac{1}{T\widetilde{h}}\sum_{s=1}^{T}\bm{\mu}(\tau_s)K\left(\frac{\tau_s-\tau_t}{\widetilde{h}}\right)-\frac{1}{T\widetilde{h}}\sum_{s=1}^{T}\bm{\mu}(\tau_s)K\left(\frac{\tau_s-\tau}{\widetilde{h}}\right) \right)\\
&&+\frac{1}{\sqrt{T\widetilde{h}}}\left(\sum_{t=1}^{T}K\left(\frac{\tau_t-\tau}{h}\right)\bm{Z}_T(\tau_t)-\bm{Z}_T(\tau)\right)+O_P\left(\frac{1}{Th} \right) \\
&:= &\bm{I}_{T,11}+ \bm{I}_{T,12}+O_P\left(\frac{1}{Th}\right),
\end{eqnarray*}
where the definitions of $ \bm{I}_{T,11}$ and $ \bm{I}_{T,12}$ should be obvious, $\bm{Z}_T(\tau)=\frac{1}{\sqrt{T\widetilde{h}}}\sum_{t=1}^{T}\bm{e}_tK\left(\frac{\tau_t-\tau}{\widetilde{h}}\right)$ and $\bm{e}_t=\sum_{j=0}^{\infty}\bm{B}_j(\tau_t)\bm{\epsilon}_{t-j}$. Similar to the development of Lemma \ref{LemmaB.4}, we can show that $\| \bm{I}_{T,12}\| =O_P((T\widetilde{h})^{-1/2} )$, which in connection of Assumption \ref{AssA.1} yields

\begin{eqnarray*}
\sqrt{Th}\| \bm{I}_{T,12}\| =O_P ((h/\widetilde{h})^{1/2} )=o_P(1).
\end{eqnarray*}
For $ \bm{I}_{T,11}$, by the definition of Riemann integral, we have

\begin{equation*}
\begin{split}
\bm{I}_{T,11}=\ & \int_{-1}^{1}K(u)\int_{-1}^{1}K(v)\left(\bm{\mu}(\tau+v\widetilde{h}+uh)-\bm{\mu}(\tau+v\widetilde{h})\right)dv du +O\left(\frac{1}{Th}\right)\\
=\ & \frac{1}{2}h^2\widetilde{c}_2\bm{\mu}^{(2)}(\tau)+O(h^2(h+\widetilde{h}))+O\left(\frac{1}{Th}\right).
\end{split}
\end{equation*}

Thus, we just need to focus on $ \bm{I}_{T,2}$, and show that

\begin{equation*}
\frac{1}{\sqrt{Th}}\sum_{t=1}^{T}K\left(\frac{\tau_t-\tau}{h}\right)\bm{e}_t^*\to_{D^*} N\left(\bm{0},\tilde{v}_0\left\{ \sum_{j=0}^{\infty}\bm{B}_j(\tau)\right\}\left\{\sum_{j=0}^{\infty}\bm{B}_j^\top(\tau)\right\}\right).
\end{equation*}
Using the Cram\'er-Wold device, this is enough to show for any conformable unit vector $\bm{d} $,
\begin{equation*}
\frac{1}{\sqrt{Th}}\sum_{t=1}^{T}K\left(\frac{\tau_t-\tau}{h}\right)\bm{d}^\top \bm{e}_t^*\to_{D^*} N\left(\bm{0},\tilde{v}_0\bm{d}^\top\left\{ \sum_{j=0}^{\infty}\bm{B}_j(\tau)\right\}\left\{\sum_{j=0}^{\infty}\bm{B}_j^\top(\tau)\right\}\bm{d}\right).
\end{equation*}

For $\forall \tau \in [h+\widetilde{h},1-h-\widetilde{h}]$, we write

\begin{eqnarray*}
\frac{1}{\sqrt{Th}}\sum_{t=1}^{T}K\left(\frac{\tau_t-\tau}{h}\right)\bm{d}^\top\bm{\widehat{e}}_t\xi_t^* & = & {Z}_T^*(\tau)+\frac{1}{\sqrt{Th}}\sum_{t=1}^{T}K\left(\frac{\tau_t-\tau}{h}\right)\bm{d}^\top(\bm{\widehat{e}}_t-\bm{e}_t)\xi_t^* \\
& = &{Z}_T^*(\tau)+o_P(1),
\end{eqnarray*}
where ${Z}_T^*(\tau)=\frac{1}{\sqrt{Th}}\sum_{t=1}^{T}K\left(\frac{\tau_t-\tau}{h}\right)\bm{d}^\top\bm{e}_t\xi_t^*$, and the second equality follows from
\begin{equation*}
\begin{split}
&EE^*\left\|\frac{1}{\sqrt{Th}}\sum_{t=1}^{T}K\left(\frac{\tau_t-\tau}{h}\right)\bm{d}^\top(\bm{\widehat{e}}_t-\bm{e}_t)\xi_t^*\right\|^2\\
\leq\ &\max_{\lfloor T(\tau-h)\rfloor \leq t \leq \lceil T(\tau+h)\rceil}E\left\|\bm{\widehat{e}}_t-\bm{e}_t\right\|^2\left(\frac{1}{Th}\sum_{t=1}^{T}\sum_{s=1}^{T}K\left(\frac{\tau_t-\tau}{h}\right)K\left(\frac{\tau_s-\tau}{h}\right)E^*(\xi_t^*\xi_s^*)\right)
 \\
 =\ & O\left(\widetilde{h}^4+1/(T\widetilde{h})\right)O(l)= o\left(1\right),
\end{split}
\end{equation*}
where $EE^*[\cdot]$ stands for taking expectation of the variables with respect to the bootstrap draws first, and then taking the exception with respect to the sample.

\medskip

In the following, we first show that

\begin{eqnarray*}
 \mathrm{Var}^*({Z}_T^*(\tau))^2=\tilde{v}_0\bm{d}^\top\left\{ \sum_{j=0}^{\infty}\bm{B}_j(\tau)\right\}\left\{\sum_{j=0}^{\infty}\bm{B}_j^\top(\tau)\right\}\bm{d}+o_P(1)
\end{eqnarray*}
and then prove its normality by blocking techniques.

Condition on the original sample, we have
\begin{eqnarray}\label{EqB.2}
&&E^*({Z}_T^*(\tau))^2\nonumber\\
& =\ &\frac{1}{Th}\sum_{t=1}^{T}\sum_{s=1}^{T}K\left(\frac{\tau_t-\tau}{h}\right)K\left(\frac{\tau_s-\tau}{h}\right)\bm{d}^\top \bm{e}_t\bm{e}_s^\top \bm{d} E^*(\xi_t^* \xi_s^*)\nonumber\\
&=&\frac{1}{Th}\sum_{t=1}^{T}K^2\left(\frac{\tau_t-\tau}{h}\right)\bm{d}^\top \bm{e}_t\bm{e}_{t}^\top \bm{d}+ \frac{1}{Th}\sum_{i=1}^{T-1}\sum_{t=1}^{T-i}K\left(\frac{\tau_t-\tau}{h}\right)K\left(\frac{\tau_{t+i}-\tau}{h}\right)\bm{d}^\top \bm{e}_t\bm{e}_{t+i}^\top \bm{d} a(i/l)\nonumber\\
&&+\frac{1}{Th}\sum_{i=1}^{T-1}\sum_{t=1}^{T-i}K\left(\frac{\tau_t-\tau}{h}\right)K\left(\frac{\tau_{t+i}-\tau}{h}\right)\bm{d}^\top \bm{e}_{t+i}\bm{e}_{t}^\top \bm{d} a(i/l).
\end{eqnarray}
For the first term on the right hand side of \eqref{EqB.2}, by Lemma \ref{LemmaB.4}, it is straightforward to obtain that

\begin{equation*}
\frac{1}{Th}\sum_{t=1}^{T}K^2\left(\frac{\tau_t-\tau}{h}\right) \bm{d}^\top \bm{e}_t\bm{e}_{t}^\top \bm{d}=\frac{1}{Th}\sum_{t=1}^{T}K^2\left(\frac{\tau_t-\tau}{h}\right)\bm{d}^\top E\left(\bm{e}_t\bm{e}_{t}^\top\right)\bm{d} +o_P(1).
\end{equation*}
For the second and third terms on the right hand side of \eqref{EqB.2}, we have

\begin{eqnarray*}
 &&E\left\| \frac{1}{Th}\sum_{i=1}^{T-1}\sum_{t=1}^{T-i}K\left(\frac{\tau_t-\tau}{h}\right)K\left(\frac{\tau_{t+i}-\tau}{h}\right)\left(\bm{e}_t\bm{e}_{t+i}^\top-E(\bm{e}_t\bm{e}_{t+i}^\top)\right) a(i/l)\right\|\\
 &\le &\frac{1}{Th}\sum_{i=1}^{T-1}a(i/l) E\left\| \sum_{t=1}^{T-i}K\left(\frac{\tau_t-\tau}{h}\right)K\left(\frac{\tau_{t+i}-\tau}{h}\right)\left(\bm{e}_t\bm{e}_{t+i}^\top-E(\bm{e}_t\bm{e}_{t+i}^\top)\right) \right\|.
\end{eqnarray*}

We now take a careful look at $E\left\| \sum_{t=1}^{T-i}K\left(\frac{\tau_t-\tau}{h}\right)K\left(\frac{\tau_{t+i}-\tau}{h}\right)\left(\bm{e}_t\bm{e}_{t+i}^\top-E(\bm{e}_t\bm{e}_{t+i}^\top)\right) \right\|$. Simple algebra shows that $E(\bm{e}_t \bm{e}_{t+i}^\top)=\sum_{j=0}^{\infty}\bm{B}_{j,t} \bm{B}_{j+i,t+i}^\top$. Applying vector transformation, we can write

\begin{eqnarray*}
&&\frac{1}{Th}\sum_{t=1}^{T-i}K\left(\frac{\tau_t -\tau}{h} \right) K\left(\frac{\tau_{t+i} -\tau}{h} \right)\mathrm{vec}\left[\bm{e}_t \bm{e}_{t+i}^\top-E(\bm{e}_t \bm{e}_{t+i}^\top)\right]\\
&=& \frac{1}{Th}\sum_{t=1}^{T-i}K\left(\frac{\tau_t -\tau}{h} \right) K\left(\frac{\tau_{t+i} -\tau}{h} \right) \sum_{j=0}^{\infty} \left(\bm{B}_{j+i,t+i}\otimes \bm{B}_{j,t}\right) \mathrm{vec}\left[\bm{\epsilon}_{t-j}\bm{\epsilon}_{t-j}^\top- \bm{I}_d\right]\\
& & + \frac{1}{Th}\sum_{t=1}^{T-i}K\left(\frac{\tau_t -\tau}{h} \right) K\left(\frac{\tau_{t+i} -\tau}{h} \right) \sum_{j=0}^{\infty}\sum_{m=0,\neq j+i}^{\infty}\left(\bm{B}_{m,t+i}\otimes \bm{B}_{j,t}\right) \mathrm{vec}\left[\bm{\epsilon}_{t-j}\bm{\epsilon}_{t+i-m}^\top\right]\\
&:= & \bm{I}_{T,3}+ \bm{I}_{T,4}.
\end{eqnarray*}
Let $\bm{w}_t=\mathrm{vec}\left[\bm{\epsilon}_t\bm{\epsilon}_t^\top-\bm{I}_d\right]$. By Assumption \ref{Ass2}, $E\left\|\bm{w}_t\right\|^{\delta/2}\leq 4 E\left\|\bm{\epsilon}_t\right\|^{\delta}< \infty$ for some $\delta>2$, which implies that $\{\bm{w}_t\}$ is uniformly integrable. Hence, for every $\nu>0$, there exists a $\lambda_\nu > 0$ such that $E\left\|\bm{w}_t I\left(\left\|\bm{w}_t\right\|> \lambda_\nu\right) \right\| < \nu$. Define $\bm{w}_{1,t}=\bm{w}_t I\left(\left\|\bm{w}_t\right\|\leq \lambda_\nu\right)$ and $\bm{w}_{2,t}=\bm{w}_t-\bm{w}_{1,t}=\bm{w}_t I\left(\left\|\bm{w}_t\right\| > \lambda_\nu\right)$.

Similar to the proof of Theorem 2.22 in \cite{hall2014martingale},

  \begin{equation*}
    \begin{split}
        E\| \bm{I}_{T,3}\|\leq\ & E\Bigg\|\frac{1}{Th}\sum_{t=1}^{T-i}K\left(\frac{\tau_t -\tau}{h} \right) K\left(\frac{\tau_{t+i} -\tau}{h} \right)\\
        &\cdot \sum_{j=0}^{\infty} \left(\bm{B}_{j+i,t+i}\otimes \bm{B}_{j,t}\right) \left(\bm{w}_{1,t-j}-E\left(\bm{w}_{1,t-j}|\mathcal{F}_{t-j-1}\right) \right)\Bigg\|\\
        &+E\Bigg\|\frac{1}{Th}\sum_{t=1}^{T-i}K\left(\frac{\tau_t -\tau}{h} \right) K\left(\frac{\tau_{t+i} -\tau}{h} \right)\\
        &\cdot\sum_{j=0}^{\infty} \left(\bm{B}_{j+i,t+i}\otimes \bm{B}_{j,t}\right) \left(\bm{w}_{2,t-j}-E\left(\bm{w}_{2,t-j}|\mathcal{F}_{t-j-1}\right) \right)\Bigg\|\\
    :=\ & I_{T,31}+I_{T,32}.\\
    \end{split}
  \end{equation*}
  For $ I_{T,31}$,
  \begin{eqnarray*}
   I_{T,31}&\leq  &\frac{1}{Th}\sum_{j=0}^{\infty}\Bigg\{E\Bigg\|\sum_{t=1}^{T-i} K\left(\frac{\tau_t -\tau}{h} \right) K\left(\frac{\tau_{t+i} -\tau}{h} \right) \\
       &&\cdot \left(\bm{B}_{j+i,t+i}\otimes \bm{B}_{j,t}\right) \left(\bm{w}_{1,t-j}-E\left(\bm{w}_{1,t-j}|\mathcal{F}_{t-j-1}\right) \right)\Bigg\|^2\Bigg\}^{\frac{1}{2}} \\
        &\leq  &\frac{1}{Th}\sum_{j=0}^{\infty}\Bigg\{\sum_{t=1}^{T-i}K^2\left(\frac{\tau_t -\tau}{h} \right) K^2\left(\frac{\tau_{t+i} -\tau}{h} \right)\\
        &&\cdot\left\|\bm{B}_{j+i,t+i}\right\|^2 \left\|\bm{B}_{j,t}\right\|^2 E\left\|\bm{w}_{1,t-j}\right\|^2\Bigg\}^{\frac{1}{2}} \\
        &\leq &\frac{1}{\sqrt{Th}}\max_{t\ge 1}\left(\sum_{j=0}^{\infty} \left\|\bm{B}_{j,t}\right\| \left\|\bm{B}_{j+i,t+i}\right\| \right)\\
        &&\cdot\left\{\frac{1}{Th}\sum_{t=1}^{T-i}K^2\left(\frac{\tau_t -\tau}{h} \right) K^2\left(\frac{\tau_{t+i} -\tau}{h} \right) E\left\|\bm{w}_{1,t-j}\right\|^2\right\}^{\frac{1}{2}}\\
        &:= &\beta_i \phi_{T,1}(\tau),
  \end{eqnarray*}
  where $\beta_i=\max_{t\ge 1} \sum_{j=0}^{\infty} \left\|\bm{B}_{j,t}\right\| \left\|\bm{B}_{j+i,t+i}\right\|$ satisfying
  \begin{equation*}
  \sum_{i=1}^{\infty}\beta_i \leq \sum_{i=1}^{\infty}\max_{t\ge 1}\sum_{j=0}^{\infty} \left\|\bm{B}_{j,t}\right\| \left\|\bm{B}_{j+i,t+i}\right\| \leq \left(\sup_{\tau}\sum_{j=0}^{\infty} \left\|\bm{B}_{j}(\tau)\right\|\right)^2< \infty,
  \end{equation*}
  and $\phi_{T,1}(\tau)=O\left(\frac{\lambda_\nu}{\sqrt{Th}}\right)$. For $I_{T,32}$,
  \begin{equation*}
    \begin{split}
        I_{T,32}\leq\ & 2\max_{t\ge 1}\sum_{j=0}^{\infty} \left\|\bm{B}_{j,t}\right\| \left\|\bm{B}_{j+i,t+i}\right\|\frac{1}{Th}\sum_{t=1}^{T-i} K\left(\frac{\tau_t -\tau}{h} \right) K\left(\frac{\tau_{t+i} -\tau}{h} \right) E\left\|\bm{w}_{2,t-j}\right\| \\
         := &\beta_i\phi_{T,2}(\tau),
    \end{split}
  \end{equation*}
  where $\phi_{T,2}(\tau)=O(\nu)$. As we can make $\nu$ arbitrarily small, it follows that
\begin{eqnarray*}
\phi_T=\sup_{\tau\in[0,1]}(\phi_{T,1}(\tau)+\phi_{T,2}(\tau))\rightarrow0
\end{eqnarray*}
  as $T\rightarrow\infty$.

  For $ \bm{I}_{T,4}$, we have
  \begin{equation*}
    \begin{split}
       E\| \bm{I}_{T,4}\|\leq\ & \frac{1}{\sqrt{Th}}\Bigg( E\Bigg\|\frac{1}{\sqrt{Th}}\sum_{t=1}^{T-i}K\left(\frac{\tau_t -\tau}{h} \right) K\left(\frac{\tau_{t+i} -\tau}{h} \right) \\
       & \cdot \sum_{j=0}^{\infty}\sum_{m=0,\neq j+i}^{\infty}\left(\bm{B}_{m,t+i}\otimes \bm{B}_{j,t}\right) \mathrm{vec}\left[\bm{\epsilon}_{t-j}\bm{\epsilon}_{t+i-m}^\top\right]\Bigg\|^2\Bigg)^{\frac{1}{2}} \\
        \leq &O(1)\frac{1}{\sqrt{Th}},
    \end{split}
  \end{equation*}
  since $E\left(\mathrm{vec}\left[\bm{\epsilon}_{t-j}\bm{\epsilon}_{t+i-m}^\top\right] \mathrm{vec}\left[\bm{\epsilon}_{s-j}\bm{\epsilon}_{s+i-m}^\top\right]^\top\right),\ m\neq j+i$ can only be non-zero if $t= s$.

Based on the above development and $\sum_{i=0}^{\infty}a(i/l)=\sum_{i=0}^{l}a(i/l)=O(l)$, we conclude that
\begin{eqnarray*}
 &&E\left\| \frac{1}{Th}\sum_{i=1}^{T-1}\sum_{t=1}^{T-i}K\left(\frac{\tau_t-\tau}{h}\right)K\left(\frac{\tau_{t+i}-\tau}{h}\right)\left(\bm{e}_t\bm{e}_{t+i}^\top-E(\bm{e}_t\bm{e}_{t+i}^\top)\right) a(i/l)\right\|\\
 &\leq &M\phi_T \sum_{i=0}^{\infty}\beta_i+\frac{M}{\sqrt{Th}}\sum_{i=0}^{\infty}a(i/l)\leq \frac{M l}{\sqrt{Th}}=o(1)
\end{eqnarray*}
since $\sum_{i=1}^{\infty}\beta_i < \infty$ and $\lim_{T\rightarrow\infty}\phi_T=0$.

\medskip

We now just need to focus on $\frac{1}{Th}\sum_{i=1}^{T-1}\sum_{t=1}^{T-i}K\left(\frac{\tau_t-\tau}{h}\right)K\left(\frac{\tau_{t+i}-\tau}{h}\right)E(\bm{e}_t\bm{e}_{t+i}^\top)a(i/l)$. Note that

\begin{equation*}
  \begin{split}
       &\frac{1}{Th}\sum_{i=1}^{T-1}\sum_{t=1}^{T-i}K\left(\frac{\tau_t-\tau}{h}\right)K\left(\frac{\tau_{t+i}-\tau}{h}\right)E(\bm{e}_t\bm{e}_{t+i}^\top)a(i/l)\\
    =\ &\frac{1}{Th}\sum_{i=1}^{T-1}\sum_{t=1}^{T-i}K\left(\frac{\tau_t-\tau}{h}\right)K\left(\frac{\tau_{t+i}-\tau}{h}\right)E(\bm{e}_t\bm{e}_{t+i}^\top)\\
    &+\frac{1}{Th}\sum_{i=1}^{T-1}\sum_{t=1}^{T-i}K\left(\frac{\tau_t-\tau}{h}\right)K\left(\frac{\tau_{t+i}-\tau}{h}\right)E(\bm{e}_t\bm{e}_{t+i}^\top)(a(i/l)-1).
  \end{split}
\end{equation*}
It is then sufficient to show that the second term is $o(1)$ since
\begin{equation*}
\begin{split}
  &\frac{1}{Th}\sum_{t=1}^{T}\sum_{s=1}^{T}K\left(\frac{\tau_t-\tau}{h}\right)K\left(\frac{\tau_{s}-\tau}{h}\right)E(\bm{e}_t\bm{e}_{s}^\top)\\
  =\ & \mathrm{Var}\left(\frac{1}{\sqrt{Th}}\sum_{t=1}^{T}K\left(\frac{\tau_t-\tau}{h}\right)\bm{e}_t \right)
 \to \widetilde{v}_0 \left\{ \sum_{j=0}^{\infty}\bm{B}_j(\tau)\right\}\left\{\sum_{j=0}^{\infty}\bm{B}_j^\top(\tau)\right\}
\end{split}
\end{equation*}
by the proof of Theorem \ref{Theorem2.2}.

Let $d_T$ satisfy $\frac{1}{d_T}+\frac{d_T^2}{l}\rightarrow 0$. The second term is then bounded by
\begin{equation*}
  \begin{split}
       &\left\|\frac{1}{Th}\sum_{i=1}^{T-1}\sum_{t=1}^{T-i}K\left(\frac{\tau_t-\tau}{h}\right)K\left(\frac{\tau_{t+i}-\tau}{h}\right)E(\bm{e}_t\bm{e}_{t+i}^\top)(a(i/l)-1)\right\| \\
      \leq &M \sum_{i=1}^{d_T}\max_t \left\|E(\bm{e}_t\bm{e}_{t+i}^\top)\right\| |a(i/l)-1|+M\sum_{i=d_T+1}^{\infty}\max_t \left\|E(\bm{e}_t\bm{e}_{t+i}^\top)\right\| |a(i/l)-1|\\
      \leq &M \sum_{i=1}^{d_T}(1-a(i/l))+M\sum_{i=d_T+1}^{\infty}\max_t \left\|E(\bm{e}_t\bm{e}_{t+i}^\top)\right\|=o(1),
  \end{split}
\end{equation*}
since

\begin{eqnarray*}
&&\sum_{i=1}^{d_T}(1-a(i/l))\leq \sum_{i=1}^{d_T}(-a^{(1)}(0) i/l+o(i/l))\leq Md_T^2/l=o(1),\\
&&\sum_{i=d_T+1}^{\infty}\max_t \left\|E(\bm{e}_t\bm{e}_{t+i}^\top)\right\|=o(1)\text{ as }d_T\rightarrow \infty.
\end{eqnarray*}

Conditional on the original sample, we now use standard arguments of a blocking technique to show the asymptotic normality of the stochastic term. Now let $Z_{T}^*(\tau)=\sum_{j=1}^{k}X_{T,j}^*(\tau)+\sum_{j=1}^{k}Y_{T,j}^*(\tau)$, where
$$
X_{T,j}^*(\tau)=\frac{1}{\sqrt{Th}}\sum_{t=B_j+1}^{B_j+r_1}K\left(\frac{\tau_t-\tau}{h}\right)\bm{d}^\top \bm{e}_t\xi_{t}^*,\quad Y_{T,j}^*(\tau)=\frac{1}{\sqrt{Th}}\sum_{t=B_j+r_1+1}^{B_j+r_1+r_2}K\left(\frac{\tau_t-\tau}{h}\right)\bm{d}^\top \bm{e}_t\xi_{t}^*,
$$
with $B_j=(j-1)(r_1+r_2)$ and $k=\lceil T/(r_1+r_2) \rceil$. Let $r_1=r_1(T)$ and $r_2=r_2(T)$ satisfying $r_1/(Th)+l/(r_1)\rightarrow 0$ and $r_2/r_1+l/r_2 \rightarrow 0$. We first show that $\sum_{j=1}^{k}Y_{T,j}^*(\tau)=o_P(1)$. Since $r_1>l$ for large enough $T$ and the blocks $Y_{T,j}^*$ are mutual independent conditionally on the original data, then we have
\begin{equation*}
\begin{split}
&EE^*\left(\sum_{j=1}^{k}Y_{T,j}^*(\tau)\right)^2= E\left(\sum_{j=1}^{k}E^*(Y_{T,j}^*(\tau))^2
  \right)\\
    \leq\ &\frac{1}{Th}\sum_{i=-r_2+1}^{r_2-1}a(i/l)\max_t\left\|E(\bm{e}_t\bm{e}_{t+i}^\top)\right\|\sum_{j=1}^{k}\sum_{t=B_j+r_1+1}^{B_j+r_1+r_2-|i|}K\left(\frac{\tau_t-\tau}{h}\right)K\left(\frac{\tau_{t+|i|}-\tau}{h}\right)\\
 \leq\ &M \frac{1}{Th} \max_{0\leq i\leq r_2-1} \sum_{j=1}^{k}\sum_{t=B_j+r_1+1}^{B_j+r_1+r_2-i}K\left(\frac{\tau_t-\tau}{h}\right)K\left(\frac{\tau_{t+i}-\tau}{h}\right)\\
 \leq\ &M\frac{kr_2h}{Th}\leq M \frac{r_2}{r_1}=o(1).
\end{split}
\end{equation*}
We employ the Lindeberg CLT to establish the asymptotic normality of $\sum_{j=1}^{k}X_{T,j}^*(\tau)$ as the blocks $X_{T,j}^*(\tau)$ are independent when $r_2>l$ for large enough $T$. As discussed before, we have already shown that the asymptotic variance is equal to $\widetilde{v}_0 \bm{d}^\top\left\{ \sum_{j=0}^{\infty}\bm{B}_j(\tau)\right\}\left\{\sum_{j=0}^{\infty}\bm{B}_j^\top(\tau)\right\}\bm{d}$. We then need to verify that for every $\nu>0$,
$$
\sum_{j=1}^{k}E^*\left(\frac{X_{T,j}^*(\tau)^2}{E^*\left(\sum_{j=1}^{k}X_{T,j}^*(\tau)\right)^2}I\left(\frac{X_{T,j}^*(\tau)^2}{E^*\left(\sum_{j=1}^{k}X_{T,j}^*(\tau)\right)^2}>\nu \right) \right)=o_P(1).
$$
Conditional on original sample, by H\"older's inequality, Chebyshev's inequality and Minkowski's inequality, we have
\begin{eqnarray*}
 &&\sum_{j=1}^{k}E^*\left(\frac{X_{T,j}^*(\tau)^2}{E^*\left(\sum_{j=1}^{k}X_{T,j}^*(\tau)\right)^2}I\left(\frac{X_{T,j}^*(\tau)^2}{E^*\left(\sum_{j=1}^{k}X_{T,j}^*(\tau)\right)^2}>\nu \right) \right)\\
    & \leq & \sum_{j=1}^{k}\left(E^*\left(\frac{X_{T,j}^*(\tau)^2}{E^*(\sum_{j=1}^{k}
     X_{T,j}^*(\tau))^2} \right)^\frac{\delta}{2}\right)^{\frac{2}{\delta}} \left(\frac{E^*\left(\frac{X_{T,j}^*(\tau)^2}{E^*(\sum_{j=1}^{k}
     X_{T,j}^*(\tau))^2} \right)^\frac{\delta}{2}}{\nu^{\frac{\delta}{2}}} \right)^{\frac{\delta-2}{\delta}} \\
     & = & \nu^{\frac{2-\delta}{2}} \sum_{j=1}^{k} \frac{E^*(X_{T,j}^*(\tau))^\delta}{\left(E^*(\sum_{j=1}^{k}
     X_{T,j}^*(\tau))^2\right)^{\frac{\delta}{2}} }\le  \nu^{\frac{2-\delta}{2}} \sum_{j=1}^{k} \frac{\sum_{t=B_j+1}^{B_j+r_1}\left(\frac{1}{\sqrt{Th}}K\left(\frac{\tau_t-\tau}{h}\right)\bm{d}^\top \bm{e}_t\right)^\delta E^*\left(\xi_{t}^* \right)^\delta }{\left(E^*(\sum_{j=1}^{k}
     X_{T,j}^*(\tau))^2\right)^{\frac{\delta}{2}} }\\
&\leq &\nu^{\frac{2-\delta}{2}}\frac{1}{(Th)^{\frac{\delta-2}{\delta}}} \frac{\frac{1}{Th}\sum_{t=1}^{T}\left(K\left(\frac{\tau_t-\tau}{h}\right)\bm{d}^\top \bm{e}_t\right)^\delta E^*\left( \xi_{t}^*\right)^\delta }{\left(E^*(\sum_{j=1}^{k}
     X_{T,j}^*(\tau))^2\right)^{\frac{\delta}{2}}}=O_P\left(\frac{1}{(Th)^{\frac{\delta-2}{\delta}}}\right)=o_P(1).
\end{eqnarray*}
The proof is now completed.
\end{proof}

\section{Proofs of the Preliminary Lemmas}\label{AppB.2}

\begin{proof}[Proof of Lemma \ref{LemmaB.3}]
\item

\noindent (1). The first result follows from the standard BN decomposition (e.g., \citealp{phillips1992asymptotics}), so the details are omitted.

\medskip

\noindent (2). For the second decomposition, write

\begin{eqnarray*}
(1-L)\widetilde{\mathbb{B}}_t^r(L)&=&\sum_{j=0}^{\infty}\left(L^j\sum_{k=j+1}^{\infty}(\bm{B}_{k+r,t}\otimes \bm{B}_{k,t})-L^{j+1}\sum_{k=j+1}^{\infty}(\bm{B}_{k+r,t}\otimes \bm{B}_{k,t}) \right) \\
&=  &\sum_{j=0}^{\infty}\left(L^{j+1}\sum_{k=j+2}^{\infty}(\bm{B}_{k+r,t}\otimes \bm{B}_{k,t})-L^{j+1}\sum_{k=j+1}^{\infty}(\bm{B}_{k+r,t}\otimes \bm{B}_{k,t}) \right) \\
      & &+ \sum_{k=1}^{\infty}(\bm{B}_{k+r,t}\otimes \bm{B}_{k,t})\\
&=  &-\sum_{j=0}^{\infty}L^{j+1}(\bm{B}_{j+1+r,t}\otimes \bm{B}_{j+1,t}) +  \sum_{k=1}^{\infty} (\bm{B}_{k+r,t}\otimes \bm{B}_{k,t})\\
&= &-\sum_{j=0}^{\infty}L^j(\bm{B}_{j+r,t}\otimes \bm{B}_{j,t}) +  \sum_{k=0}^{\infty} (\bm{B}_{k+r,t}\otimes \bm{B}_{k,t}) \\
&= &\bm{\mathbb{B}}_{t}^r(1)-\bm{\mathbb{B}}_{t}^r(L).
\end{eqnarray*}

\medskip

\noindent (3). By Assumption \ref{Ass1},

\begin{eqnarray*}
\max_{t\ge 1}\sum_{j=0}^{\infty}\left\|\bm{\widetilde{B}}_{j,t}\right\|\leq \max_{t\ge 1}\sum_{j=0}^{\infty}\sum_{k=j+1}^{\infty}\left\|\bm{B}_{k,t}\right\|=\max_{t\ge 1}\sum_{j=1}^{\infty}j\left\|\bm{B}_{j,t}\right\|<\infty.
\end{eqnarray*}

\medskip

\noindent (4). By Assumption \ref{Ass1},
\begin{eqnarray*}
\sum_{t=1}^{T-1}\|\widetilde{\mathbb{B}}_{t+1}(1)-\widetilde{\mathbb{B}}_{t}(1)\| &\leq &\sum_{t=1}^{T-1} \sum_{j=0}^{\infty}\sum_{k=j+1}^{\infty}\left\|\bm{B}_{k,t+1}-\bm{B}_{k,t}\right\| \\
&=&\sum_{t=1}^{T-1} \sum_{j=1}^{\infty}j\left\|\bm{B}_{j,t+1}-\bm{B}_{j,t}\right\| < \infty.
\end{eqnarray*}

\medskip

\noindent (5). By Assumption \ref{Ass1},

  \begin{equation*}
    \begin{split}
\max_{t\ge 1} \sum_{j=0}^{\infty}\left\| \bm{\widetilde{B}}_{j,t}^r \right\| & \leq \max_{t\ge 1} \sum_{j=0}^{\infty}\sum_{k=j+1}^{\infty}\left\|\bm{B}_{k+r,t}\otimes \bm{B}_{k,t} \right\| = \max_{t\ge 1} \sum_{j=0}^{\infty}\sum_{k=j+1}^{\infty}\left\| \bm{B}_{k+r,t}\right\|\cdot\left\|\bm{B}_{k,t}\right\|\\
       & =\max_{t\ge 1} \sum_{j=1}^{\infty}j\left\|\bm{B}_{j+r,t}\right\| \cdot \left\|\bm{B}_{j,t}\right\| \leq M \max_{t\ge 1} \sum_{j=1}^{\infty}j\left\|\bm{B}_{j,t}\right\| < \infty.
    \end{split}
  \end{equation*}

\medskip

\noindent (6). Write

\begin{equation*}
  \begin{split}
  \max_{t\ge 1}   \sum_{r=1}^{\infty}\left\| \widetilde{\mathbb{B}}_t^r(1) \right\|&\leq  \max_{t\ge 1} \sum_{r=1}^{\infty}\sum_{j=0}^{\infty}\sum_{k=j+1}^{\infty} \left\|\bm{B}_{k+r,t}\right\| \cdot \left\|\bm{B}_{k,t}\right\| \\
  &= \max_{t\ge 1} \sum_{j=0}^{\infty} \sum_{k=j+1}^{\infty} \left\|  \bm{B}_{k,t}\right\| \cdot \left( \sum_{r=1}^{\infty}\left\| \bm{B}_{k+r,t}\right\|\right)\\
     &\leq  \max_{t\ge 1} \left(\sum_{r=1}^{\infty}\left\|\bm{B}_{r,t}\right\|\right)\cdot\left(\sum_{j=1}^{\infty}j\left\|\bm{B}_{j,t}\right\|\right) < \infty.
  \end{split}
\end{equation*}

\medskip

\noindent (7). Write

\begin{eqnarray*}
&&\sum_{t=1}^{T-1}\sum_{r=0}^{\infty}\sum_{j=0}^{\infty}\left\|\widetilde{\mathbb{B}}_{t+1}^r(1)-\widetilde{\mathbb{B}}_t^r(1)\right\|
        \leq \sum_{t=1}^{T-1}\sum_{r=0}^{\infty}\sum_{j=0}^{\infty} \left\|\sum_{k=j+1}^{\infty}  \bm{B}_{k+r,t+1}\otimes \bm{B}_{k,t+1}-\bm{B}_{k+r,t}\otimes \bm{B}_{k,t}\right\|\\
    & \leq &
      \sum_{t=1}^{T-1} \sum_{r=0}^{\infty}\sum_{j=0}^{\infty}\sum_{k=j+1}^{\infty}\left(\left\| \bm{B}_{k+r,t+1}-\bm{B}_{k+r,t}\right\|\cdot\left\|  \bm{B}_{k,t+1}\right\|+\left\|\bm{B}_{k+r,t}\right\|\cdot \left\| \bm{B}_{k,t+1} -\bm{B}_{k,t} \right\| \right)\\
      &= &\sum_{t=1}^{T-1} \sum_{j=0}^{\infty}\sum_{k=j+1}^{\infty}\left( \left\| \bm{B}_{k,t+1}\right\| \cdot \sum_{r=0}^{\infty}\left\|\bm{B}_{k+r,t+1}-\bm{B}_{k+r,t}\right\| + \left\| \bm{B}_{k,t+1}-\bm{B}_{k,t}\right\| \cdot \sum_{r=0}^{\infty}\left\|\bm{B}_{k+r,t}\right\| \right)\\
      &\leq & \left(\sum_{t=1}^{T-1}\sum_{r=1}^{\infty}\left\| \bm{B}_{r,t+1}-\bm{B}_{r,t}\right\|\right)\cdot\left(\max_{t\ge 1} \sum_{k=1}^{\infty}k\left\| \bm{B}_{k,t}\right\| \right)\\
      &&+\left(\sum_{t=1}^{T-1}\sum_{k=1}^{\infty}k\left\| \bm{B}_{k,t+1}-\bm{B}_{k,t}\right\|\right)\cdot \left(\max_{t\ge 1} \sum_{r=1}^{\infty}\left\| \bm{B}_{r,t}\right\| \right) < \infty.
\end{eqnarray*}

The proof is now completed.
\end{proof}

\medskip

\begin{proof}[Proof of Lemma \ref{LemmaB.4}]

\item

By Lemma \ref{LemmaB.3}, we have

\begin{equation*}
 \bm{x}_t=\bm{\mu}_t+\bm{\mathbb{B}}_t(1)\bm{\epsilon}_t+\widetilde{\mathbb{B}}_t(L)\bm{\epsilon}_{t-1}-\widetilde{\mathbb{B}}_t(L)\bm{\epsilon}_t,
\end{equation*}
which yields that

\begin{equation*}
\begin{split}
\sum_{t=1}^{T}\bm{W}_{T,t}\left(\bm{x}_t-E(\bm{x}_t)\right)=\ & \sum_{t=1}^{T}\bm{W}_{T,t}\bm{\mathbb{B}}_t(1)\bm{\epsilon}_t+\bm{W}_{T,1}\widetilde{\mathbb{B}}_1(L)\bm{\epsilon}_{0}-\bm{W}_{T,T}\widetilde{\mathbb{B}}_T(L)\bm{\epsilon}_{T} \\
     &+\sum_{t=1}^{T-1}\left(\bm{W}_{T,t+1}\widetilde{\mathbb{B}}_{t+1}(L)-\bm{W}_{T,t}\widetilde{\mathbb{B}}_{t}(L)\right)\bm{\epsilon}_t\\
:=\ &\bm{I}_{T,1}+\bm{I}_{T,2}+\bm{I}_{T,3}+\bm{I}_{T,4}.
\end{split}
\end{equation*}

For $\bm{I}_{T,1}$,
\begin{equation*}
\begin{split}
E\left\|\sum_{t=1}^{T}\bm{W}_{T,t}\bm{\mathbb{B}}_t(1)\bm{\epsilon}_t\right\|^2  &= \mathrm{tr}\left(\sum_{t=1}^{T}\bm{W}_{T,t}\bm{\mathbb{B}}_t(1)E(\bm{\epsilon}_t\bm{\epsilon}_t^\top) \bm{\mathbb{B}}_t^\top(1)\bm{W}_{T,t}^\top\right) \\
     &\leq M\sum_{t=1}^{T}\left\|\bm{W}_{T,t}\right\|^2\leq M\max_{t\ge 1}\left\|\bm{W}_{T,t}\right\|\sum_{t=1}^{T}\left\|\bm{W}_{T,t}\right\|=o(1).
\end{split}
\end{equation*}
Hence, $\|\bm{I}_{T,1}\|=o_P(1)$.

Also, $\|\bm{I}_{T,2}\| =o_P(1)$ and $\|\bm{I}_{T,3}\|=o_P(1)$, since $\max_{t\ge 1}\left\|\bm{W}_{T,t}\right\|=o(1)$, $E\|\widetilde{\mathbb{B}}_1(L)\bm{\epsilon}_0\| < \infty$ and $E \|\widetilde{\mathbb{B}}_T(L)\bm{\epsilon}_t\| < \infty$ by Lemma \ref{LemmaB.3}.

For $\bm{I}_{T,4}$,
\begin{eqnarray}\label{EqB.3}
&& \sum_{t=1}^{T-1}\left(\bm{W}_{T,t+1}\widetilde{\mathbb{B}}_{t+1}(L)-\bm{W}_{T,t}\widetilde{\mathbb{B}}_{t}(L)\right)\bm{\epsilon}_t\nonumber  \\
     &=  &\sum_{t=1}^{T-1}\left(\bm{W}_{T,t+1}-\bm{W}_{T,t}\right)\widetilde{\mathbb{B}}_{t+1}(L)\bm{\epsilon}_t+\sum_{t=1}^{T-1}\bm{W}_{T,t}\left(\widetilde{\mathbb{B}}_{t+1}(L)-\widetilde{\mathbb{B}}_{t}(L)\right)\bm{\epsilon}_t.
\end{eqnarray}
Note that for the first term on the right hand side of \eqref{EqB.3}
\begin{eqnarray*}
   && E\left\|\sum_{t=1}^{T-1}\left(\bm{W}_{T,t+1}-\bm{W}_{T,t}\right)\widetilde{\mathbb{B}}_{t+1}(L)\bm{\epsilon}_t\right\| \\
   &\leq& \max_{t\geq 1}E\left\|\widetilde{\mathbb{B}}_{t+1}(L)\bm{\epsilon}_t\right\|\cdot \sum_{t=1}^{T-1}\left\|\bm{W}_{T,t+1}-\bm{W}_{T,t}\right\|=o(1)
\end{eqnarray*}
by Lemma \ref{LemmaB.3} and the conditions on $\bm{W}_{T,t}$. For the second term on the right hand side of \eqref{EqB.3}, write
\begin{equation*}
  \begin{split}
     & E\left\|\sum_{t=1}^{T-1}
     \bm{W}_{T,t}\left(\widetilde{\mathbb{B}}_{t+1}(L)-\widetilde{\mathbb{B}}_{t}(L)\right)\bm{\epsilon}_t\right\|\\
   \leq\ & \max_{t\ge 1}E\left\|\bm{\epsilon}_t\right\|\cdot \max_{t\ge 1} \left\|\bm{W}_{T,t}\right\| \sum_{t=1}^{T-1}\|\widetilde{\mathbb{B}}_{t+1}(1)-\widetilde{\mathbb{B}}_{t}(1) \|\\
   \leq\ & M \max_{t\ge 1} \left\|\bm{W}_{T,t}\right\| \sum_{t=1}^{T-1}\sum_{j=0}^{\infty}\sum_{k=j+1}^{\infty}
     \left\|\bm{B}_{j,t+1}-\bm{B}_{j,t}\right\|\\
   =\ & M \max_{t\ge 1} \left\|\bm{W}_{T,t}\right\| \sum_{t=1}^{T-1}\sum_{j=1}^{\infty}j
    \left\|\bm{B}_{j,t+1}-\bm{B}_{j,t}\right\|=o(1).
  \end{split}
\end{equation*}
Thus, we have proved that $\|\sum_{t=1}^{T}\bm{W}_{T,t}\left(\bm{x}_t-E(\bm{x}_t)\right)\| =o_P(1)$.

\medskip

We now prove $\|\sum_{t=1}^{T}\bm{W}_{T,t}\left(\bm{x}_t\bm{x}_{t+p}^\top-E\left(\bm{x}_t\bm{x}_{t+p}^\top\right)\right)\| =o_P(1)$. Start from $p=0$ and write

\begin{equation*}
\begin{split}
   \bm{x}_t \bm{x}_t^\top=\ & \bm{\mu}_t\bm{\mu}_t^\top+\bm{\mu}_t\sum_{j=0}^{\infty}\bm{\epsilon}_{t-j}^\top \bm{B}_{j,t}^\top+\sum_{j=0}^{\infty}\bm{B}_{j,t}\bm{\epsilon}_{t-j}\bm{\mu}_t^\top+\sum_{j=0}^{\infty}\bm{B}_{j,t}\bm{\epsilon}_{t-j}\bm{\epsilon}_{t-j}^\top \bm{B}_{j,t}^\top \\
     & +\sum_{r=1}^{\infty}\sum_{j=0}^{\infty}\bm{B}_{j,t}\bm{\epsilon}_{t-j}\bm{\epsilon}_{t-j-r}^\top \bm{B}_{j+r,t}^\top+  \sum_{r=1}^{\infty}\sum_{j=0}^{\infty}\bm{B}_{j+r,t}\epsilon_{t-j-r}\bm{\epsilon}_{t-j}^\top \bm{B}_{j,t}^\top,
\end{split}
\end{equation*}
which yields that

\begin{eqnarray*}
  &&\mathrm{vec}\left[\bm{W}_{T,t}\left(\bm{x}_t\bm{x}_t^\top-E\left(\bm{x}_t\bm{x}_t^\top\right)\right)\right]\\
 &= &\left(\bm{I}_d \otimes \bm{W}_{T,t}\right)\sum_{j=0}^{\infty}\left(\bm{B}_{j,t}\otimes\bm{\mu}_t\right)\bm{\epsilon}_{t-j}+\left(\bm{I}_d \otimes \bm{W}_{T,t}\right)\sum_{j=0}^{\infty}\left(\bm{\mu}_t\otimes \bm{B}_{j,t}\right)\bm{\epsilon}_{t-j}\\
   &   &+\left(\bm{I}_d\otimes \bm{W}_{T,t} \right)\sum_{j=0}^{\infty} \left(\bm{B}_{j,t}\otimes \bm{B}_{j,t}\right)\mathrm{vec}[\bm{\epsilon}_{t-j}\bm{\epsilon}_{t-j}^\top-\bm{I}_d]\\
    &  &+\left(\bm{I}_d\otimes \bm{W}_{T,t} \right)\sum_{r=1}^{\infty}\sum_{j=0}^{\infty}(\bm{B}_{j+r,t}\otimes \bm{B}_{j,t})\mathrm{vec}[\bm{\epsilon}_{t-j}\bm{\epsilon}_{t-j-r}^\top]\\
    &  &+\left(\bm{I}_d\otimes \bm{W}_{T,t} \right)\sum_{r=1}^{\infty}\sum_{j=0}^{\infty}(\bm{B}_{j,t}\otimes \bm{B}_{j+r,t})\mathrm{vec}[\bm{\epsilon}_{t-j-r}\bm{\epsilon}_{t-j}^\top].
\end{eqnarray*}
Consequently, we obtain

\begin{equation*}
  \begin{split}
       & \left\|\sum_{t=1}^{T} \bm{W}_{T,t}\left(\bm{x}_t\bm{x}_t^\top-E\left(\bm{x}_t\bm{x}_t^\top\right)\right)\right\| \\
 \leq\ & 2\left\|\sum_{t=1}^{T}\left(\bm{I}_d \otimes \bm{W}_{T,t}\right)\sum_{j=0}^{\infty}\left(\bm{\mu}_t\otimes \bm{B}_{j,t}\right)\bm{\epsilon}_{t-j}\right\|\\
 & +\left\|\sum_{t=1}^{T}\left(\bm{I}_d \otimes \bm{W}_{T,t}\right)\sum_{j=0}^{\infty} \left(\bm{B}_{j,t}\otimes \bm{B}_{j,t}\right)\mathrm{vec}[\bm{\epsilon}_{t-j}\bm{\epsilon}_{t-j}^\top-\bm{I}_d]\right\|\\
       & +2\left\|\sum_{t=1}^{T}\left(\bm{I}_d\otimes \bm{W}_{T,t} \right)\sum_{r=1}^{\infty}\sum_{j=0}^{\infty}(\bm{B}_{j+r,t}\otimes \bm{B}_{j,t})\mathrm{vec}[\bm{\epsilon}_{t-j}\bm{\epsilon}_{t-j-r}^\top]\right\|\\
 :=\ &I_{T,5}+I_{T,6}+I_{T,7}.\\
  \end{split}
\end{equation*}

By the development of $\sum_{t=1}^{T}\bm{W}_{T,t}\left(\bm{x}_t-E(\bm{x}_t)\right)$, it is easy to know that $I_{T,5}$ is $o_P(1)$.

For $I_{T,6}$, by Lemma \ref{LemmaB.3}, write

\begin{equation*}
  \begin{split}
     I_{T,6}\leq\ &\left\|\sum_{t=1}^{T}\left(\bm{I}_d \otimes \bm{W}_{T,t}\right)\mathbb{B}_t^0(1)\left(\mathrm{vec}(\bm{\epsilon}_t\bm{\epsilon}_t^\top)-\mathrm{vec}(\bm{I}_d)\right)\right\| \\
       & +\left\|\left(\bm{I}_d \otimes \bm{W}_{T,1}\right)\widetilde{\mathbb{B}}_1^0(L)\mathrm{vec}(\bm{\epsilon}_0 \bm{\epsilon}_0^\top)\right\|+ \left\|\left(\bm{I}_d \otimes \bm{W}_{T,T}\right)\widetilde{\mathbb{B}}_T^0(L)\mathrm{vec}(\bm{\epsilon}_T\bm{\epsilon}_T^\top)\right\|\\
       & +\left\|\sum_{t=1}^{T-1}\left(\left(\bm{I}_d \otimes \bm{W}_{T,t+1}\right)\widetilde{\mathbb{B}}_{t+1}^0(L)- \left(\bm{I}_d \otimes \bm{W}_{T,t}\right)\widetilde{\mathbb{B}}_{t}^0(L) \right)\mathrm{vec}\left(\bm{\epsilon}_{t}\bm{\epsilon}_{t}^\top\right)\right\| \\
     :=\ &I_{T,61}+I_{T,62}+I_{T,63}+I_{T,64}.\\
  \end{split}
\end{equation*}
Let $\bm{Z}_t=\mathrm{vec}(\bm{\epsilon}_t\bm{\epsilon}_t^\top-\bm{I}_d)$ for notational simplicity. By Assumption \ref{Ass2}, for any $\nu > 0$, there exists $\lambda_{\nu} > 0$ that for all $t$, $ E\left[\left\|\bm{Z}_t\right\| \cdot I\left(\left\|\bm{Z}_t \right\|> \lambda_{\nu} \right)  \right] < \nu$. Then let further $\bm{Z}_{1,t}=\bm{Z}_t\cdot I\left(\left\|\bm{Z}_t\right\| \leq \lambda_\nu \right)$ and $\bm{Z}_{2,t}=\bm{Z}_t-\bm{Z}_{1,t}$. We are now ready to write

\begin{eqnarray}\label{EqB.4}
I_{T,61}&\leq & \left\|\sum_{t=1}^{T}\left(\bm{I}_d \otimes \bm{W}_{T,t}\right){\mathbb{B}}_t^0(1)\left(\bm{Z}_{1,t}-E\left( \bm{Z}_{1,t}|\mathcal{F}_{t-1}\right)\right)\right\|\nonumber \\
    & &+\left\|\sum_{t=1}^{T}\left(\bm{I}_d \otimes \bm{W}_{T,t}\right){\mathbb{B}}_t^0(1)\left(\bm{Z}_{2,t}-E\left( \bm{Z}_{2,t}|\mathcal{F}_{t-1}\right)\right)\right\|.
\end{eqnarray}
For the first term on the right hand side of \eqref{EqB.4}, by Lemma \ref{LemmaB.3}, write

\begin{eqnarray*}
 && E\left\|\sum_{t=1}^{T}\left(\bm{I}_d \otimes \bm{W}_{T,t}\right){\mathbb{B}}_t^0(1)\left(\bm{Z}_{1,t}-E\left(\bm{Z}_{1,t}|\mathcal{F}_{t-1}\right)\right)\right\|^2 \\
  &= &\mathrm{tr}\Bigg\{ \sum_{t=1}^{T}\sum_{s=1}^{T}(\bm{I}_d \otimes \bm{W}_{T,t}){\mathbb{B}}_t^0(1)E\left(\left(\bm{Z}_{1,t}-E\left(\bm{Z}_{1,t}|\mathcal{F}_{t-1}\right)\right)\left(\bm{Z}_{1,s}-E\left(\bm{Z}_{1,s}|\mathcal{F}_{s-1}\right)\right)^\top \right)\\
  &&\cdot{\mathbb{B}}_s^{0,\top}(1)(\bm{I}_d \otimes \bm{W}_{T,s}^\top)\Bigg\} \\
    &  \leq & M\left(\max_{t\ge 1}\sum_{j=0}^{\infty} \left\|\bm{B}_{j,t} \right\|^2\right)^2 \sum_{t=1}^{T} \left\|\bm{W}_{T,t}\right\|^2 E\left\|\bm{Z}_{1,t}\right\|^2 \leq M \lambda_{\nu}^2 \max_{t\ge 1}\left\| \bm{W}_{T,t}\right\|\sum_{t=1}^{T}\left\|\bm{W}_{T,t}\right\|=o(1).
\end{eqnarray*}
For the second term on the right hand side of \eqref{EqB.4}, we have

\begin{equation*}
E\left\|\sum_{t=1}^{T}\left(\bm{I}_d \otimes \bm{W}_{T,t}\right){\mathbb{B}}_t^0(1)\left(\bm{Z}_{2,t}-E\left( \bm{Z}_{2,t}|\mathcal{F}_{t-1}\right)\right)\right\|\leq M \sum_{t=1}^{T}\left\|\bm{W}_{T,t}\right\| E\left\|\bm{Z}_{2,t}\right\|\le M\nu.
\end{equation*}
By choosing $\nu$ sufficiently small, and then it follows that $I_{T,61}=o_P(1)$. Similar to the proof of $\bm{I}_{T,2}$ and $\bm{I}_{T,3}$, we can prove that $I_{T,62}$ and $I_{T,63}$ are $o_P(1)$. For $I_{T,64}$, we have
\begin{equation*}
  \begin{split}
I_{T,64}\leq\ &\left\|\sum_{t=1}^{T-1}\left( \bm{I}_d\otimes (\bm{W}_{T,t+1}-\bm{W}_{T,t})\right)\widetilde{\mathbb{B}}_{t+1}^0(L)\mathrm{vec}\left(\bm{\epsilon}_t\bm{\epsilon}_t^\top\right)  \right\|   \\
         &+ \left\|\sum_{t=1}^{T-1}\left( \bm{I}_d\otimes \bm{W}_{T,t}\right)\left(\widetilde{\mathbb{B}}_{t+1}^0(L)-\widetilde{\mathbb{B}}_{t}^0(L)\right)\mathrm{vec}\left(\bm{\epsilon}_t\bm{\epsilon}_t^\top\right)\right\|. \\
  \end{split}
\end{equation*}
Similar to the proof of $\bm{I}_{T,4}$, by Lemma \ref{LemmaB.3}, we can prove that $I_{T,64}$ is $o_P(1)$. Then we can conclude that $I_{T,6}=o_P(1)$.

For $I_{T,7}$, using Lemma \ref{LemmaB.3}, we have

\begin{equation*}
  \begin{split}
  I_{T,7}\leq\ & \left\|\sum_{t=1}^{T}(\bm{I}_d\otimes \bm{W}_{T,t})\sum_{r=1}^{\infty}\bm{\mathbb{B}}_t^r(1)\mathrm{vec}\left(\bm{\epsilon}_t\bm{\epsilon}_{t-r}^\top\right)\right\|+\left\|(\bm{I}_d\otimes \bm{W}_{T,1})\sum_{r=1}^{\infty}\widetilde{\mathbb{B}}_1^r(L)\mathrm{vec}\left(\bm{\epsilon}_0\bm{\epsilon}_{-r}^\top\right) \right\|\\
       & + \left\|(\bm{I}_d\otimes \bm{W}_{T,T})\sum_{r=1}^{\infty}\widetilde{\mathbb{B}}_T^r(L)\mathrm{vec}\left(\bm{\epsilon}_T\bm{\epsilon}_{T-r}^\top\right) \right\|\\
       & +\left\|\sum_{t=1}^{T-1}\sum_{r=1}^{\infty}\left((\bm{I}_d\otimes \bm{W}_{T,t+1})\widetilde{\mathbb{B}}_{t+1}^r(L)-(\bm{I}_d\otimes \bm{W}_{T,t})\widetilde{\mathbb{B}}_{t}^r(L) \right)\mathrm{vec}\left(\bm{\epsilon}_t\bm{\epsilon}_{t-r}^\top\right) \right\|\\
  :=\ & I_{T,71}+I_{T,72}+I_{T,73}+I_{T,74}.\\
  \end{split}
\end{equation*}
For $I_{T,71}$, by Lemma \ref{LemmaB.3}, we further write

\begin{equation*}
  \begin{split}
  & E\left\|\sum_{t=1}^{T}(\bm{I}_d\otimes \bm{W}_{T,t})\sum_{r=1}^{\infty}\bm{\mathbb{B}}_t^r(1)\mathrm{vec}\left(\bm{\epsilon}_t\bm{\epsilon}_{t-r}^\top\right)\right\|^2 \\
=\  &E \, \mathrm{tr}\left\{\sum_{t=1}^{T}\sum_{s=1}^{T}(\bm{I}_d\otimes \bm{W}_{T,t})\sum_{r,k=1}^{\infty}\bm{\mathbb{B}}_t^r(1)\mathrm{vec}\left(\bm{\epsilon}_t\bm{\epsilon}_{t-r}^\top\right)\mathrm{vec}^\top\left(\bm{\epsilon}_s\bm{\epsilon}_{s-k}^\top\right) \mathbb{B}_s^{k,\top}(1)(\bm{I}_d\otimes \bm{W}_{T,s}^\top) \right\}  \\
\leq\ & M \sum_{t=1}^{T}\left\|\bm{W}_{T,t}\right\|^2 \sum_{r=1}^{\infty}\left\|\bm{\mathbb{B}}_t^r(1)\right\|^2\leq M \left(\max_{t\ge 1}\sum_{r=1}^{\infty}\left\|\bm{\mathbb{B}}_t^r(1)\right\|\right)^2\max_{t\ge 1}\left\|\bm{W}_{T,t}\right\|\sum_{t=1}^{T}\left\| \bm{W}_{T,t}\right\|=o(1).
  \end{split}
\end{equation*}
In addition, similar to the proof of $\bm{I}_{T,2}$ to $\bm{I}_{T,4}$, we can show that $I_{T,72}$ to $I_{T,74}$ are $o_P(1)$.

Combining the above results,  we have proved the case of $p=0$.

\medskip

Similar to the development of $p=0$, we can consider the case with $p\ge 1$ given $p$ is a fixed number. The details are omitted due to similarity. The proof is now complete.
\end{proof}

\medskip

\begin{proof}[Proof of Lemma \ref{LemmaB.5}]

\item

In the following proof, we cover the interval $[a,b]$ by a finite number of subintervals $\{S_l\}$, which are centered at $s_l$ with the length denoted by $\delta_T$. Denote the number of these intervals by $N_T$, then $N_T=O(\delta_T^{-1})$. In addition, let $\delta_T=O(T^{-1}\gamma_T )$  with $\gamma_T=\sqrt{d_T \log T}$.

Write

\begin{equation*}
\begin{split}
  \sup_{\tau\in[a,b]} \left\|\sum_{t=1}^{T}\bm{W}_{T,t}(\tau)\bm{\mathbb{B}}_{t}(1)\bm{\epsilon}_t\right\|
   \leq\ &\max _{1 \leq l \leq N_{T}} \left\|\sum_{t=1}^{T} \bm{W}_{T,t}(s_l)\bm{\mathbb{B}}_{t}(1)\bm{\epsilon}_t\right\| \\
     & + \max_{1\leq l \leq N_T} \sup_{\tau \in S_l} \left\|\sum_{t=1}^{T}\left(\bm{W}_{T,t}(\tau)-\bm{W}_{T,t}(s_l)\right) \bm{\mathbb{B}}_{t}(1)\bm{\epsilon}_t \right\|\\
   :=\ & J_{T,1}+J_{T,2}.
\end{split}
\end{equation*}

For $J_{T,2}$, since $\bm{W}_{T,t}(\cdot)$ is Lipschitz continuous and $\max_{t \geq 1}\left\|\bm{\mathbb{B}}_{t}(1)\right\| < \infty$ by Assumption \ref{Ass1}, we have
\begin{eqnarray*}
  E|J_{T,2}| &\leq& \sum_{t=1}^{T} \max_{1\leq l \leq N_T} \sup_{\tau \in S_l}\left\|\bm{W}_{T,t}(\tau)-\bm{W}_{T,t}(s_l)\right\| E\left\|\bm{\mathbb{B}}_{t}(1)\bm{\epsilon}_t\right\|\\
   &\leq& M T \delta_T \max_{t \geq 1}E\|\mathbb{B}_t(1)\bm{\epsilon}_t\|=O(\gamma_T).
\end{eqnarray*}

For $J_{T,1}$, we apply the truncation method. Define $\bm{\epsilon}_t^\prime=\bm{\epsilon}_t I(\|\bm{\epsilon}_t\|\leq T^{\frac{1}{\delta}})$ and $\bm{\epsilon}_t^{\prime\prime}=\bm{\epsilon}_t-\bm{\epsilon}_t^\prime$, where $\delta$ is defined in Assumption \ref{Ass2}, and $I(\cdot)$ is the indicator function. Write
\begin{eqnarray*}
J_{T,1}&= &\max_{1\leq l \leq N_T} \left\| \sum_{t=1}^{T} \bm{W}_{T,t}(s_l) \bm{\mathbb{B}}_t(1)\left(\bm{\epsilon}_t^\prime+\bm{\epsilon}_t^{\prime\prime}-E\left(\bm{\epsilon}_t^\prime+\bm{\epsilon}_t^{\prime\prime}|\mathcal{F}_{t-1}\right)\right) \right\| \\
     &\leq &\max_{1\leq l \leq N_T}\left\|\sum_{t=1}^{T}  \bm{W}_{T,t}(s_l) \bm{\mathbb{B}}_t(1)\left(\bm{\epsilon}_t^\prime-E(\bm{\epsilon}_t^\prime|\mathcal{F}_{t-1})\right)\right\|+\max_{1\leq l \leq N_T}\left\| \sum_{t=1}^{T} \bm{W}_{T,t}(s_l)\bm{\mathbb{B}}_t(1)\bm{\epsilon}_t^{\prime\prime}\right\|\\
    & &+\max_{1\leq l \leq N_T}\left\| \sum_{t=1}^{T}  \bm{W}_{T,t}(s_l) \bm{\mathbb{B}}_t(1)E(\bm{\epsilon}_t^{\prime\prime}|\mathcal{F}_{t-1})\right\|\\
     &:= &J_{T,11}+J_{T,12}+J_{T,13}.
\end{eqnarray*}

Start from $J_{T,12}$. By H\"older's inequality and Markov's inequality,
\begin{eqnarray*}
E|J_{T,12}| &\leq & O(1)d_T \sum_{t=1}^{T}E\left\| \bm{\epsilon}_t^{\prime\prime}\right\| = O(1)d_T \sum_{t=1}^{T}E \| \bm{\epsilon}_t I(\|\bm{\epsilon}_t\|\ge T^{\frac{1}{\delta}}) \|\\
&\le &O(1)d_T \sum_{t=1}^{T}\left\{ E \| \bm{\epsilon}_t \|^\delta\right\}^{\frac{1}{\delta}}\left\{ E \|I(\|\bm{\epsilon}_t\|\ge T^{\frac{1}{\delta}}) \|\right\}^{\frac{\delta-1}{\delta}}\\
&= &O(1)d_T \sum_{t=1}^{T}\left\{ E \| \bm{\epsilon}_t \|^\delta\right\}^{\frac{1}{\delta}}\left\{ \Pr(\|\bm{\epsilon}_t\|\ge T^{\frac{1}{\delta}}) \right\}^{\frac{\delta-1}{\delta}}\\
&\le &O(1)d_T \sum_{t=1}^{T}\left\{ E \| \bm{\epsilon}_t \|^\delta\right\}^{\frac{1}{\delta}}\left\{\frac{ E \| \bm{\epsilon}_t \|^\delta}{T}\right\}^{\frac{\delta-1}{\delta}}\\
&=&O(T^\frac{1}{\delta}d_T)=o\left( \sqrt{d_T \log T}\right),
\end{eqnarray*}
where the second inequality follows from H\"older's inequality, and the third inequality follows from Markov's inequality. Similarly, $ J_{T,13} =O_P(T^\frac{1}{\delta}d_T)=o_P\left( \sqrt{d_T \log T}\right)$.

We now turn to $J_{T,11}$. For notational simplicity, let $\bm{Y}_t= \bm{W}_{T,t}(s_l) \bm{\mathbb{B}}_t(1)\left(\bm{\epsilon}_t^\prime-E(\bm{\epsilon}_t^\prime|\mathcal{F}_{t-1})\right)$ for $1\le t\le T$ and $A_T=2T^{\frac{1}{\delta}}d_T\max_{t\ge 1}\left\|\bm{\mathbb{B}}_t(1)\right\| $. Simple algebra shows that $\left\|\bm{Y}_t\right\|\leq A_T$ uniformly in $t$ and $s_l$. By Assumption \ref{Ass2} and the first condition in the body of this lemma,
\begin{equation*}
\max_{1\leq l\leq N_T}\sum_{t=1}^{T}E\left(\left\|\bm{Y}_t\right\|^2|\mathcal{F}_{t-1}\right) \leq M d_T \max_{1\leq l\leq N_T}\sum_{t=1}^{T}\left\|\bm{W}_{T,t}(s_l)\right\| E\left(\left\|\bm{\epsilon}_t\right\|^2|\mathcal{F}_{t-1}\right)
     =O_{a.s.}(d_T).
\end{equation*}

By Lemma \ref{LemmaB.2} and $T^{\frac{2}{\delta}}d_T \log T \to 0$, choose some $\beta>0$ (such as $\beta=4$), and write
\begin{eqnarray*}
&&\Pr\left(J_{T,11} > \sqrt{\beta M} \gamma_T \right)\\
&= &\Pr\left( J_{T,11} > \sqrt{\beta M} \gamma_T , \ \max_{1\leq l \leq N_T}\left\|\sum_{t=1}^{T}E (\bm{Y}_t \bm{Y}_t^\top |\mathcal{F}_{t-1}) \right\|\leq M d_T\right)\\
&&+\Pr\left( J_{T,11} > \sqrt{\beta M} \gamma_T , \ \max_{1\leq l \leq N_T}\left\|\sum_{t=1}^{T}E (\bm{Y}_t \bm{Y}_t^\top |\mathcal{F}_{t-1}) \right\| > M d_T\right)\\
&\leq & \Pr\left( J_{T,11} > \sqrt{\beta M} \gamma_T , \ \max_{1\leq l \leq N_T}\left\|\sum_{t=1}^{T}E (\bm{Y}_t \bm{Y}_t^\top |\mathcal{F}_{t-1}) \right\|\leq M d_T\right)\\
&&+\Pr\left(\max_{1\leq l \leq N_T}\left\|\sum_{t=1}^{T}E (\bm{Y}_t \bm{Y}_t^\top |\mathcal{F}_{t-1}) \right\| > M d_T\right)\\
&\leq & N_T \exp\left(-\frac{\beta M\gamma_T^2 }{2(Md_T+\gamma_T 2 A_T )}\right)+0\\
&\leq &N_T \exp\left(-\frac{\beta}{2} \log T \right)=O(\delta_T^{-1}) T^{-\frac{\beta}{2}}\to 0.
\end{eqnarray*}

Based on the above development, the proof is now complete.
\end{proof}

\medskip

\begin{proof}[Proof of Lemma \ref{LemmaB.6}]

\item

(1). Similar to the proof of Lemma \ref{LemmaB.5}, we use a finite number of subintervals $\{S_l\}$ to cover the interval $[a,b]$, which are centered at $s_l$ with the length $\delta_T$. Denote the number of these intervals by $N_T$ then $N_T=O(\delta_T^{-1})$. In addition, let $\delta_T=O(T^{-1}\gamma_T)$ with $\gamma_T=\sqrt{d_T \log T}$.

\begin{equation*}
\begin{split}
  &\sup _{\tau \in[a,b]} \left\| \sum_{t=1}^{T}(\bm{I}_d\otimes \bm{W}_{T,t}(\tau))\bm{\mathbb{B}}_{t}^0(1)\left(\mathrm{vec}\left(\bm{\epsilon}_t \bm{\epsilon}_t^\top\right)- \mathrm{vec}(\bm{I}_d) \right) \right\| \\
 \leq\ &  \max _{1 \leq l \leq N_{T}} \left\| \sum_{t=1}^{T}(\bm{I}_d\otimes \bm{W}_{T,t}(s_l))\bm{\mathbb{B}}_{t}^0(1)\left(\mathrm{vec}\left(\bm{\epsilon}_t \bm{\epsilon}_t^\top\right)- \mathrm{vec}(\bm{I}_d) \right)\right\| \\
     & + \max_{1\leq l \leq N_T} \sup_{\tau \in S_l} \left\| \sum_{t=1}^{T}\left(\bm{I}_d\otimes (\bm{W}_{T,t}(\tau)-\bm{W}_{T,t}(s_l))\right)\bm{\mathbb{B}}_{t}^0(1)\left(\mathrm{vec}\left(\bm{\epsilon}_t \bm{\epsilon}_t^\top\right)- \mathrm{vec}(\bm{I}_d) \right)  \right\|\\
    :=\ &J_{T,1}+J_{T,2}.
\end{split}
\end{equation*}

Start from $J_{T,2}$. Similar to the proof of Lemma \ref{LemmaB.5}, since

\begin{eqnarray*}
\left\|\bm{\mathbb{B}}_{t}^0(1)\right\| \leq \sum_{j=0}^{\infty}\left\|\bm{B}_{j,t}\right\|^2 \leq\left(\sum_{j=0}^{\infty}\left\|\bm{B}_{j,t}\right\|\right)^2 <\infty
\end{eqnarray*}
by Assumption \ref{Ass1}, we have
\begin{equation*}
  E|J_{T,2}|\leq M T \delta_T \max_{t\ge 1}E\left\|\bm{\mathbb{B}}_{t}^0(1)\left(\mathrm{vec}\left(\bm{\epsilon}_t \bm{\epsilon}_t^\top\right)- \mathrm{vec}(\bm{I}_d) \right)\right\|=O(\gamma_T).
\end{equation*}

We then apply the truncation method. Define $\bm{u}_t=\bm{\mathbb{B}}_{t}^0(1)\left(\mathrm{vec}\left(\bm{\epsilon}_t \bm{\epsilon}_t^\top\right)- \mathrm{vec}\left(\bm{I}_d\right) \right)$, $\bm{u}_t^\prime=\bm{u}_t I(\|\bm{u}_t\|\leq T^{\frac{2}{\delta}} )$ and $\bm{u}_t^{\prime\prime}=\bm{u}_t-\bm{u}_t^\prime$. For $J_{T,1}$, write

\begin{equation*}
\begin{split}
  J_{T,1}=\ &\max_{1\leq l \leq N_T} \left\|\sum_{t=1}^{T}(\bm{I}_d\otimes \bm{W}_{T,t}(s_l))\left(\bm{u}_t^\prime+\bm{u}_t^{\prime\prime}-E(\bm{u}_t^\prime+\bm{u}_t^{\prime\prime}|\mathcal{F}_{t-1}) \right) \right\| \\
   \leq\  &\max_{1\leq l \leq N_T}\left\| \sum_{t=1}^{T}(\bm{I}_d\otimes \bm{W}_{T,t}(s_l)) \left(\bm{u}_t^\prime-E(\bm{u}_t^\prime|\mathcal{F}_{t-1})\right)\right\|+ \max_{1\leq l \leq N_T}\left\|\sum_{t=1}^{T} (\bm{I}_d\otimes \bm{W}_{T,t}(s_l))\bm{u}_t^{\prime\prime}\right\|\\
     &+\max_{1\leq l \leq N_T}\left\|\sum_{t=1}^{T}(\bm{I}_d\otimes \bm{W}_{T,t}(s_l)) E(\bm{u}_t^{\prime\prime}|\mathcal{F}_{t-1}) \right\|\\
    :=\ &J_{T,11}+J_{T,12}+J_{T,13}.\\
\end{split}
\end{equation*}
As in the proof of Lemma \ref{LemmaB.5}, we can show that $J_{T,12}=o_P(\sqrt{d_T\log T})$ and $J_{T,13}=o_P(\sqrt{d_T\log T})$ respectively. We focus on $J_{T,11}$ below.

For any $1\leq l \leq N_T$, let $\bm{Y}_t= (\bm{I}_d\otimes \bm{W}_{T,t}(s_l)) (\bm{u}_{t}^\prime-E(\bm{u}_{t}^\prime|\mathcal{F}_{t-1})) $. We then have $E\left(\bm{Y}_t|\mathcal{F}_{t-1}\right)=0$ and $\left\|\bm{Y}_t\right\|\leq 2 T^{2/\delta}d_T \max_t\left\|\bm{\mathbb{B}}_t^0(1)\right\|$. Since $ \max_{t\ge 1} E\left(\left\|\bm{\epsilon}_t\right\|^4|\mathcal{F}_{t-1}\right)< \infty$ a.s., we can write

\begin{equation*}
\max_{1\leq l \leq N_T}\left\|\sum_{t=1}^{T}E (\bm{Y}_t \bm{Y}_t^\top|\mathcal{F}_{t-1}) \right\|
\leq M d_T \max_{t\ge 1} E\left( \left\|\bm{u}_t\right\|^2|\mathcal{F}_{t-1}\right) \max_{1\leq l \leq N_T} \sum_{t=1}^{T} \left\|\bm{W}_{T,t}(s_l)\right\|=O_{a.s.}\left(d_T\right).
\end{equation*}
Similar to Lemma \ref{LemmaB.2}, choose $\beta>0$ (such as $\beta=4$). In view of the fact that $T^\frac{4}{\delta}d_T\log T \rightarrow 0$, we write

\begin{equation*}
  \begin{split}
\Pr\left(J_{T,11} > \sqrt{\beta M} \gamma_T \right)
  =\ &\Pr\left( J_{T,11} > \sqrt{\beta M} \gamma_T , \max_{1\leq l \leq N_T}\left\|\sum_{t=1}^{T}E (\bm{Y}_t \bm{Y}_t^\top |\mathcal{F}_{t-1}) \right\|\leq M d_T\right)\\
  &+\Pr\left( J_{T,11} > \sqrt{\beta M} \gamma_T , \max_{1\leq l \leq N_T}\left\|\sum_{t=1}^{T}E (\bm{Y}_t \bm{Y}_t^\top |\mathcal{F}_{t-1}) \right\| > M d_T\right)\\
   \leq\ & \Pr\left( J_{T,11} > \sqrt{\beta M} \gamma_T , \max_{1\leq l \leq N_T}\left\|\sum_{t=1}^{T}E (\bm{Y}_t \bm{Y}_t^\top |\mathcal{F}_{t-1}) \right\|\leq M d_T\right)\\
   &+\Pr\left(\max_{1\leq l \leq N_T}\left\|\sum_{t=1}^{T}E (\bm{Y}_t \bm{Y}_t^\top |\mathcal{F}_{t-1}) \right\| > M d_T\right)\\
     \leq\  & N_T \exp\left(-\frac{\beta M\gamma_T^2 }{2(Md_T+M\gamma_T T^\frac{2}{\delta}d_T )}\right)+0\\
     \leq\  &N_T \exp\left(-\frac{\beta}{2} \log T \right)=N_T T^{-\frac{\beta}{2}}=o(1).\\
  \end{split}
\end{equation*}
The first result then follows.

\medskip

\noindent (2). Let $\{S_l\}$ be a finite number of subintervals covering the interval $[a,b]$, which are centered at $s_l$ with the length $\delta_T$. Denote the number of these intervals by $N_T$ then $N_T=O(\delta_T^{-1})$. In addition, let $\delta_T=O(T^{-1}\gamma_T)$ with $\gamma_T=\sqrt{d_T \log T}$. Then
\begin{equation*}
\begin{split}
  \sup _{\tau \in [a,b]} \left\|\sum_{t=1}^{T}\left(\bm{I}_d\otimes \bm{W}_{T,t}(\tau)\right)\bm{\zeta}_{t}\bm{\epsilon}_{t} \right\| \leq\ &\max _{1 \leq l \leq N_{T}} \left\| \sum_{t=1}^{T}\left(\bm{I}_d\otimes \bm{W}_{T,t}(s_l)\right)\bm{\zeta}_{t}\bm{\epsilon}_{t}\right\| \\
     & + \max_{1\leq l \leq N_T} \sup_{\tau \in S_l} \left\|\sum_{t=1}^{T}\left( \bm{I}_d \otimes (\bm{W}_{T,t}(\tau)-\bm{W}_{T,t}(s_l)) \right)\bm{\zeta}_{t}\bm{\epsilon}_{t} \right\|\\
     :=\ &J_{T,3}+J_{T,4}.\\
\end{split}
\end{equation*}

Consider $J_{T,4}$. By the fact that $|\mathrm{tr}(\bm{A})|\leq d \left\|\bm{A}\right\|$ for any $d\times d$ matrix $\bm{A}$ and Assumption \ref{Ass1},
\begin{eqnarray*}
  E\left\|\bm{\zeta}_{t}\bm{\epsilon}_{t}\right\| &=& E\left\|\sum_{r=1}^{\infty}\left(\sum_{s=0}^{\infty}\bm{B}_{s+r,t}\otimes\bm{B}_{s,t}\right)\mathrm{vec}\left(\bm{\epsilon}_t \bm{\epsilon}_{t-r}^\top\right) \right\| \\
   &\leq& \left(E\left\|\sum_{r=1}^{\infty}\left(\sum_{s=0}^{\infty}\bm{B}_{s+r,t}\otimes\bm{B}_{s,t}\right)\mathrm{vec}\left(\bm{\epsilon}_t \bm{\epsilon}_{t-r}^\top\right) \right\|^2\right)^{1/2}\\
   &\leq&\left\{\mathrm{tr}\left[\sum_{r=1}^{\infty}\left(\sum_{s=0}^{\infty}\bm{B}_{s+r,t}\otimes \bm{B}_{s,t} \right)\cdot\left(\bm{I}_d\otimes\bm{I}_d\right)\cdot\left(\sum_{s=0}^{\infty}\bm{B}_{s+r,t}\otimes \bm{B}_{s,t}\right)^\top\right]\right\}^{1/2}\\
   &\leq&M \left(\sum_{r=1}^{\infty} \left\|\sum_{s=0}^{\infty}\bm{B}_{s+r,t}\otimes \bm{B}_{s,t}\right\|^2\right)^{1/2}\\
   &\leq&M \left(\sum_{r=1}^{\infty} \left(\sum_{s=0}^{\infty}\left\|\bm{B}_{s+r,t}\right\|^2 \right)\cdot \left(\sum_{s=0}^{\infty}\left\|\bm{B}_{s,t}\right\|^2\right)\right)^{1/2}\\
   &=&M \left(\left(\sum_{r=1}^{\infty}r\left\|\bm{B}_{r,t}\right\|^2 \right)\cdot \left(\sum_{s=0}^{\infty}\left\|\bm{B}_{s,t}\right\|^2\right)\right)^{1/2}\\
   &\leq&M \left(\left(\sum_{r=1}^{\infty}r\left\|\bm{B}_{r,t}\right\| \right)^2\cdot \left(\sum_{s=0}^{\infty}\left\|\bm{B}_{s,t}\right\|^2\right)\right)^{1/2}<\infty.\\
\end{eqnarray*}
Similarly, we have
\begin{equation*}
  E|J_{T,4}|\leq MT\delta_T \max_{t\ge 1}E\left\|\bm{\zeta}_{t}\bm{\epsilon}_t\right\|=O(\gamma_T).
\end{equation*}

Before investigating $J_{T,3}$, we first show that

\begin{eqnarray}\label{EqB.5}
\max_{1\leq l\leq N_T}\sum_{t=1}^{T}\left\|\bm{W}_{T,t}(s_l)\right\|E\left(\left\|\bm{\zeta}_t\bm{\epsilon}_t\right\|^2 |\mathcal{F}_{t-1} \right)=O_P(1).
\end{eqnarray}
Note that

\begin{equation*}
  \begin{split}
       & \max_{1\leq l\leq N_T}\sum_{t=1}^{T}\left\|\bm{W}_{T,t}(s_l)\right\|E\left(\left\|\bm{\zeta}_t\bm{\epsilon}_t\right\|^2 |\mathcal{F}_{t-1} \right) \\
      \leq\ & \max_{1\leq l\leq N_T}\left|\sum_{t=1}^{T}\left\|\bm{W}_{T,t}(s_l)\right\|\left(E\left(\left\|\bm{\zeta}_t\bm{\epsilon}_t\right\|^2 |\mathcal{F}_{t-1} \right)-E\left\|\bm{\zeta}_t\bm{\epsilon}_t\right\|^2\right) \right|\\
      &+\max_{1\leq l\leq N_T}\sum_{t=1}^{T}\left\|\bm{W}_{T,t}(s_l)\right\|E\left\|\bm{\zeta}_t\bm{\epsilon}_t\right\|^2
  \end{split}
\end{equation*}
and $\max_{1\leq l\leq N_T}\sum_{t=1}^{T}\left\|\bm{W}_{T,t}(s_l)\right\|E\left\|\bm{\zeta}_t\bm{\epsilon}_t\right\|^2=O(1)$. Thus, to prove \eqref{EqB.5}, it is sufficient to show

\begin{eqnarray*}
\max_{1\leq l\leq N_T}\left|\sum_{t=1}^{T}\left\|\bm{W}_{T,t}(s_l)\right\|\left(E\left(\left\|\bm{\zeta}_t\bm{\epsilon}_t\right\|^2 |\mathcal{F}_{t-1} \right)-E\left\|\bm{\zeta}_t\bm{\epsilon}_t\right\|^2\right)\right|=o_P(1).
\end{eqnarray*}
In order to do so, we write

\begin{equation*}
\begin{split}
     &\max_{1\leq l\leq N_T}\left|\sum_{t=1}^{T}\left\|\bm{W}_{T,t}(s_l)\right\|\left(E\left(\left\|\bm{\zeta}_t\bm{\epsilon}_t\right\|^2 |\mathcal{F}_{t-1} \right)-E\left\|\bm{\zeta}_t\bm{\epsilon}_t\right\|^2\right) \right|\\
    =\ &\max_{1\leq l\leq N_T}\left|\sum_{t=1}^{T}\left\|\bm{W}_{T,t}(s_l)\right\|\mathrm{tr}\left(\sum_{r,r^*=1}^{\infty}\bm{\mathbb{B}}_{t}^r(1) (\bm{\epsilon}_{t-r}\bm{\epsilon}_{t-r^*}^\top \otimes \bm{I}_d)\bm{\mathbb{B}}_{t}^{r^*,\top}(1)-\sum_{r=1}^{\infty}\bm{\mathbb{B}}_{t}^r(1) \bm{\mathbb{B}}_{t}^{r,\top}(1)\right)\right|\\
     \leq\ &d^2\cdot \max_{1\leq l\leq N_T}\left\| \sum_{t=1}^{T}\left\|\bm{W}_{T,t}(s_l)\right\| \sum_{r=1}^{\infty}\left(\bm{\mathbb{B}}_t^r(1)\otimes \bm{\mathbb{B}}_t^r(1)\right)\left(\mathrm{vec}\left(\bm{\epsilon}_{t-r}\bm{\epsilon}_{t-r}^\top \otimes \bm{I}_d\right)-\mathrm{vec}(\bm{I}_{d^2})\right) \right\| \\
       & + 2d^2\cdot\max_{1\leq l\leq N_T}\left\|\sum_{t=1}^{T}\left\|\bm{W}_{T,t}(s_l)\right\| \sum_{r=1}^{\infty}\sum_{j=1}^{\infty}\left(\bm{\mathbb{B}}_t^{r+j}(1)\otimes \bm{\mathbb{B}}_t^r(1) \right) \mathrm{vec}\left(\bm{\epsilon}_{t-r}\bm{\epsilon}_{t-r-j}^\top \otimes \bm{I}_d\right) \right\|\\
      :=\ &J_{T,5}+J_{T,6}.
\end{split}
\end{equation*}

Let $\mathbb{F}_{r,t}(L)=\sum_{j=1}^{\infty}\bm{\mathbb{B}}_{t}^{r+j}(1)\otimes \bm{\mathbb{B}}_{t}^j(1) L^j$. Similar to the second result of Lemma \ref{LemmaB.3}, we have
\begin{equation}\label{EqB.6}
  \mathbb{F}_{r,t}(L)=\mathbb{F}_{r,t}(1)-(1-L)\widetilde{\mathbb{F}}_{r,t}(L)
\end{equation}
where $\widetilde{\mathbb{F}}_{r,t}(L)=\sum_{j=1}^{\infty}\widetilde{\mathbb{F}}_{rj,t}L^j$ and $\widetilde{\mathbb{F}}_{rj,t}=\sum_{k=j+1}^{\infty}\bm{\mathbb{B}}_{t}^{r+k}(1)\otimes \bm{\mathbb{B}}_{t}^k(1) $. For notational simplicity, denote
\begin{equation*}
  \begin{split}
     \widetilde{X}_{at}& = \sum_{j=1}^{\infty}\left(\bm{\mathbb{B}}_{t}^{j}(1)\otimes \bm{\mathbb{B}}_{t}^j(1)\right) \mathrm{vec}\left(\bm{\epsilon}_{t-j}\bm{\epsilon}_{t-j}^\top \otimes \bm{I}_d\right),\\
       \widetilde{X}_{bt}& =\sum_{r=1}^{\infty}\sum_{j=1}^{\infty}
\left(\bm{\mathbb{B}}_{t}^{r+j}(1)\otimes \bm{\mathbb{B}}_{t}^j(1)\right)\mathrm{vec}\left(\bm{\epsilon}_{t-j}\bm{\epsilon}_{t-r-j}^\top \otimes \bm{I}_d\right).
  \end{split}
\end{equation*}
Applying \eqref{EqB.6} to $\widetilde{X}_{at}$ and $\widetilde{X}_{bt}$ yields that
\begin{equation*}
\begin{split}
   \widetilde{X}_{at} & =\mathbb{F}_{0,t}(1)\mathrm{vec}\left(\bm{\epsilon}_{t}\bm{\epsilon}_{t}^\top \otimes \bm{I}_d\right)-(1-L)\widetilde{\mathbb{F}}_{0,t}(L)\mathrm{vec}\left(\bm{\epsilon}_{t}\bm{\epsilon}_{t}^\top \otimes \bm{I}_d\right),\\
     \widetilde{X}_{bt}&  =\sum_{r=1}^{\infty}\mathbb{F}_{r,t}(1)\mathrm{vec}\left(\bm{\epsilon}_{t}\bm{\epsilon}_{t-r}^\top \otimes \bm{I}_d\right)-(1-L)\sum_{r=1}^{\infty}\widetilde{\mathbb{F}}_{r,t}(L)\mathrm{vec}\left(\bm{\epsilon}_{t}\bm{\epsilon}_{t-r}^\top \otimes \bm{I}_d\right).
\end{split}
\end{equation*}

For $J_{T,5}$, summing up $\widetilde{X}_{at}$ over $t$ yields
\begin{eqnarray*}
&&\max_{1\leq l\leq N_T}\left\| \sum_{t=1}^{T}\left\|\bm{W}_{T,t}(s_l)\right\|\sum_{r=1}^{\infty}\left(\bm{\mathbb{B}}_{t}^r(1)\otimes \bm{\mathbb{B}}_{t}^r(1)\right)\left(\mathrm{vec}\left(\bm{\epsilon}_{t-r}\bm{\epsilon}_{t-r}^\top \otimes \bm{I}_d\right)-\mathrm{vec} (\bm{I}_{d^2} )\right)\right\| \\
&\leq & \max_{1\leq l\leq N_T}\left\| \sum_{t=1}^{T}\left\|\bm{W}_{T,t}(s_l)\right\|\mathbb{F}_{0,t}(1)\left(\mathrm{vec}\left(\bm{\epsilon}_{t}\bm{\epsilon}_{t}^\top \otimes \bm{I}_d\right)-\mathrm{vec}\left(\bm{I}_{d^2}\right)\right)\right\|\\
& & +\max_{1\leq l\leq N_T}\left\| \left\|\bm{W}_{T,1}(s_l)\right\|\widetilde{\mathbb{F}}_{0,1}(L)\mathrm{vec}\left(\bm{\epsilon}_{0}\bm{\epsilon}_{0}^\top \otimes \bm{I}_d\right) \right\|+\sup_{0\leq \tau \leq 1}\left\|\left\|\bm{W}_{T,T}(s_l)\right\|\widetilde{\mathbb{F}}_{0,T}(L)\mathrm{vec}\left(\bm{\epsilon}_{T}\bm{\epsilon}_{T}^\top \otimes \bm{I}_d\right)\right\|\\
& & +\max_{1\leq l\leq N_T}\left\|\sum_{t=1}^{T-1}\left(\left\|\bm{W}_{T,t+1}(s_l)\right\|\widetilde{\mathbb{F}}_{0,t+1}(L)-\left\|\bm{W}_{T,t}(s_l)\right\|\widetilde{\mathbb{F}}_{0,t}(L)\right)\mathrm{vec}\left(\bm{\epsilon}_{t}\bm{\epsilon}_{t}^\top \otimes \bm{I}_d\right) \right\|\\
&:= &J_{T,51}+J_{T,52}+J_{T,53}+J_{T,54}.
\end{eqnarray*}
Similar to the proof of Lemma \ref{LemmaB.6}.1, we can show that $J_{T,51}=O_P\left(\sqrt{d_T \log T}\right)$, since
\begin{equation*}
  \begin{split}
     \max_{t\ge 1}\left\|\mathbb{F}_{0,t}(1)\right\|&\leq\max_{t\ge 1}\sum_{j=1}^{\infty}\left\|\sum_{k=0}^{\infty}\bm{B}_{k+j,t}\otimes \bm{B}_{k,t}\right\|^2\leq\max_{t\ge 1}\sum_{j=1}^{\infty}\left(\sum_{k=0}^{\infty}\left\|\bm{B}_{k+j,t}\right\|^2 \right) \left(\sum_{k=0}^{\infty}\left\|\bm{B}_{k,t}\right\|^2\right)\\
       & \leq \max_{t\ge 1}\left(\sum_{k=0}^{\infty}\left\|\bm{B}_{k,t}\right\|^2\right) \left(\sum_{j=1}^{\infty}j\left\|\bm{B}_{j,t}\right\|^2\right)< \infty.
  \end{split}
\end{equation*}
Also, we can show that $J_{T,52}=O_P(d_T)$ and $J_{T,53}=O_P(d_T)$, since
\begin{eqnarray*}
\max_{t\ge 1}\left\|\widetilde{\mathbb{F}}_{0,t}(1)\right\|& \leq &\max_{t\ge 1} \sum_{j=1}^{\infty} \sum_{k=j+1}^{\infty}\left\|\bm{\mathbb{B}}_{t}^k(1)\right\|^2
        \leq\max_{t\ge 1} \sum_{r=1}^{\infty} \sum_{k=r+1}^{\infty}\left(\sum_{j=0}^{\infty}\left\|\bm{B}_{j+k,t}\right\|^2\right) \left(\sum_{j=0}^{\infty}\left\|\bm{B}_{j,t}\right\|^2 \right) \nonumber \\
       & \leq &\max_{t\ge 1} \left(\sum_{j=0}^{\infty}\left\|\bm{B}_{j,t}\right\|^2 \right) \left(\sum_{r=1}^{\infty}\sum_{k=r+1}^{\infty}(k-r)\left\|\bm{B}_{k,t}\right\|^2 \right) \nonumber \\
       &\leq &\max_{t\ge 1} \left(\sum_{j=0}^{\infty}\left\|\bm{B}_{j,t}\right\|^2 \right) \left(\sum_{r=1}^{\infty} \frac{r(r+1)}{2}\left\|\bm{B}_{r+1,t}\right\|^2 \right)\nonumber \\
       &\leq &\max_{t\ge 1} \left(\sum_{j=0}^{\infty}\left\|\bm{B}_{j,t}\right\|^2 \right) \left(\sum_{j=1}^{\infty}j^2\left\|\bm{B}_{j,t}\right\|^2 \right)<\infty.
\end{eqnarray*}

We can easily show $J_{T,54}=o_P\left(1\right)$, since
\begin{eqnarray*}
&&\sup_{\tau\in[a,b]}\left(\sum_{t=1}^{T-1}\left\|\bm{W}_{T,t+1}(\tau)\right\|-\left\|\bm{W}_{T,t}(\tau)\right\|\right)\\
&\leq&\sup_{\tau\in[a,b]}\sum_{t=1}^{T-1}\left\|\bm{W}_{T,t+1}(\tau)-\bm{W}_{T,t}(\tau)\right\|=o(1)
\end{eqnarray*}
and
\begin{eqnarray*}
&&\sum_{t=1}^{T-1}\left\|\widetilde{\mathbb{F}}_{0,t+1}(1)-\widetilde{\mathbb{F}}_{0,t}(1)\right\|\leq \sum_{t=1}^{T-1} \sum_{r=1}^{\infty}\sum_{k=r+1}^{\infty}\left\|\bm{\mathbb{B}}_{t+1}^k(1)\otimes \bm{\mathbb{B}}_{t+1}^k(1)-\bm{\mathbb{B}}_{t}^k(1)\otimes \bm{\mathbb{B}}_{t}^k(1)\right\| \\
  &\leq & \sum_{t=1}^{T-1} \sum_{r=1}^{\infty}\sum_{k=r+1}^{\infty}\left\|\bm{\mathbb{B}}_{t+1}^k(1)-\bm{\mathbb{B}}_{t}^k(1)\right\|\cdot\left(\left\|\bm{\mathbb{B}}_{t+1}^k(1)\right\|+\left\|\bm{\mathbb{B}}_{t}^k(1)\right\|\right)\\
  &\leq & M\sum_{t=1}^{T-1} \sum_{r=1}^{\infty}\sum_{k=r+1}^{\infty}\sum_{j=0}^{\infty}\left\|\bm{B}_{j+k,t+1}\otimes \bm{B}_{j,t+1}-\bm{B}_{j+k,t}\otimes \bm{B}_{j,t}\right\|\\
  &\leq & M\sum_{t=1}^{T-1} \sum_{r=1}^{\infty}\sum_{k=r+1}^{\infty}\sum_{j=0}^{\infty}\left(\left\|\bm{B}_{j+k,t+1}-\bm{B}_{j+k,t}\right\|\cdot \left\|\bm{B}_{j,t+1} \right\|+ \left\|\bm{B}_{j,t+1}-\bm{B}_{j,t}\right\|\cdot\left\|\bm{B}_{j+k,t} \right\|\right)\\
  &\leq  & M\sum_{t=1}^{T-1} \sum_{r=1}^{\infty}\sum_{k=r+1}^{\infty}\left(\left\|\bm{B}_{k,t+1}-\bm{B}_{k,t}\right\|\cdot\sum_{j=0}^{\infty} \left\|\bm{B}_{j,t+1} \right\|+\left\|\bm{B}_{k,t} \right\|\cdot\sum_{j=0}^{\infty} \left\|\bm{B}_{j,t+1}-\bm{B}_{j,t}\right\| \right)\\
  &\leq & M \left(\sum_{t=1}^{T-1}\sum_{k=1}^{\infty}k \left\|\bm{B}_{k,t+1}-\bm{B}_{k,t}\right\| \right)\cdot\left(\max_{t\ge 1}\sum_{j=0}^{\infty} \left\|\bm{B}_{j,t+1}\right\|\right)\\
       &&+M \left(\max_{t\ge 1}\sum_{k=1}^{\infty}k\left\|\bm{B}_{k,t}\right\| \right)\cdot\left(\sum_{t=1}^{T-1}\sum_{j=0}^{\infty} \left\|\bm{B}_{j,t+1}-\bm{B}_{j,t}\right\|\right)\\
 & = &O(1).
\end{eqnarray*}
Based on the above development, we conclude that $J_{T,5}=o_P(1)$.

Next, we focus on $J_{T,6}$, and write
\begin{eqnarray*}
  && \max_{1\leq l\leq N_T}\left\|\sum_{t=1}^{T}\left\|\bm{W}_{T,t}(s_l) \right\|\sum_{r=1}^{\infty}\sum_{j=1}^{\infty}\bm{\mathbb{B}}_{t}^{r+j}(1)\otimes \bm{\mathbb{B}}_{t}^r(1)\mathrm{vec}\left(\bm{\epsilon}_{t-r}\bm{\epsilon}_{t-r-j}^\top \otimes \bm{I}_d\right)\right\|\\
  &\leq &\max_{1\leq l\leq N_T}\left\|\sum_{t=1}^{T}\left\|\bm{W}_{T,t}(s_l) \right\|\sum_{r=1}^{\infty}\mathbb{F}_{r,t}(1)\mathrm{vec}\left(\bm{\epsilon}_t \bm{\epsilon}_{t-r}^\top \otimes \bm{I}_d\right)\right\| \\
   &&+\max_{1\leq l\leq N_T}\left\|\left\|\bm{W}_{T,1}(s_l) \right\|\sum_{r=1}^{\infty}\widetilde{\mathbb{F}}_{r,1}(L)\mathrm{vec}\left(\bm{\epsilon}_0\bm{\epsilon}_{-r}^\top\otimes \bm{I}_d\right)\right\|\\
   &&+\max_{1\leq l\leq N_T}\left\|\left\|\bm{W}_{T,T}(s_l) \right\|\sum_{r=1}^{\infty}\widetilde{\mathbb{F}}_{r,T}(L)\mathrm{vec}\left(\bm{\epsilon}_t\bm{\epsilon}_{T-r}^\top\otimes \bm{I}_d\right)\right\| \\
       &&+\max_{1\leq l\leq N_T}\left\|\sum_{t=1}^{T-1}\sum_{r=1}^{\infty}\left(\left\|\bm{W}_{T,t+1}(s_l) \right\|\widetilde{\mathbb{F}}_{r,t+1}(L)-\left\|\bm{W}_{T,t}(s_l) \right\|\widetilde{\mathbb{F}}_{r,t}(L) \right)\mathrm{vec}\left(\bm{\epsilon}_t \bm{\epsilon}_{t-r}^\top \otimes \bm{I}_d\right)\right\|\\
      &:= &J_{T,61}+J_{T,62}+J_{T,63}+J_{T,64}.
\end{eqnarray*}

We can show that $J_{T,62}$ and $J_{T,63}$ are $O_P(d_T)$, since
\begin{eqnarray*}
&&\max_{t\ge 1}\left\|\sum_{r=1}^{\infty}\widetilde{\mathbb{F}}_{r,t}(1)\right\|\leq \max_{t\ge 1}\sum_{r=1}^{\infty}\sum_{j=1}^{\infty}\left\|\widetilde{\mathbb{F}}_{rj,t}\right\|\leq \max_{t\ge 1}\sum_{r=1}^{\infty}\sum_{j=1}^{\infty}\sum_{k=j+1}^{\infty}\left\|\bm{\mathbb{B}}_{t}^{r+k}(1)\right\| \left\|\bm{\mathbb{B}}_{t}^k(1)\right\|\\
&\leq  &\max_{t\ge 1}\left(\sum_{r=1}^{\infty}\left\|\bm{\mathbb{B}}_{t}^r(1)\right\|\right)\left(\sum_{j=1}^{\infty}j\left\|\bm{\mathbb{B}}_{t}^j(1)\right\|\right)\leq M \max_{t\ge 1}\sum_{j=1}^{\infty}j\sum_{k=0}^{\infty}\left\|\bm{B}_{k+j,t} \right\|\left\|\bm{B}_{k,t} \right\| \\
&\leq  & M \max_{t\ge 1}\left(\sum_{j=1}^{\infty}j\left\|\bm{B}_{j,t}\right\|\right) \left(\sum_{k=0}^{\infty}\left\|\bm{B}_{k,t}\right\|\right)<\infty.
\end{eqnarray*}

Similar to $J_{T,54}$, we have $J_{T,64}=o_P\left(1\right)$, since
\begin{eqnarray*}
 && \sum_{t=1}^{T-1} \sum_{r=1}^{\infty}\left\|\widetilde{\mathbb{F}}_{r,t+1}(1)-\widetilde{\mathbb{F}}_{r,t}(1)\right\|\\
    &\leq &\sum_{t=1}^{T-1}\sum_{r=1}^{\infty} \sum_{j=1}^{\infty}\sum_{k=j+1}^{\infty}\left\|\bm{\mathbb{B}}_{t+1}^{r+k}(1)\otimes \bm{\mathbb{B}}_{t+1}^k(1)-\bm{\mathbb{B}}_{t}^{r+k}(1)\otimes \bm{\mathbb{B}}_{t}^k(1)\right\|\\
    &\leq &\sum_{t=1}^{T-1}\sum_{r=1}^{\infty} \sum_{j=1}^{\infty}\sum_{k=j+1}^{\infty}\left(\left\|\bm{\mathbb{B}}_{t+1}^{r+k}(1)-\bm{\mathbb{B}}_{t}^{r+k}(1)\right\| \left\|\bm{\mathbb{B}}_{t+1}^k(1) \right\|+\left\|\bm{\mathbb{B}}_{t+1}^k(1)-\bm{\mathbb{B}}_{t}^k(1)\right\| \left\|\bm{\mathbb{B}}_{t}^{r+k}(1) \right\|\right)\\
    &\leq & \left(\sum_{t=1}^{T-1}\sum_{r=1}^{\infty}\left\|\bm{\mathbb{B}}_{t+1}^r(1)-\bm{\mathbb{B}}_{t}^r(1)\right\| \right)\left(\max_{t\ge 1}\sum_{j=1}^{\infty}j \left\|\bm{\mathbb{B}}_{t}^j\right\| \right)\\
    &&+\left(\sum_{t=1}^{T-1}\sum_{j=1}^{\infty}j \left\|\bm{\mathbb{B}}_{t+1}^j(1)-\bm{\mathbb{B}}_{t}^j(1)\right\| \right)\left(\max_{t\ge 1}\sum_{r=1}^{\infty}\left\|\bm{\mathbb{B}}_{t}^r\right\| \right)\\
    &= &O(1).
\end{eqnarray*}

Now consider term $J_{T,61}$. Define $\bm{u}_t=\sum_{r=1}^{\infty}\mathbb{F}_{r,t}(1)\mathrm{vec}\left(\bm{\epsilon}_t \bm{\epsilon}_{t-r}^\top \otimes \bm{I}_d\right)$, $\bm{u}_t^\prime=\bm{u}_t I(\|\bm{u}_t\|\leq T^{\frac{2}{\delta}} )$ and $\bm{u}_t^{\prime\prime}=\bm{u}_t-\bm{u}_t^\prime$. Then we have
\begin{equation*}
\begin{split}
  J_{T,61}=\ & \max_{1\leq l \leq N_T} \left\| \sum_{t=1}^{T}\left\| \bm{W}_{T,t}(s_l)\right\|\left(\bm{u}_t^\prime+\bm{u}_t^{\prime\prime}-E(\bm{u}_t^\prime+\bm{u}_t^{\prime\prime}|\mathcal{F}_{t-1})\right)  \right\| \\
    \leq\ & \max_{1\leq l \leq N_T}\left\|\sum_{t=1}^{T} \left\| \bm{W}_{T,t}(s_l)\right\|\left(\bm{u}_t^\prime-E(\bm{u}_t^\prime|\mathcal{F}_{t-1})\right) \right\|+ \max_{1\leq l \leq N_T}\left\| \sum_{t=1}^{T} \left\| \bm{W}_{T,t}(s_l)\right\|\bm{u}_t^{\prime\prime} \right\|\\
     &+\max_{1\leq l \leq N_T}\left\| \sum_{t=1}^{T}\left\| \bm{W}_{T,t}(s_l)\right\| E(\bm{u}_t^{\prime\prime}|\mathcal{F}_{t-1}) \right\|\\
    :=\ &J_{T,611}+J_{T,612}+J_{T,613}.\\
\end{split}
\end{equation*}
Using an argument as in the proof for $J_{T,12}$ of Lemma \ref{LemmaB.5}, we can show that $J_{T,612}$ and $J_{T,613}$ are $O_P\left(T^{\frac{2}{\delta}}d_T\right)$.

Next, consider $J_{T,611}$. For any $1\leq l \leq N_T$, let $\bm{Y}_t=\left\| \bm{W}_{T,t}(s_l)\right\|(\bm{u}_{t}^\prime-E(\bm{u}_{t}^\prime|\mathcal{F}_{t-1}))$. We then have $E\left(\bm{Y}_t|\mathcal{F}_{t-1}\right)=0$ and $\left\|\bm{Y}_t\right\|\leq 2 T^{2/\delta}d_T$. In addition, we have
\begin{eqnarray*}
&&\max_{1\leq l \leq N_T}\left\|\sum_{t=1}^{T}E (\bm{Y}_t \bm{Y}_t^\top|\mathcal{F}_{t-1}) \right\|\\
&\leq & 4 \max_{1\leq l \leq N_T}\sum_{t=1}^{T}\left\| \bm{W}_{T,t}(s_l)\right\|^2E\left( \left\|\bm{u}_t\right\|^2|\mathcal{F}_{t-1}\right)\\
&\leq & M\cdot d_T \max_{1\leq l \leq N_T}\sum_{t=1}^{T}\left\| \bm{W}_{T,t}(s_l)\right\|\sum_{r=1}^{\infty}\left\|\mathbb{F}_{r,t}(1)\right\|^2 \left\|\bm{\epsilon}_{t-r}\right\|^2\\
&\leq & M\cdot d_T \max_{t\ge 1}\sum_{r=1}^{\infty}\left\|\mathbb{F}_{r,t}(1)\right\|^2 \left\|\bm{\epsilon}_{t-r}\right\|^2 \\
&\leq & M\cdot d_T\sum_{r=1}^{\infty} \max_{t\ge 1}\left\|\mathbb{F}_{r,t}(1)\right\|^2 \left(\sum_{t=1}^{T}\left\|\bm{\epsilon}_{t-r}\right\|^\delta\right)^{\frac{2}{\delta}}\\
&= &O_P\left(d_T T^\frac{2}{\delta}\right).
\end{eqnarray*}
Therefore, we have $\max_{1\leq l \leq N_T}\left\|\sum_{t=1}^{T}E (\bm{Y}_t \bm{Y}_t^\top |\mathcal{F}_{t-1}) \right\| =O_P(d_T T^\frac{2}{\delta})$. By Lemma \ref{LemmaB.2}, and choosing $\beta=4$, we have
\begin{eqnarray*}
&&\Pr\left(J_{T,611} > \sqrt{\beta M}\sqrt{d_TT^{\frac{2}{\delta}}\log T}  \right) \\
 & = &\Pr\left( J_{T,611} > \sqrt{\beta M} \sqrt{d_TT^{\frac{2}{\delta}}\log T}, \max_{1\leq l \leq N_T}\left\|\sum_{t=1}^{T}E (\bm{Y}_t \bm{Y}_t^\top |\mathcal{F}_{t-1}) \right\|\leq M d_T T^{\frac{2}{\delta}}\right)\\
  &&+\Pr\left( J_{T,611} > \sqrt{\beta M} \sqrt{d_TT^{\frac{2}{\delta}}\log T} , \max_{1\leq l \leq N_T}\left\|\sum_{t=1}^{T}E (\bm{Y}_t \bm{Y}_t^\top |\mathcal{F}_{t-1}) \right\| > M d_T T^{\frac{2}{\delta}}\right)\\
  & \leq & \Pr\left( J_{T,611} > \sqrt{\beta M} \sqrt{d_TT^{\frac{2}{\delta}}\log T} , \max_{1\leq l \leq N_T}\left\|\sum_{t=1}^{T}E (\bm{Y}_t \bm{Y}_t^\top |\mathcal{F}_{t-1}) \right\|\leq M d_T T^{\frac{2}{\delta}}\right)\\
   &&+\Pr\left(\max_{1\leq l \leq N_T}\left\|\sum_{t=1}^{T}E (\bm{Y}_t \bm{Y}_t^\top |\mathcal{F}_{t-1}) \right\| > M d_T T^{\frac{2}{\delta}}\right)\\
   &  \leq & N_T \exp\left(-\frac{\beta Md_TT^{\frac{2}{\delta}}\log T}{2(Md_TT^{\frac{2}{\delta}}+M\sqrt{d_TT^{\frac{2}{\delta}}\log T} T^\frac{2}{\delta}d_T )}\right)+o(1)\\
   &  \leq &N_T \exp\left(-\frac{\beta}{2} \log T \right)=N_T T^{-\frac{\beta}{2}}=o(1)
\end{eqnarray*}
given $d_T T^{\frac{4}{\delta}} \log T\rightarrow 0$. Hence, we have $J_{T,611}=O_P (\{d_T T^{\frac{2}{\delta}}\log T\}^{1/2} ) $. Combining the above results, we have proved that $\sup_{ \tau \in[a,b]}\left|\sum_{t=1}^{T}\left\|\bm{W}_{T,t}(\tau)\right\|E\left(\left\|\bm{\zeta}_{t}\bm{\epsilon}_t\right\|^2 |\mathcal{F}_{t-1} \right)\right|=O_P(1)$.

Finally, we turn to $J_{T,3}$, and apply the truncation method. Let $\bm{u}_t=\bm{\zeta}_t\bm{\epsilon}_{t}$, $\bm{u}_t^\prime=\bm{u}_t I\left(\|\bm{u}_t\|\leq T^{\frac{2}{\delta}} \right)$ and $\bm{u}_t^{\prime\prime}=\bm{u}_t-\bm{u}_t^\prime$. Then we have
\begin{eqnarray*}
 J_{T,3}&=& \max_{1\leq l \leq N_T} \left\|\sum_{t=1}^{T}\left(\bm{I}_d\otimes \bm{W}_{T,t}(s_l) \right)\left(\bm{u}_t^\prime+\bm{u}_t^{\prime\prime}-E(\bm{u}_t^\prime+\bm{u}_t^{\prime\prime}|\mathcal{F}_{t-1})\right) \right\| \\
    &\leq &\max_{1\leq l \leq N_T}\left\|\sum_{t=1}^{T}\left(\bm{I}_d\otimes \bm{W}_{T,t}(s_l) \right)\left(\bm{u}_t^\prime-E(\bm{u}_t^\prime|\mathcal{F}_{t-1})\right)\right\|+ \max_{1\leq l \leq N_T}\left\|\sum_{t=1}^{T} \left(\bm{I}_d\otimes \bm{W}_{T,t}(s_l) \right)\bm{u}_t^{\prime\prime}\right\|\\
     &&+\max_{1\leq l \leq N_T}\left\| \sum_{t=1}^{T}\left(\bm{I}_d\otimes \bm{W}_{T,t}(s_l) \right)E(\bm{u}_t^{\prime\prime}|\mathcal{F}_{t-1}) \right\|\\
    &= &J_{T,31}+J_{T,32}+J_{T,33}.
\end{eqnarray*}
It's easy to show that $J_{T,32}=O_P (T^{\frac{2}{\delta}}d_T )$ and $J_{T,33}=O_P (T^{\frac{2}{\delta}}d_T )$. Thus, we focus on $J_{T,31}$.

For any $1\leq l \leq N_T$, let $\bm{Y}_t=(\bm{I}_d\otimes \bm{W}_{T,t}(s_l))(\bm{u}_{t}^\prime-E(\bm{u}_{t}^\prime|\mathcal{F}_{t-1})) $, then we have $E\left(\bm{Y}_t|\mathcal{F}_{t-1}\right)=0$ and $\left\|\bm{Y}_t\right\|\leq 2 T^{2/\delta}d_T$. Also,
\begin{eqnarray*}
\max_{1\leq l \leq N_T}\sum_{t=1}^{T}\left\|\bm{W}_{T,t}(s_l)\right\|E\left(\left\|\bm{\zeta}_{t}\bm{\epsilon}_t\right\|^2 |\mathcal{F}_{t-1} \right)=O_P(1),
\end{eqnarray*}
which yields that
\begin{equation*}
\max_{1\leq l \leq N_T}\left\|\sum_{t=1}^{T}E (\bm{Y}_t \bm{Y}_t^\top|\mathcal{F}_{t-1}) \right\| \leq M d_T \max_{1\leq l \leq N_T}\sum_{t=1}^{T}\left\|\bm{W}_{T,t}(s_l)\right\|E\left( \left\|\bm{u}_t\right\|^2|\mathcal{F}_{t-1}\right)=O_P\left(d_T \right).
\end{equation*}
Therefore, we have $\max_{1\leq l \leq N_T}\left\|\sum_{t=1}^{T}E (\bm{Y}_t \bm{Y}_t^\top |\mathcal{F}_{t-1}) \right\|= O_P(d_T)$. By Lemma \ref{LemmaB.2} and choosing $\beta=4$, we have
\begin{equation*}
  \begin{split}
\Pr\left(J_{T,31} > \sqrt{\beta M} \gamma_T \right)
  =\ &\Pr\left( J_{T,31} > \sqrt{\beta M} \gamma_T , \max_{1\leq l \leq N_T}\left\|\sum_{t=1}^{T}E (\bm{Y}_t \bm{Y}_t^\top |\mathcal{F}_{t-1}) \right\|\leq M d_T\right)\\
  &+\Pr\left( J_{T,31} > \sqrt{\beta M} \gamma_T , \max_{1\leq l \leq N_T}\left\|\sum_{t=1}^{T}E (\bm{Y}_t \bm{Y}_t^\top |\mathcal{F}_{t-1}) \right\| > M d_T\right)\\
   \leq\ & \Pr\left( J_{T,31} > \sqrt{\beta M} \gamma_T , \max_{1\leq l \leq N_T}\left\|\sum_{t=1}^{T}E (\bm{Y}_t \bm{Y}_t^\top |\mathcal{F}_{t-1}) \right\|\leq M d_T\right)\\
   &+\Pr\left(\max_{1\leq l \leq N_T}\left\|\sum_{t=1}^{T}E (\bm{Y}_t \bm{Y}_t^\top |\mathcal{F}_{t-1}) \right\| > M d_T\right)\\
     \leq\  & N_T \exp\left(-\frac{\beta M\gamma_T^2 }{2(Md_T+M\gamma_T T^\frac{2}{\delta}d_T )}\right)+o(1)\\
     \leq\  &N_T \exp\left(-\frac{\beta}{2} \log T \right)=N_T T^{-\frac{\beta}{2}}=o(1).\\
  \end{split}
\end{equation*}
given $d_T T^{\frac{4}{\delta}}\log T \rightarrow 0$.

We now have completed the proof of the second result.
\end{proof}

\medskip

\begin{proof}[Proof of Lemma \ref{LemmaB.7}]
\item
\noindent (1). For any fixed $\tau\in(0,1)$, let $\bm{W}_{T,t}=\bm{I_d}\frac{1}{T}\left(\frac{\tau_t-\tau}{h}\right)^kK_h\left(\tau_t-\tau\right)$. It is straightforward to verify the conditions imposed on $\bm{W}_{T,t}$, then the first result follows from Lemma \ref{LemmaB.4} immediately.

\medskip

\noindent (2)--(3). Let $\bm{W}_{T,t}(\tau)=\bm{I_d}\frac{1}{T}\left(\frac{\tau_t-\tau}{h}\right)^kK_h\left(\tau_t-\tau\right)$, then the second and third results follow from Theorem \ref{Theorem2.1}.

\end{proof}

\medskip

\begin{proof}[Proof of Lemma \ref{LemmaB.8}]

\item

\noindent (1). This proof is similar to Lemma \ref{LemmaB.6} with $\bm{Z}_{t-1}$ replacing $\bm{\zeta}_{t}$, so omitted here.

\medskip

\noindent (2). Note that
\begin{eqnarray*}
 && \frac{1}{\sqrt{Th}}\sum_{t=1}^{T}\bm{\eta}_t\left(\bm{\eta}_t-\bm{\widehat{\eta}}_t\right)^\top K\left(\frac{\tau_t-\tau}{h}\right) \\
   &   =  & \frac{1}{\sqrt{Th}}\sum_{t=1}^{T}\bm{\eta}_t\bm{z}_{t-1}^\top\left(\bm{\widehat{A}}(\tau_t)-\bm{A}(\tau_t)\right)^\top K\left(\frac{\tau_t-\tau}{h}\right) \\
    &  = & \frac{1}{\sqrt{Th}}\sum_{t=1}^{T}\bm{\eta}_t\bm{z}_{t-1}^\top h \bm{S}_{T,0}^{-1}(\tau_t) \bm{S}_{T,1}(\tau_t) \bm{A}^{(1),\top}(\tau_t)K\left(\frac{\tau_t-\tau}{h}\right)\\
  &    &+\frac{1}{\sqrt{Th}}\sum_{t=1}^{T}\bm{\eta}_t\bm{z}_{t-1}^\top \frac{1}{2} h^2 \bm{S}_{T,0}^{-1}(\tau_t) \bm{S}_{T,2}(\tau_t) \bm{A}^{(2),\top}(\tau_t)K\left(\frac{\tau_t-\tau}{h}\right)\\
    &  &+\frac{1}{\sqrt{Th}}\sum_{t=1}^{T}\bm{\eta}_t\bm{z}_{t-1}^\top \bm{S}_{T,0}^{-1}(\tau_t) \left(\frac{1}{Th}\sum_{s=1}^{T}\bm{z}_{s-1}\bm{z}_{s-1}^\top \bm{M}^\top(\tau_s)K\left(\frac{\tau_s-\tau_t}{h}\right) \right)K\left(\frac{\tau_t-\tau}{h}\right)\\
    &  &+\frac{1}{\sqrt{Th}}\sum_{t=1}^{T}\bm{\eta}_t\bm{z}_{t-1}^\top \bm{S}_{T,0}^{-1}(\tau_t) \left(\frac{1}{Th}\sum_{s=1}^{T}\bm{z}_{s-1}\bm{\eta}_s^\top K\left(\frac{\tau_s-\tau_t}{h}\right)\right)K\left(\frac{\tau_t-\tau}{h}\right)\\
   &   := & J_{T,1}+J_{T,2}+J_{T,3}+J_{T,4}.
\end{eqnarray*}

For $J_{T,1}$ to $J_{T,3}$, using Lemma \ref{LemmaB.7}, we can replace the sample covariance matrix with its converged value with rate $O_P\left(\sqrt{\frac{\log T}{Th}}\right)$ and hence it's easy to show that $J_{T,1}$ to $J_{T,3}$ are $o_P(1)$.

For $J_{T,4}$, for notational simplicity, we ignore $\bm{S}_{T,0}^{-1}(\tau_t)$ and hence,
\begin{equation*}
  \begin{split}
      J_{T,4}=\ &\frac{1}{(Th)^{3/2}}\sum_{t=1}^{T}\bm{\eta}_t \bm{z}_{t-1}^\top \bm{z}_{t-1}\bm{\eta}_t^\top K(0)K\left(\frac{\tau_t-\tau}{h}\right) \\
                &+\frac{1}{(Th)^{3/2}}\sum_{i=1}^{T-1} \sum_{t=1}^{T-i}\bm{\eta}_t \bm{z}_{t-1}^\top \bm{z}_{t+i-1}\bm{\eta}_{t+i}^\top K\left(\frac{i}{Th}\right)K\left(\frac{\tau_t-\tau}{h}\right) \\
                &+\frac{1}{(Th)^{3/2}}\sum_{i=1}^{T-1} \sum_{t=1}^{T-i}\bm{\eta}_{t+i} \bm{z}_{t+i-1}^\top \bm{z}_{t-1}\bm{\eta}_{t}^\top K\left(\frac{i}{Th}\right)K\left(\frac{\tau_t-\tau}{h}\right) \\
      :=\ &J_{T,41}+J_{T,42}+J_{T,43}.\\
  \end{split}
\end{equation*}
It's easy to see $J_{T,41}=O_P\left((Th)^{-1/2}\right)$. For $J_{T,42}$,
\begin{equation*}
  \begin{split}
      J_{T,42}=\ &\frac{1}{(Th)^{3/2}}\sum_{i=1}^{T-1} \sum_{t=1}^{T-i}\bm{\eta}_t E\left(\bm{z}_{t-1}^\top \bm{z}_{t+i-1}\right)\bm{\eta}_{t+i}^\top K\left(\frac{i}{Th}\right)K\left(\frac{\tau_t-\tau}{h}\right) \\
      &+\frac{1}{(Th)^{3/2}}\sum_{i=1}^{T-1} \sum_{t=1}^{T-i}\bm{\eta}_t\left(\bm{z}_{t-1}^\top \bm{z}_{t+i-1}-E\left(\bm{z}_{t-1}^\top \bm{z}_{t+i-1}\right)\right)\bm{\eta}_{t+i}^\top K\left(\frac{i}{Th}\right)K\left(\frac{\tau_t-\tau}{h}\right) \\
  :=\ &J_{T,421}+J_{T,422}.
  \end{split}
\end{equation*}
For $J_{T,421}$,
\begin{equation*}
  \begin{split}
     E\left\|J_{T,421}\right\|^2\leq\ &\frac{1}{(Th)^3}\sum_{i=1}^{T-1}\sum_{t=1}^{T-i}\left\{E\left(\bm{z}_{t-1}^\top \bm{z}_{t+i-1}\right)E\left(\bm{\eta}_t^\top\bm{\eta}_t\right)E\left(\bm{\eta}_{t+i}^\top\bm{\eta}_{t+i}\right)\right\}^2K^2\left(\frac{i}{Th}\right)K^2\left(\frac{\tau_t-\tau}{h}\right)  \\
      =\ & O\left(\frac{1}{Th} \right),
  \end{split}
\end{equation*}
which then yields that $J_{T,421}=O_P\left((Th)^{-1/2}\right)$. For $J_{T,422}$,
\begin{equation*}
  J_{T,422}=\frac{1}{(Th)^{3/2}}\sum_{i=1}^{T-1} \sum_{t=1}^{T-i}\bm{\eta}_t\sum_{m=1}^{p}\left(\bm{x}_{t-m}^\top \bm{x}_{t+i-m}-E\left(\bm{x}_{t-m}^\top \bm{x}_{t+i-m}\right)\right)\bm{\eta}_{t+i}^\top K\left(\frac{i}{Th}\right)K\left(\frac{\tau_t-\tau}{h}\right).
\end{equation*}
For notational simplicity, assume $p=1$ and thus
\begin{eqnarray*}
J_{T,422}&=&\frac{1}{(Th)^{3/2}}\sum_{i=1}^{T-1} \sum_{t=1}^{T-i}\bm{\eta}_t\left(\bm{\mu}_{t-1}^{\top}\sum_{j=0}^{\infty} \bm{\Psi}_{j,t+i-1}\bm{\eta}_{t+i-1-j} \right)\bm{\eta}_{t+i}^\top K\left(\frac{i}{Th}\right)K\left(\frac{\tau_t-\tau}{h}\right)\\
       && +\frac{1}{(Th)^{3/2}}\sum_{i=1}^{T-1} \sum_{t=1}^{T-i}\bm{\eta}_t\left(\sum_{j=0}^{\infty}\bm{\eta}_{t-1-j}^\top \bm{\Psi}_{j,t-1}^{\top} \bm{\mu}_{t+i-1}\right)\bm{\eta}_{t+i}^\top K\left(\frac{i}{Th}\right)K\left(\frac{\tau_t-\tau}{h}\right)\\
       && +\frac{1}{(Th)^{3/2}}\sum_{i=1}^{T-1} \sum_{t=1}^{T-i}\bm{\eta}_t\left(\sum_{j=0}^{\infty}\left(\bm{\eta}_{t-1-j}^\top\otimes \bm{\eta}_{t-1-j}^\top-E\left(\bm{\eta}_{t-1-j}^\top\otimes \bm{\eta}_{t-1-j}^\top\right) \right) \right.\\
       &&\cdot\left.\mathrm{vec}\left(\bm{\Psi}_{j,t-1}^{\top} \bm{\Psi}_{j+i,t+i-1}\right) \right)\bm{\eta}_{t+i}^\top K\left(\frac{i}{Th}\right)K\left(\frac{\tau_t-\tau}{h}\right)\\
       &&+\frac{1}{(Th)^{3/2}}\sum_{i=1}^{T-1} \sum_{t=1}^{T-i}\bm{\eta}_t\left(\sum_{j=0}^{\infty}\sum_{m=0,\neq j+i}^{\infty}\left(\bm{\eta}_{t+i-1-m}^\top\otimes \bm{\eta}_{t-1-j}^\top\right)
       \mathrm{vec}\left(\bm{\Psi}_{j,t-1}^{\top} \bm{\Psi}_{m,t+i-1}\right) \right)\\
       &&\cdot\bm{\eta}_{t+i}^\top K\left(\frac{i}{Th}\right)K\left(\frac{\tau_t-\tau}{h}\right)\\
  &:=&J_{T,4221}+J_{T,4222}+J_{T,4223}+J_{T,4224}.
\end{eqnarray*}
For $J_{T,4221}$,
\begin{eqnarray*}
&&E\left\|J_{T,4221}\right\|^2 \\ &\leq&\frac{1}{(Th)^3}\sum_{i=1}^{T-1}\sum_{t=1}^{T-i}E\left\|\bm{\eta}_t\left(\bm{\mu}_{t-1}^{\top}\sum_{j=0}^{\infty} \bm{\Psi}_{j,t+i-1}\bm{\eta}_{t+i-1-j} \right)\right\|^2 E\left\|\bm{\eta}_{t+i}\right\|^2 K^2\left(\frac{i}{Th}\right)K^2\left(\frac{\tau_t-\tau}{h}\right)\\
&=&O\left(\frac{1}{Th} \right).
\end{eqnarray*}
Similarly, $J_{T,4222}$ and $J_{T,4223}$ are $O_P\left((Th)^{-1/2}\right)$. For $J_{T,4224}$,
\begin{eqnarray*}
J_{T,4224}&=&\frac{1}{(Th)^{3/2}}\sum_{i=1}^{T-1} \sum_{t=1}^{T-i}\bm{\eta}_t\left(\sum_{j=0}^{\infty}\left(\bm{\eta}_{t}^\top\otimes \bm{\eta}_{t-1-j}^\top\right)
       \mathrm{vec}\left(\bm{\Psi}_{j,t-1}^{\top} \bm{\Psi}_{i-1,t+i-1}\right) \right)\bm{\eta}_{t+i}^\top K\left(\frac{i}{Th}\right)K\left(\frac{\tau_t-\tau}{h}\right)\\
       &&+\frac{1}{(Th)^{3/2}}\sum_{i=1}^{T-1} \sum_{t=1}^{T-i}\bm{\eta}_t\left(\sum_{j=0}^{\infty}\sum_{m=0,\neq j+i,\neq i-1}^{\infty}\left(\bm{\eta}_{t+i-1-m}^\top\otimes \bm{\eta}_{t-1-j}^\top\right) \right.\\
       &&\cdot\left.\mathrm{vec}\left(\bm{\Psi}_{j,t-1}^{\top} \bm{\Psi}_{m,t+i-1}\right) \right)\bm{\eta}_{t+i}^\top K\left(\frac{i}{Th}\right)K\left(\frac{\tau_t-\tau}{h}\right)\\
       &:=&J_{T,42241}+J_{T,42242}.
\end{eqnarray*}
Similar to the proof of $J_{T,4221}$, we can show that $J_{T,42242}=O_P\left((Th)^{-1/2}\right)$.

Let $\bm{w}_{t,i}=\sum_{j=0}^{\infty}\left(\bm{\eta}_t\bm{\eta}_t^\top \otimes\bm{\eta}_{t-1-j}^\top \right)\mathrm{vec}\left(\bm{\Psi}_{j,t-1}^{\top} \bm{\Psi}_{i-1,t+i-1}^{\top}\right)$. For $J_{T,42241}$,
\begin{eqnarray*}
E\left\|J_{T,42241}\right\|^2&=&\frac{1}{(Th)^3}\sum_{i_1=1}^{T-1}\sum_{i_2=1}^{T-1}\sum_{t_1=1}^{T-i_1}E\left\|\bm{\eta}_{t_1+i_1}^\top\bm{\eta}_{t_1+i_1}\right\|E\left\|\bm{w}_{t_1+i_1-i_2,i_2}^\top \bm{w}_{t_1,i_1}\right\|\\
&&\cdot K\left(\frac{i_1}{Th} \right)K\left(\frac{i_2}{Th} \right)K\left(\frac{\tau_{t_1}-\tau}{h} \right)K\left( \frac{\tau_{t_1+i_1-i_2}-\tau}{h}\right)\\
&\leq&\frac{M}{(Th)^3}\sum_{t_1=1}^{T}\left(\max_t\sum_{i=1}^{T-1}\left\|\bm{\Psi}_{i,t}\right\|\right)^2\left(\max_t\sum_{j=0}^{\infty}\left\|\bm{\Psi}_{j,t}\right\|\right)^2K\left( \frac{\tau_{t_1}-\tau}{h}\right)\\
&=&O\left((Th)^{-2}\right).
\end{eqnarray*}
Hence, $J_{T,42}=O_P\left((Th)^{-1/2}\right)$. Similar to $J_{T,42}$, $J_{T,43}=O_P\left((Th)^{-1/2}\right)$. The proof is now complete.
\end{proof}

\bigskip

Define $\bm{\Lambda}_{\mathsf{p}}(\tau)=[\bm{a}(\tau),\bm{A}_{\mathsf{p},1}(\tau),...,\bm{A}_{\mathsf{p},\mathsf{p}}(\tau)]$, where $\bm{A}_{\mathsf{p},j}(\tau)=\bm{A}_{j}(\tau)$ for $1\leq j \leq p$ and $\bm{A}_{\mathsf{p},j}(\tau)=0$ for $j > p$. Let $\bm{z}_{\mathsf{p},t-1}=[1,\bm{x}_{t-1}^\top,...,\bm{x}_{t-\mathsf{p}}^\top]^\top$, $\bm{M}_\mathsf{p}(\tau_t)=\bm{\Lambda}_\mathsf{p}(\tau_t)-\bm{\Lambda}_\mathsf{p}(\tau)-\bm{\Lambda}_\mathsf{p}^{(1)}(\tau)(\tau_t-\tau)-\frac{1}{2}h^2\bm{\Lambda}_\mathsf{p}^{(2)}(\tau)(\tau_t-\tau)^2$,
$\bm{\Lambda}_{\overline{\mathsf{p}}}(\tau)=[\bm{A}_{\mathsf{p},\mathsf{p}+1}(\tau),...,\bm{A}_{\mathsf{p},\mathsf{P}}(\tau)]$ and $\bm{z}_{\overline{\mathsf{p}},t-1}=[\bm{x}_{t-\mathsf{p}-1}^\top,...,\bm{x}_{t-\mathsf{P}}^\top]^\top$.

\begin{proof}[Proof of Lemma \ref{LemmaB.9}]
  \item

\noindent (1). Since $\mathsf{p}\geq p$, we have $\widehat{\bm{\eta}}_{\mathsf{p},t}=\bm{\eta}_t+\left(\bm{\Lambda}_\mathsf{p}(\tau_t)-\widehat{\bm{\Lambda}}_\mathsf{p}(\tau_t)\right) \bm{z}_{\mathsf{p},t-1}$ and
\begin{eqnarray*}
  \text{RSS}(\mathsf{p})&=&\frac{1}{T}\sum_{t=1}^{T}\bm{\eta}_t^\top\bm{\eta}_t+\frac{1}{T}\sum_{t=1}^{T}\bm{z}_{\mathsf{p},t-1}^\top\left(\bm{\Lambda}_\mathsf{p}(\tau_t)-\widehat{\bm{\Lambda}}_\mathsf{p}(\tau_t)\right)^\top \left(\bm{\Lambda}_\mathsf{p}(\tau_t)-\widehat{\bm{\Lambda}}_\mathsf{p}(\tau_t)\right) \bm{z}_{\mathsf{p},t-1}\\
  &&-2\frac{1}{T}\sum_{t=1}^{T}\mathrm{tr}\left(\bm{\eta}_t\left(\bm{\eta}_t-\widehat{\bm{\eta}}_{\mathsf{p},t}\right)^\top\right)\\
  &:=& \frac{1}{T}\sum_{t=1}^{T}\bm{\eta}_t^\top\bm{\eta}_t+ I_{T,1}+I_{T,2}.
\end{eqnarray*}
Since $\bm{\eta}_t^\top\bm{\eta}_t$ is m.d.s., we have $\frac{1}{T}\sum_{t=1}^{T}\bm{\eta}_t^\top\bm{\eta}_t=\frac{1}{T}\sum_{t=1}^{T}E(\bm{\eta}_t^\top\bm{\eta}_t)+T^{-1/2}$. By Theorem \ref{Theorem3.1},
\begin{eqnarray*}
  I_{T,1} &\leq& \frac{1}{T}\sum_{t=1}^{\lfloor Th \rfloor}\left\|\bm{z}_{\mathsf{p},t-1}\right\|^2 \cdot \left\|\widehat{\bm{\Lambda}}_\mathsf{p}(\tau_t)-\bm{\Lambda}_\mathsf{p}(\tau_t)\right\|^2 +\frac{1}{T}\sum_{t=\lfloor Th \rfloor+1}^{T-\lfloor Th \rfloor}\left\|\bm{z}_{\mathsf{p},t-1}\right\|^2 \cdot \left\|\widehat{\bm{\Lambda}}_\mathsf{p}(\tau_t)-\bm{\Lambda}_\mathsf{p}(\tau_t)\right\|^2\\
          && +\frac{1}{T}\sum_{t=T-\lfloor Th \rfloor+1}^{T}\left\|\bm{z}_{\mathsf{p},t-1}\right\|^2 \cdot \left\|\widehat{\bm{\Lambda}}_\mathsf{p}(\tau_t)-\bm{\Lambda}_\mathsf{p}(\tau_t)\right\|^2\\
          &\leq&\sup_{0\leq \tau \leq h}\left\|\widehat{\bm{\Lambda}}_\mathsf{p}(\tau_t)-\bm{\Lambda}_\mathsf{p}(\tau)\right\|^2\cdot\frac{1}{T}\sum_{t=1}^{\lfloor Th \rfloor}\left\|\bm{z}_{\mathsf{p},t-1}\right\|^2+\sup_{h\leq \tau \leq 1-h}\left\|\widehat{\bm{\Lambda}}_\mathsf{p}(\tau_t)-\bm{\Lambda}_\mathsf{p}(\tau)\right\|^2\cdot\frac{1}{T}\sum_{t=\lfloor Th \rfloor+1}^{T-\lfloor Th \rfloor}\left\|\bm{z}_{\mathsf{p},t-1}\right\|^2\\
          &&+\sup_{1-h\leq \tau \leq 1}\left\|\widehat{\bm{\Lambda}}_\mathsf{p}(\tau_t)-\bm{\Lambda}_\mathsf{p}(\tau)\right\|^2\cdot\frac{1}{T}\sum_{t=T-\lfloor Th \rfloor+1}^{T}\left\|\bm{z}_{\mathsf{p},t-1}\right\|^2\\
          &=&O_P\left( 2h(h+(\log T/(Th))^{1/2})^2+(1-2h)c_T^2\right)=O_P\left( (h^{3/2}+(\log T/(Th))^{1/2})^2 \right).
\end{eqnarray*}

For $I_{T,2}$,
\begin{eqnarray*}
    && \frac{1}{T}\sum_{t=1}^{T} \bm{\eta}_t\left(\bm{\eta}_t-\widehat{\bm{\eta}}_{\mathsf{p},t}\right)^\top = \frac{1}{T}\sum_{t=1}^{T} \bm{\eta}_t\bm{z}_{\mathsf{p},t-1}^\top\left(\bm{\widehat{\Lambda}}_\mathsf{p}(\tau_t)-\bm{\Lambda}_\mathsf{p}(\tau_t)\right)^\top \\
      &=&h\cdot \frac{1}{T}\sum_{t=1}^{T}\bm{\eta}_t\bm{z}_{\mathsf{p},t-1}^\top \bm{S}_{T,0}^{-1}(\tau_t) \bm{S}_{T,1}(\tau_t) \bm{A}_\mathsf{p}^{(1),\top}(\tau_t)\\
      &&+\frac{1}{2} h^2 \cdot \frac{1}{T}\sum_{t=1}^{T}\bm{\eta}_t\bm{z}_{\mathsf{p},t-1}^\top \bm{S}_{T,0}^{-1}(\tau_t) \bm{S}_{T,2}(\tau_t) \bm{A}_\mathsf{p}^{(2),\top}(\tau_t)\\
      &&+\frac{1}{T}\sum_{t=1}^{T}\bm{\eta}_t\bm{z}_{\mathsf{p},t-1}^\top \bm{S}_{T,0}^{-1}(\tau_t) \left(\frac{1}{Th}\sum_{s=1}^{T}\bm{z}_{\mathsf{p},s-1}\bm{z}_{s-1}^\top \bm{M}_\mathsf{p}^\top(\tau_s)K\left(\frac{\tau_s-\tau_t}{h}\right) \right)\\
      &&+\frac{1}{T}\sum_{t=1}^{T}\bm{\eta}_t\bm{z}_{\mathsf{p},t-1}^\top \bm{S}_{T,0}^{-1}(\tau_t) \left(\frac{1}{Th}\sum_{s=1}^{T}\bm{z}_{s-1}\bm{\eta}_s^\top K\left(\frac{\tau_s-\tau_t}{h}\right)\right)\\
      &:=&I_{T,3}+I_{T,4}+I_{T,5}+I_{T,6}.
\end{eqnarray*}
By the uniform convergence results stated in Theorem \ref{Theorem2.1}, we replace the weighed sample covariance with its converged value plus the rate $O_P\left((\log T/(Th))^{1/2}\right)$. For $I_{T,3}$, by the fact that $\bm{S}_{T,1}(\tau)=O_P\left( (\log T/(Th))^{1/2} \right)$ for $\tau \in [h,1-h]$,
\begin{eqnarray*}
\left\|I_{T,3}\right\|=O_P\left(T^{-\frac{1}{2}}h^{\frac{3}{2}}+h(\log T/(Th))^{1/2}\right).
\end{eqnarray*}
Similarly,

\begin{eqnarray*}
\|I_{T,4}\|+\|I_{T,5}\|=O_P\left(T^{-\frac{1}{2}}h^2+h^2(\log T/(Th))^{1/2}\right).
\end{eqnarray*}

For $I_{T,6}$, let $\bm{\Sigma}(\tau)=\plim_{T\to \infty}\bm{S}_{T,0}(\tau)$, we have
$$
I_{T,6}=\frac{1}{T}\sum_{t=1}^{T} \bm{\eta}_t\bm{z}_{\mathsf{p},t-1}^\top \bm{\Sigma}^{-1}(\tau_t) \left(\frac{1}{Th}\sum_{s=1}^{T}\bm{z}_{s-1}\bm{\eta}_s^\top K\left(\frac{\tau_s-\tau_t}{h}\right)\right)+O_P\left((Th)^{-1/2}(\log T/(Th))^{1/2}\right).
$$
Similar to the proof of $J_{T,4}$ in Lemma \ref{LemmaB.8}, we can show
$$
\frac{1}{T}\sum_{t=1}^{T} \bm{\eta}_t\bm{z}_{\mathsf{p},t-1}^\top \bm{\Sigma}^{-1}(\tau_t) \left(\frac{1}{Th}\sum_{s=1}^{T}\bm{z}_{s-1}\bm{\eta}_s^\top K\left(\frac{\tau_s-\tau_t}{h}\right)\right)=O_P((Th)^{-1}).
$$
Since $(Th)^{-1}+T^{-\frac{1}{2}}h^{\frac{3}{2}}=o\left(c_T \phi_T\right)$, result (1) follows.

\medskip

\noindent (2). For $\mathsf{p} < p$, we have $\widehat{\bm{\Lambda}}_\mathsf{p}(\tau)-\bm{\Lambda}_{\mathsf{p}}(\tau)=\bm{B}_\mathsf{p}(\tau)+o_P(1)$ uniformly over $\tau \in [0,1]$, where $\bm{B}_\mathsf{p}(\tau)$ is a nonrandom bias term. Since $\widehat{\bm{\eta}}_{\mathsf{p},t}=\bm{\eta}_{t}+\left(\bm{\Lambda}_{\mathsf{p}}(\tau_t)-\widehat{\bm{\Lambda}}_\mathsf{p}(\tau_t)\right)\bm{z}_{\mathsf{p},t-1}+\bm{\Lambda}_{\overline{\mathsf{p}}}(\tau_t)\bm{z}_{\overline{\mathsf{p}},t-1}$, by Theorem \ref{Theorem2.1}, we have
$$
  \text{RSS}(\mathsf{p})= \frac{1}{T}\sum_{t=1}^{T}E\left(\bm{\eta}_{t}^\top\bm{\eta}_{t}\right)+\frac{1}{T}\sum_{t=1}^{T}\mathrm{tr}\left( [\bm{B}_\mathsf{p}(\tau_t),\bm{\Lambda}_{\overline{\mathsf{p}}}(\tau_t)]E\left(\bm{z}_{\mathsf{P},t-1}\bm{z}_{\mathsf{P},t-1}^\top\right)[\bm{B}_\mathsf{p}(\tau_t),\bm{\Lambda}_{\overline{\mathsf{p}}}(\tau_t)]^\top \right)+o_P(1).
$$
Since $[\bm{B}_\mathsf{p}(\tau_t),\bm{\Lambda}_{\overline{\mathsf{p}}}(\tau_t)]\neq 0$ and $E\left(\bm{z}_{\mathsf{P},t-1}\bm{z}_{\mathsf{P},t-1}^\top\right)$ is a positive definite matrix, the result follows.
\end{proof}

\end{document}